\documentclass[12pt]{article}
\usepackage{latexsym}
\usepackage{amsthm}
\usepackage{amsfonts,bm,amsbsy,bbm}
\usepackage{epsfig}
\usepackage{graphicx,rotating,lscape}
\usepackage{verbatim}
\usepackage{natbib}
\usepackage{multirow,color}
\usepackage{geometry}
\usepackage{setspace}
\usepackage{subcaption}
\usepackage{float}
\usepackage{dcolumn}
\usepackage{siunitx}
\usepackage{algorithm}
\usepackage{longtable}
\usepackage{makecell}
\usepackage{enumerate}
\usepackage[shortlabels]{enumitem}
\usepackage{booktabs}
\usepackage{authblk}
\usepackage{tikz}
  \usetikzlibrary{positioning,shapes.geometric}
\usepackage{hyperref}
\hypersetup{colorlinks=true,linkcolor=blue,urlcolor=blue,citecolor=blue}

\usepackage{amsmath}
\usepackage{amssymb}
\usepackage{amsfonts}
\usepackage{multirow}
\usepackage{amsthm}
\usepackage{algorithm}
\usepackage{algorithmic}
\usepackage{tcolorbox}
\tcbuselibrary{breakable}

\newtheorem{theorem}{Theorem}
\newtheorem{lemma}{Lemma}

\newtheorem{assumption}{Assumption}
\theoremstyle{definition}

\newtheorem{remark}{Remark}

\usepackage[sort]{cleveref}
\crefname{theorem}{theorem}{theorems}
\crefname{assumption}{assumption}{assumptions}
\crefname{lemma}{lemma}{lemmas}
\crefname{proposition}{proposition}{propositions}
\crefname{corollary}{corollary}{corollaries}

\Crefname{theorem}{Theorem}{Theorems}
\Crefname{assumption}{Assumption}{Assumptions}
\Crefname{lemma}{Lemma}{Lemmas}
\Crefname{proposition}{Proposition}{Propositions}
\Crefname{corollary}{Corollary}{Corollaries}

\crefrangelabelformat{assumption}{#3#1#4--#5#2#6}
\crefrangelabelformat{section}{#3#1#4--#5#2#6}
\crefrangelabelformat{theorem}{#3#1#4--#5#2#6}

\DeclareMathOperator{\var}{var}
\DeclareMathOperator{\pr}{pr}

\newcommand{\trans}{^{\mathrm{\scriptscriptstyle T}}} 
\def\indep{\bot\!\!\!\bot}

\def\wh{\widehat}
\def\wt{\widetilde}

\newcommand{\amin}{\operatornamewithlimits{arg\,min}}

\newcommand{\convd}{\xrightarrow{d}}

\newcommand{\convp}{\xrightarrow{p}}

\def\beqr{\begin{eqnarray}}
\def\eeqr{\end{eqnarray}}
\def\beqrs{\begin{eqnarray*}}
\def\eeqrs{\end{eqnarray*}}

\def\mR{\mathbb{R}}
\def\mN{\mathbb{N}}
\def\mE{\mathbb{E}}

\def\calE{\mathcal{E}}

\def\calL{\mathcal{L}}
\def\calM{\mathcal{M}}

\def\O{\bm{{O}}}

\def\calT{\mathcal{T}}

\def\calV{\mathcal{V}}

\def\calX{\mathcal{X}}

\def\calZ{\mathcal{Z}}

\def\a{{\bm a}}
\def\b{{\bm b}}

\def\d{{\bm d}}
\def\e{{\bm e}}

\def\g{{\bm g}}
\def\m{{\bm m}}
\def\u{{\bm u}}
\def\v{{\bm v}}
\def\w{{\bm w}}
\def\p{{\bm p}}

\def\t{{\bm t}}
\def\x{{\bm x}}

\def\1{{\bm 1}}
\def\0{{\bm 0}}

\def\A{{\bm A}}
\def\B{{\bm B}}

\def\D{{\bm D}}
\def\G{{\bm G}}
\def\I{{\bm I}}

\def\M{{\bm M}}

\def\Q{{\bm Q}}

\def\V{{\bm V}}

\def\bb{\mbox{\boldmath$\beta$}}

\newcommand{\bigbc}[1]{\left\{#1\right\}}

\usepackage{threeparttable}

\def\bb{\mbox{\boldmath$\beta$}}

\def\plugin{\textnormal{plug-in}}
\def\eff{\text{\textit{eff}}}

\makeatletter
\newcommand*{\ind}{%
\mathbin{%
\mathpalette{\@ind}{}%
}%
}
\newcommand*{\nind}{%
\mathbin{
\mathpalette{\@ind}{\not}
}%
}
\newcommand*{\@ind}[2]{%
\sbox0{$#1\perp\m@th$}
\sbox2{$#1=$}
\sbox4{$#1\vcenter{}$}
\rlap{\copy0}
\dimen@=\dimexpr\ht2-\ht4-.2pt\relax
\kern\dimen@
{#2}%
\kern\dimen@
\copy0 
}
\makeatother

\begin{document}

\title{Efficient Estimation of Average Treatment Effects with Unmeasured Confounding and Proxies}

\author[1]{Chunrong Ai}
\author[2]{Jiawei Shan \thanks{E-mail: \texttt{jiawei.shan@wisc.edu}}}
\affil[1]{School of Management and Economics, The Chinese University of Hong Kong, Shenzhen}
\affil[2]{Department of Biostatistics \& Medical Informatics, University of Wisconsin-Madison}

\date{\today}

\maketitle
\thispagestyle{empty}

\begin{abstract}
Proximal causal inference provides a framework for estimating the average treatment effect (ATE) in the presence of unmeasured confounding by leveraging outcome and treatment proxies. Identification in this framework relies on the existence of a so-called bridge function. Standard approaches typically postulate a parametric specification for the bridge function, which is estimated in a first step and then plugged into an ATE estimator. However, this sequential procedure suffers from two potential sources of efficiency loss: (i) the difficulty of efficiently estimating a bridge function defined by an integral equation, and (ii) the failure to account for the correlation between the estimation steps. To overcome these limitations, we propose a novel approach that approximates the integral equation with increasing moment restrictions and jointly estimates the bridge function and the ATE. We show that, under suitable conditions, our estimator is efficient. Additionally, we provide a data-driven procedure for selecting the tuning parameter (i.e., the number of moments). Simulation studies reveal that the proposed method performs well in finite samples, and an application to the right heart catheterization dataset from the SUPPORT study demonstrates its practical value.
\end{abstract}

\newpage

\section{Introduction}
The unconfoundedness assumption is a key condition for identifying treatment effect parameters and establishing the consistency of many popular estimators. This condition may not hold if some confounders are unmeasured and omitted from the empirical analysis.  One approach to address the omitted variables problem is proximal causal inference, which assumes the availability of outcome and treatment confounding proxies \citep{MiaoGengTchetgenTchetgen2018, TchetgenTchetgenYingCuiShiEtAl2024}. Such an approach has seen a wide range of applications \citep{ShiMiaoNelsonTchetgenTchetgen2020, KallusMaoUehara2022,DukesShpitserTchetgenTchetgen2023, YingMiaoShiTchetgenTchetgen2023,EgamiTchetgenTchetgen2024, GhassamiYangShpitserTchetgenTchetgen2024,QiMiaoZhang2024,QiuShiMiaoDobribanEtAl2024,Ying2024}. 

With the outcome and treatment confounding proxies, \cite{MiaoGengTchetgenTchetgen2018} showed that an \emph{outcome bridge function}, analogous to the outcome regression function one would use if all confounders were observed, identifies the average treatment effect. A common approach is to parameterize the bridge function, estimate it using standard methods, and then estimate the average treatment effect (ATE) via a plug-in. For instance, \cite{MiaoShiLiTchetgenTchetgen2024} suggested a recursive generalized method of moments (RGMM) by transforming the integral equation to fixed-dimensional, user-specified moment restrictions. \label{resp:R1-Q1-1} \cite{TchetgenTchetgenYingCuiShiEtAl2024} introduced a proximal g-computation method, which results in a simple proximal two-stage least squares procedure in the special case of linear working models. The proximal g-computation method is typically more efficient than RGMM, benefiting from an additional parametric restriction on the joint distribution of covariates, but is also more prone to bias if the distribution model is misspecified. These methods suffer from efficiency losses because estimating the bridge function may be inefficient, and the sequential procedure may fail to utilize all information \citep{BrownNewey1998,AiChen2012}. As a remedy, \cite{CuiPuShiMiaoEtAl2023} proposed a doubly robust (DR) locally efficient approach that incorporates an additional treatment bridge function. They proved that their proximal DR estimator achieves the semiparametric local efficiency bound if both bridge functions are correctly specified and consistently estimated, even if they are not efficiently estimated. If one of the bridge functions is misspecified, the DR estimator fails to achieve local efficiency, and it could benefit from more efficient estimation of the bridge functions. 
Moreover, the DR estimator may not be the most efficient, even if both bridge functions are consistently estimated, as illustrated in \Cref{plt:comparison}.
 
This paper's primary contribution is to develop a simple, data-driven method for efficiently estimating the average treatment effect. The key step is to transform the conditional moment restriction that defines the bridge function into an expanding set of unconditional moment restrictions via a sieve basis, and to estimate the bridge function and the ATE jointly. \label{resp:R1-Q1-2} As the number of moments increases, these unconditional moment restrictions provide a good approximation for the unknown conditional moment restriction. 
We show that: (i) the proposed estimator for the bridge function achieves the semiparametric efficiency bound established in \cite{CuiPuShiMiaoEtAl2023}, without requiring any modeling of the data distribution beyond the bridge function;  (ii) the proposed estimator for ATE has an asymptotic variance that is never larger than that of the DR locally efficient estimator; and (iii) the proposed estimator outperforms the theoretically optimal plug-in estimator.
Moreover, the proposed method is readily implemented using widely available software for GMM estimation \citep{Hansen1982}. We also propose a data-driven procedure for selecting the number of moments.

The theoretical results of this paper are closely related to established theories in causal inference under the assumption of no unmeasured confounding and in the missing data literature under the missing at random assumption, which demonstrate that regression-based estimators are more efficient than doubly robust estimators when the outcome model is correctly specified \citep{ScharfsteinRotnitzkyRobins1999,BangRobins2005}. We extend this principle to the proximal causal inference.
Our work is also related to several strands of the existing literature. First, it connects to research on instrumental variable estimation and partial identification strategies \citep{AiChen2003,NeweyPowell2003,Abadie2003,KlineTamer2023}. Second, it is related to recent advances that incorporate high-dimensional and machine learning methods into proximal estimation frameworks \citep{MastouriZhuGultchinKorbaEtAl2021,KompaBellamyKolokotronesRobinsEtAl2022}. Third, it is linked to developments in causal graphical models for identifying effects in the presence of unmeasured confounding, including work on the front-door criterion \citep{Pearl1995,RichardsonEvansRobinsShpitser2023,GuoBenkeserNabi2023,BhattacharyaNabiShpitser2022,GuoNabi2024}, which also leverage proxy-like mediators to achieve nonparametric identification under structural assumptions. \label{res:R1-Q2}

The rest of the paper is organized as follows. 
\Cref{sec:setup} reviews the identification results of proximal causal inference and challenges in constructing an efficient estimator.
\Cref{sec:estimation} describes the proposed method for efficient estimation. 
\Cref{sec:asym_prop} establishes large-sample properties of the proposed estimators and compares them with the existing methods.
\Cref{sec:implementation} discusses the choice of tuning parameters for practical implementation.
\Cref{sec:numerical} provides a simulation study and applies our method to reanalyze the SUPPORT dataset.
\Cref{sec:discuss} offers a brief discussion.
Technical proofs and additional results are provided in the supplementary material.

\begin{figure}[h]
    \centering  
    \begin{subfigure}[b]{0.48\textwidth}
        \centering
        \includegraphics[width=\textwidth]{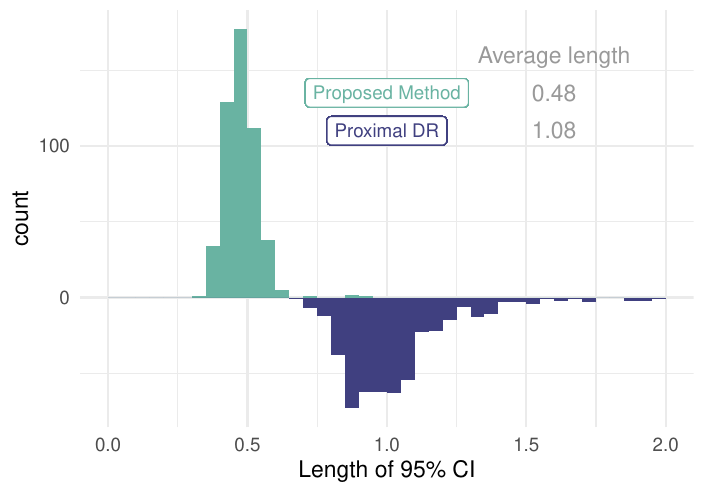}
    \end{subfigure}
    \hfill
    \begin{subfigure}[b]{0.48\textwidth}
        \centering
        \includegraphics[width=\textwidth]{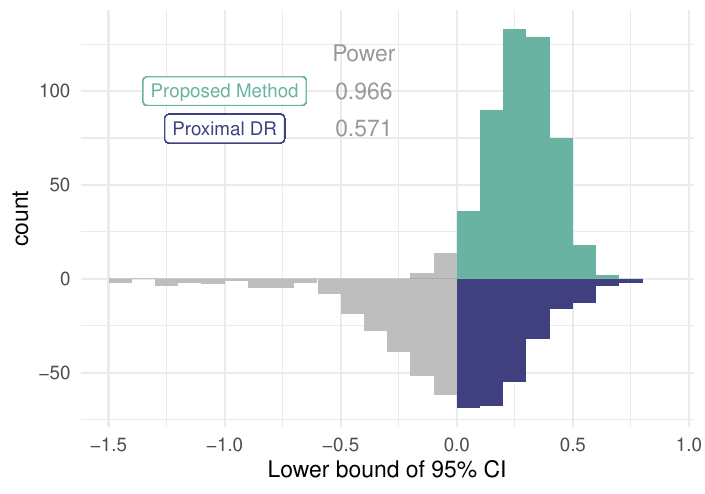}
    \end{subfigure}
    \caption{Comparison of two methods with a sample size of 400 over 500 replications. Data is generated following Scenario II in \Cref{ssec:simulation}, with a true causal effect of 0.5. The gray shading in the right panel indicates replicates that fail to reject the null hypothesis of no causal effect. }
    \label{plt:comparison}
\end{figure}

\section{Basic framework}\label{sec:setup}
Let $A\in\{0,1\}$ denote the treatment variable, and $Y$ denote the observed outcome. Let $Y(a)$ denote the potential outcome if treatment $A=a$ is assigned. The observed outcome is $Y=Y(A)$. Our objective is to estimate the average treatment effect, $\tau_0\triangleq \mE\{Y(1)-Y(0)\}$, in settings with unmeasured confounders. Let $U$ denote the unmeasured confounders that may affect both the treatment $A$ and outcome $Y$. Suppose there are three types of observable covariates $(X,W,Z)$:  $X$ affect both $A$ and $Y$; $W$ are outcome-inducing confounding proxies which are related to $A$ only through $(X, U)$; and $Z$ are treatment-inducing confounding proxies which are associated with $Y$ only through $(X, U)$. See \Cref{plt:DAG} for a causal directed acyclic graph (DAG).  We formalize the relationships among these variables in the following assumption. 

\begin{assumption}\label{ass:causal_and_proxy}
    We assume that:
    \begin{enumerate}[(i)] 
        \item (Latent exchangeability) $Y(a)\indep A\mid (U,X)$ for $a=0,1$;
        \item (Overlap) $0<c_1<\pr(A=a\mid U,X)<c_2<1$ for some constants $c_1$ and $c_2$;
        \item (Proxy variables) $Z\indep Y\mid A,U,X$, and $W\indep (A,Z) \mid U,X$;
        \item (Completeness) For any square-integrable function $g$ and any $a, x$, $\mE\{g(U)\mid Z,A=a,X=x\} = 0$ almost surely if and only if $g(U) = 0$ almost surely.
    \end{enumerate}
\end{assumption}

\begin{figure}[h]
    \centering
    \includegraphics[width=0.3\textwidth]{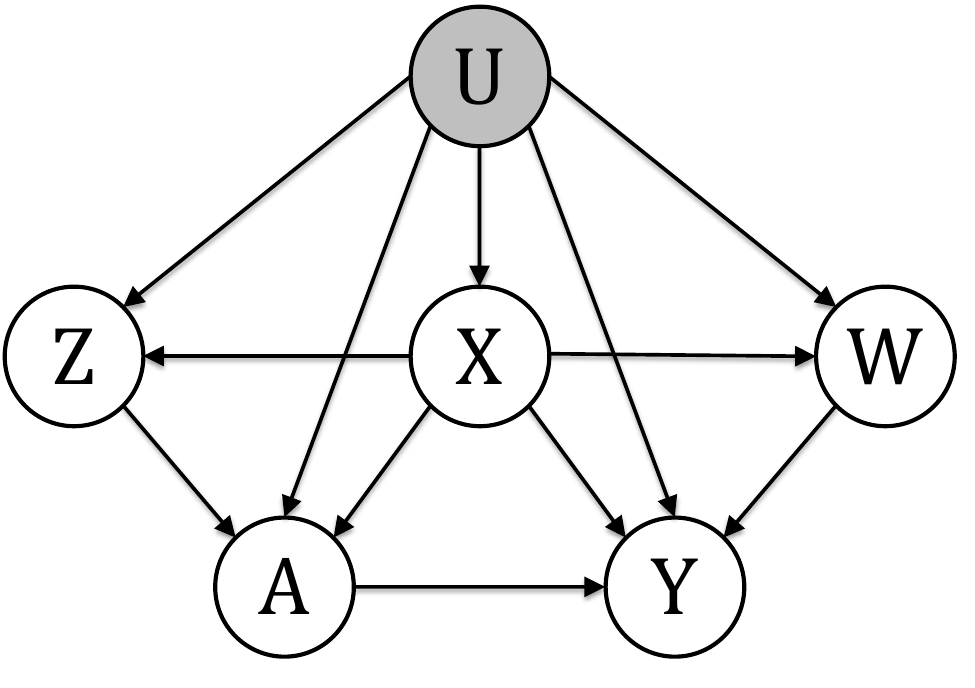}
    \caption{A causal DAG of proximal causal inference with unmeasured confounding and proxies.}
    \label{plt:DAG}
\end{figure}

\Cref{ass:causal_and_proxy}(i)-(ii) are standard in the literature of causal inference.
\Cref{ass:causal_and_proxy}(iii) describes the nature of the two types of proxies: treatment proxies $Z$ cannot affect the outcome $Y$, and the outcome proxies $W$ cannot be affected by either the treatment $A$ or the treatment proxies $Z$, upon conditioning on the confounders $X$ and $U$. These conditions are critical for identification but untestable as they involve conditional independence statements given the unmeasured variable $U$. Justifying these conditions and selecting proxy variables requires subject-matter knowledge; however, they can be satisfied in applications through careful study design. For instance, by incorporating post-outcome variables in $Z$ and pre-treatment measurements of the outcome in $W$, we can reasonably expect \Cref{ass:causal_and_proxy}(iii) to hold as the future cannot affect the past. \cite{MiaoShiLiTchetgenTchetgen2024} provided an example of evaluating the short-term effect of air pollution on elderly hospitalization using time series data, in which air pollution measurements taken after hospitalization are included in $Z$, and hospitalization measurements taken before air pollution are included in $W$.\label{res:R1-Q3}
\Cref{ass:causal_and_proxy}(iv) can be intuitively interpreted as a requirement for $Z$ to exhibit sufficient variability relative to that of the unmeasured confounder $U$. In the case of categorical $Z$ and $U$, it requires $Z$ with at least as many categories as $U$. In the continuous case, \cite{ChenChernozhukovLeeNewey2014} and \cite{Andrews2017} showed that if the dimension of $Z$ exceeds that of $U$, then under mild conditions, completeness holds generically, except for distributions in negligible sets.
\cite{TchetgenTchetgenYingCuiShiEtAl2024} also suggest measuring a rich set of baseline characteristics to make the completeness assumption more plausible; however, this benefit should be balanced against the efficiency loss incurred by including irrelevant proxies and the increased sample size required for estimating high-dimensional nuisance functions. Completeness is also shown to hold in many parametric and semiparametric models, such as exponential families \citep{NeweyPowell2003} and location-scale families \citep{HuShiu2018}. For a detailed discussion of completeness, see \cite{DHaultfoeuille2011}.
We also refer readers to \cite{TchetgenTchetgenYingCuiShiEtAl2024} for a detailed discussion of these assumptions, additional examples of proxies, and approaches for sensitivity analysis in practice.

Under \Cref{ass:causal_and_proxy}, \cite{MiaoGengTchetgenTchetgen2018} established the identification of ATE through an \emph{outcome-confounding bridge function} $h(w,a,x)$, which solves 
\begin{align} \label{eq:iden_bridge_h}
    \mE(Y\mid Z,A,X)=\mE\left\{h(W,A,X)\mid Z,A,X\right\}.
\end{align}
It is worth noting that the bridge function solves a Fredholm integral equation of the first kind, which is known to be an ill-posed inverse problem \citep{Kress1989}.
Then, ATE is identified by 
\begin{align}\label{eq:iden_ate}
    \tau_0 = \mE\left\{h(W,1,X)-h(W,0,X)\right\}.  
\end{align}
Equation \eqref{eq:iden_ate} is referred to as the \emph{proximal g-formula} in \cite{TchetgenTchetgenYingCuiShiEtAl2024}.

We consider the  estimation of $\tau_0$ using $N$ independent and identically distributed observations drawn from the joint distribution of $\O=(A,Y,W,X,Z)$. We follow \cite{TchetgenTchetgenYingCuiShiEtAl2024} by parameterizing the bridge function $h(W,A,X;\bm\gamma)$ with the parameter value $\bm\gamma_0\in\Gamma\subset\mR^p$. Then, one can estimate $\bm\gamma_0$ by solving
\begin{align}\label{eq:ee_gamma}
    \frac{1}{N}\sum_{i=1}^{N} \left\{Y_i-h(W_i,A_i,X_i;\bm\gamma)\right\}\m(Z_i,A_i,X_i)=\0,
\end{align}
where $\m(\cdot)$ is a user-specified vector-valued function, with dimension equal to that of $\bm\gamma$, and satisfies that $\mE\left\{\nabla_\gamma h(W,A,X;\bm\gamma_0)\m(Z,A,X)\trans\right\}$ is nonsingular. For instance, if $h(W,A,X;\bm\gamma)=(1,W,A,X)\bm\gamma$, one can choose $\m(Z,A,X)=(1,Z,A,X)\trans$. Under mild regularity conditions \citep{White1982}, $\wh{\bm\gamma}(\m)$ that solves \eqref{eq:ee_gamma} is consistent and asymptotically normally distributed for any $\m$. However, the choice of $\m$ may affect the efficiency, namely, the asymptotic variance of $\wh{\bm\gamma}(\m)$. Let $\nabla_\x\bm{\ell}(\x)=\partial\bm{\ell}(\x)\trans/\partial\x$ denote the gradient of a (vector-valued) function $\bm{\ell}$. \cite{CuiPuShiMiaoEtAl2023} showed that among all score functions of $\bm\gamma$, the efficient score is the one with \label{res:R1-Q6-1}
\begin{align*}
    \m_\eff(Z,A,X) = \frac{ \mE\{\nabla_{\bm\gamma} h(W,A,X;\bm\gamma_0)\mid Z,A,X\}}{ \mE[\{Y-h(W,A,X)\}^2 \mid Z,A,X]},
\end{align*}
that accounts for the potential heteroskedasticity. Unfortunately, such $\m_\eff$ is impractical because it requires modeling complex features of the observed data distribution, which are difficult to capture accurately. \cite{MiaoShiLiTchetgenTchetgen2024} suggested a recursive GMM approach, allowing the dimension of $\m$ to exceed that of $\bm\gamma$. 
While their approach may offer some efficiency gains for $\wh{\bm\gamma}(\m)$, the specific low-dimensional $\m$ does not necessarily lead to an efficient estimator for $\bm\gamma$, in contrast to the increasing moment conditions proposed in this paper, as elaborated later. \label{res:R1-Q6-2}
An inefficient $\wh{\bm\gamma}(\m)$ may lead to inefficiency of the plug-in estimator:
\begin{align*}
    \wh{\tau}^\plugin(\m) = \frac{1}{N}\sum_{i=1}^{N} \left\{h(W_i,1,X_i;\wh{\bm\gamma}(\m))-h(W_i,0,X_i;\wh{\bm\gamma}(\m))\right\}.
\end{align*}
Furthermore, the plug-in estimators, including $\wh{\tau}^\plugin(\m_\eff)$, fail to account for the correlation between the score functions of $\bm\gamma_0$ and $\tau_0$, thereby introducing an additional source of efficiency loss \citep{BrownNewey1998,AiChen2012}. 
 
\section{Estimation}\label{sec:estimation}
To improve existing approaches, we propose estimating $\bm\gamma_0$ and $\tau_0$ jointly and using an increasing number of moment restrictions. Specifically, let $\u_{K}(z,a,x) = \big(u_{K1}(z,a,x),\ldots, u_{KK}(z,a,x)\big)\trans$ denote a vector of known basis functions (such as power series, splines, Fourier series, etc.) with dimension $K\in\mN$, which provides approximation sieves that can approximate a large class of smooth functions arbitrarily well as $K\rightarrow\infty$. The selection of $\u_{K}(z,a,x)$ is discussed in \Cref{sec:implementation}. Model \eqref{eq:iden_bridge_h} implies the following unconditional moment restrictions of $\bm\gamma_0$:
\begin{align}\label{eq:iden_gamma_sieve}
   \mE\left[\left\{Y-h(W,A,X;\bm\gamma_0)\right\}\u_{K}(Z,A,X)\right]=\0.
\end{align}
Denote the joint score function
$$\g_K(\O;\bm\gamma,\tau)=\left(
    \begin{array}{c}
        \{Y-h(W,A,X;\bm\gamma)\}\u_K(Z,A,X)\\
        \tau-h(W,1,X;\bm\gamma)+h(W,0,X;\bm\gamma)
    \end{array}
    \right),$$
and $\G_K(\bm\gamma,\tau)= N^{-1}\sum_{i=1}^N \g_K(\O_i;\bm\gamma,\tau)$. Since $K$ increases with the sample size, the number of moment restrictions typically exceeds that of unknown parameters. So, we apply the GMM method for estimation. For a user-specified $(K+1)\times(K+1)$ positive definite matrix $\bm\Omega$, the GMM estimator of $(\bm\gamma,\tau)$ is given by 
 \begin{align}\label{eq:gmm_init}
 (\check{\bm\gamma},\check\tau) = \amin_{(\bm\gamma,\tau)\in\Gamma\times\calT} \G_K(\bm\gamma,\tau)\trans \bm\Omega \G_K(\bm\gamma,\tau).
 \end{align}
\cite{Hansen1982} showed that, with a fixed $K\ge p$, under some regularity conditions, $(\check{\bm\gamma},\check\tau)$ are consistent and asymptotically normally distributed but perform best only when $\bm\Omega$ is selected as the inverse of ${\bm\Upsilon}_{(K+1)\times(K+1)}= \mE\{\g_K(\O;\bm\gamma_0,\tau_0)\g_K(\O;\bm\gamma_0,\tau_0)\trans\}$.
We use the initial estimator $(\check{\bm\gamma},\check\tau)$ to obtain an estimator of ${\bm\Upsilon}_{(K+1)\times(K+1)}$: 
\begin{align*}
    \wh{\bm\Upsilon}_{(K+1)\times(K+1)}= \frac{1}{N} \sum_{i=1}^{N} \g_K(\O_i;\check{\bm\gamma},\check\tau)\g_K(\O_i;\check{\bm\gamma},\check\tau)\trans.
\end{align*}
We then obtain the optimal GMM estimator
\begin{align}\label{eq:opt_GMM_est}
    (\wh{\bm\gamma},\wh\tau) = \amin_{(\bm\gamma,\tau)\in\Gamma\times\calT} \G_K(\bm\gamma,\tau)\trans \wh{\bm\Upsilon}_{(K+1)\times(K+1)}^{-1} \G_K(\bm\gamma,\tau).
\end{align}   
With a fixed $K$, \cite{Hansen1982} showed that under regularity conditions,
\begin{align}\label{eq:GMM_normal}
\V_K^{-1/2}\left( 
    \begin{array}{c}
        \sqrt{N}(\wh{\bm\gamma}-\bm\gamma_0)\\\sqrt{N}(\wh\tau-\tau_0)
    \end{array}
 \right) \convd N\left(\0,\I_{(p+1)\times(p+1)}\right),
\end{align}
where $\V_K = \{\B_{(K+1)\times(p+1)}\trans {\bm\Upsilon}_{(K+1)\times(K+1)}^{-1} \B_{(K+1)\times(p+1)}\}^{-1}$ and
\begin{align*}
\B_{(K+1)\times(p+1)} 
\triangleq  \mE\left\{\nabla_{\bm\gamma,\tau}\G_K(\bm\gamma_0,\tau_0)\trans\right\}
=  \left(
    \begin{array}{cl}
        -\mE\{\u_K(Z,A,X)\nabla_{\bm\gamma} h(W,A,X;\bm\gamma_0)\trans\}  & \0\\
        -\mE\left\{\frac{(-1)^{1-A}}{f(A\mid W,X)} \nabla_{\bm\gamma} h(W,A,X;\bm\gamma_0)\trans\right\}  & 1
    \end{array}
\right).
\end{align*}

\begin{remark}
    When $K$ is fixed, the optimal GMM estimator in \eqref{eq:opt_GMM_est} is generally not semiparametrically efficient. In our method, we allow $K$ to increase slowly with the sample size, such that $\{\u_K(z,a,x)\}$ spans the space of measurable functions and provides a good approximation of $\m_{\eff}(z,a,x)$, thereby achieving full efficiency.
\end{remark}

\begin{remark}\label{remark:moment_sel}
    Numerical instability may arise if the combined moment conditions \eqref{eq:iden_gamma_sieve} exhibit near-linear dependence, causing the asymptotic covariance matrix to be singular. When this occurs, some linear combinations of the estimating equations contain negligible independent information about the data-generating process. Therefore, discarding these combinations incurs minimal information loss. Specifically, we recommend a two-step regularization procedure. First, we orthonormalize the basis functions $\u_K$ so that their empirical second-moment matrix is the identity, i.e., $N^{-1}\sum_{i=1}^N\u_{K}(Z_i,A_i,X_i)\u_{K}(Z_i,A_i,X_i)\trans=\I_K$. Second, we apply GMM to all estimating equations using the identity weighting (i.e., setting $\bm\Omega=\I_{K+1}$ in \eqref{eq:gmm_init}) to obtain initial estimators $(\check{\bm\gamma},\check\tau)$. We apply spectral decomposition to the estimated moment covariance matrix $\wh{\mathbf{\Upsilon}}_{(K+1)\times(K+1)}= \Q \mathbf{\Lambda} \Q\trans$. With $\Q_{K_1}$ as eigenvectors corresponding to the $K_1$ eigenvalues larger than a small threshold $c$, we apply GMM to the linear combination of the estimating equations $\Q_{K_1}\trans \g_K$, using the optimal weight matrix $\mathbf{\Lambda}_{K_1}^{-1}$. Here,  $\mathbf{\Lambda}_{K_1}$ is the diagonal matrix composed of the first $K_{1}$  eigenvalues. This procedure effectively filters out highly irrelevant moment conditions and ensures that the weighting matrix is invertible.
\end{remark}

\begin{remark}\label{remak:stagewise}
      The plug-in estimator $\wh{\tau}^\plugin$ can be interpreted as a GMM estimator that utilizes a block-diagonal weighting matrix. From this perspective, the efficiency loss arises from the suboptimality of this weighting scheme, which fails to exploit the covariance among the estimation equations. In a broader context, this phenomenon mirrors the distinction between simple two-stage estimators (e.g., 2SLS) and joint estimation methods such as efficient GMM or limited-information maximum likelihood (LIML), in which joint estimation achieves greater efficiency by accounting for the correlation structure across stages.
\end{remark}

\section{Large-sample properties}\label{sec:asym_prop}

In this section, we establish the large-sample properties of the proposed estimator as the number of moment restrictions increases. We impose the following assumptions.

\begin{assumption}\label{ass:sieve}
    (i) The eigenvalues of $\mE\{\u_{K}(Z,A,X)\u_{K}(Z,A,X)\trans\}$ are bounded and bounded away from zero for all $K$; 
    (ii) For any $p$-smooth function $h(z,a,x)$ defined in \Cref{sec:app_regularity_conditions}, there is $\bb_h\in\mR^{K}$ such that  $\sup_{(z,a,x)\in\calZ\times\{0,1\}\times\calX}  |h(z,a,x)-\u_K(z,a,x)\trans\bb_h| = O(K^{-\alpha})$ with $\alpha>0$; 
     (iii) $\zeta(K)^2K/N=o(1)$, where $\zeta(K)\triangleq \sup_{(z,a,x)\in\calZ\times\{0,1\}\times\calX} \|\u_K(z,a,x)\|$.
\end{assumption}
\Cref{ass:sieve}(i) rules out the degeneracy of moment restrictions. 
As indicated in \Cref{remark:moment_sel}, we recommend orthonormalizing the basis functions such that the empirical second-moment matrix is the identity.
 \Cref{ass:sieve}(ii)  requires sieve approximation error rates for the $p$-smooth function class, which have been well studied in the mathematical literature on approximation theory. For instance, suppose $\calX$ and $\calZ$ are compact subsets in $\mR^{d_x}$ and $\mR^{d_z}$, respectively, and $\u_K$ is a tensor product of polynomials or B-splines as introduced in \Cref{sec:implementation}, \citet[p. 5537]{Chen2007Handbook} showed that \Cref{ass:sieve}(ii) holds with $\alpha=p/(d_x+d_z)$. \Cref{ass:sieve}(iii) restricts the number of moments to ensure the convergence and asymptotic normality of the proposed estimator. \cite{Newey1997} showed that if $\u_K$ is a power series, then $\zeta(K)=O(K)$, and it requires $K=o(N^{1/3})$. If $\u_K$ is a B-spline, then $\zeta(K)=O(\sqrt{K})$ and $K=o(\sqrt{N})$. These conditions guide the choice of $\u_K$, which is discussed in \Cref{sec:implementation}.\label{res:R1-Q14}
The asymptotic distribution of $\wh{\bm\gamma}$ and $\wh\tau$ is formally established in the following theorem.

\begin{theorem}\label{thm:asym_normal}
    Suppose \Cref{ass:causal_and_proxy,ass:sieve}, and regularity conditions in \Cref{sec:app_regularity_conditions} hold.
        We have $\sqrt{N}(\wh{\bm\gamma}-\bm\gamma_0)\convd N(\0,\V_{\bm\gamma})$ and $\sqrt{N}(\wh\tau-\tau_0)\convd N(0,V_{\tau})$ with $\V_{\bm\gamma}=\mE\{\bm\psi_1(\O)\bm\psi_1(\O)\trans\}$, $V_{\tau}=\mE\{\psi_2(\O)^2\}$, and 
            \begin{align*}
                & \bm\psi_1(\O) = 
                \left[\mE\left\{\frac{\partial h(W,A,X;\bm\gamma_0)\m_\eff(Z,A,X)}{\partial\bm\gamma\trans }\right\}\right]^{-1} 
                \{Y-h(W,A,X)\}\m_\eff(Z,A,X),\\ 
                & \psi_2(\O) = h(W,1,X)-h(W,0,X)-\tau_0+t(Z,A,X)\{Y-h(W,A,X)\},\\
                & t(Z,A,X) = \bm\kappa\trans\m_\eff(Z,A,X) - R(Z,A,X),\\ 
                & R(Z,A,X)= \frac{ \mE[\{Y-h(W,A,X)\}\{ h(W,1,X)-h(W,0,X)-\tau_0\}\mid Z,A,X]}{ \mE[\{Y-h(W,A,X)\}^2 \mid Z,A,X]},\\
                & \bm\kappa = \V_{\bm\gamma}\mE\left[\left\{R(Z,A,X)+\frac{(-1)^{1-A}}{f(A\mid W,X)}\right\}\nabla_{\bm\gamma} h(W,A,X;\bm\gamma_0)\right].
            \end{align*}
\end{theorem}
\Cref{thm:asym_normal} demonstrates that $(\wh{\bm\gamma},\wh\tau)$ remains $\sqrt{N}$-consistent and asymptotically normally distributed when the number of moments increases slowly with the sample size. It follows intermediate lemmas in \Cref*{sec:app_intermediate_result} of the supplementary material. Specifically, we show that the initial estimators $(\check{\bm\gamma},\check\tau)$ obtained from \eqref{eq:gmm_init} are $\sqrt{N}$-consistent. We then demonstrate that \eqref{eq:GMM_normal} still holds as $K\to\infty$ slowly. As a result, we obtain \Cref{thm:asym_normal} by calculating the limit of $\V_K$ in \eqref{eq:GMM_normal}.
Notably, $\V_{\bm\gamma}$ is the semiparametric efficiency bound of $\bm\gamma_0$, derived in Theorem E.2 of \cite{CuiPuShiMiaoEtAl2023}. Thus, the proposed estimator of $\bm\gamma_0$ is efficient. To see the efficiency of $\wh\tau$, we consider the semiparametric local efficiency bound of $\tau_0$ derived in \citet[Theorem 3.1]{CuiPuShiMiaoEtAl2023}, which requires the following assumption.

\begin{assumption}\label{ass:proxy_cui}
    We assume that:
    \begin{enumerate}[(i)]
        \item
         For any square-integrable function $g$ and for any $a, x$, $\mE\{g(U)\mid W,A=a,X=x\} = 0$ almost surely if and only if $g(U) = 0$ almost surely.
         \item There exists a treatment-confounding bridge function $q(z,a,x)$ that satisfies
         \begin{align}\label{eq:iden_bridge_q}
             \mE\{q(Z,A,X)\mid W,A,X\}=\frac{1}{f(A \mid W, X)}.
         \end{align} 
        \item The conditional expectation mappings $T(g)\equiv E\{g(W,A,X)\mid Z,A,X\}$ and $T'(g)\equiv E\{g(Z,A,X)\mid W,A,X\}$ are surjective.
    \end{enumerate}
\end{assumption}
\Cref{ass:proxy_cui}(i)-(ii) permits an alternative identification formula of ATE, given by $\tau=\mE\{(-1)^{1-A}\allowbreak q(Z,A,X)Y\}$, using the treatment-confounding bridge function $q$ instead of $h$ in the proximal g-formula \eqref{eq:iden_ate} as established in \cite{CuiPuShiMiaoEtAl2023}. \Cref{ass:proxy_cui}(iii) ensures that $h$ and $q$ are uniquely identified by the integral equations \eqref{eq:iden_bridge_h} and \eqref{eq:iden_bridge_q}, respectively.
Combining these two identification results, \cite{CuiPuShiMiaoEtAl2023} showed that the efficient influence function of $\tau_0$, under the semiparametric model $\mathcal{M}_{sp}$, which does not restrict the observed data distribution other than the existence of a bridge function $h$ that solves \eqref{eq:iden_bridge_h}, evaluated at the submodel where \Cref{ass:proxy_cui} holds, is 
\begin{align}\label{eq:eif}
   \psi_{\eff}(\O) = h(W,1,X)-h(W,0,X)-\tau_0+(-1)^{1-A}q(Z,A,X)\{Y-h(W,A,X)\}.
\end{align}
Therefore, the semiparametric local efficiency bound of $\tau_0$ under $\calM_{sp}$ is $V_{\tau,\eff}=\mE\{\psi_{\eff}(\O)^2\}$.

\begin{theorem}\label{coro:super_eff}
    (i) $V_{\tau}\le V_{\tau,\eff}$,
    and the equality holds if and only if there is a vector of constants $\bm\alpha$ such that $R(Z,A,X)+(-1)^{1-A}q(Z,A,X)  = \bm\alpha\trans\m_\eff(Z,A,X)$; 
    (ii) $V_{\tau}$ is the semiparametric local efficiency bound of $\tau_0$ under the submodel $\calM_{sub}$ where $h(w,a,x;\bm\gamma)$ is correctly specified and uniquely determined by \eqref{eq:iden_bridge_h}.
\end{theorem}
\Cref{coro:super_eff} shows that the proposed estimator $\wh{\tau}$ has an asymptotic variance that is always no larger than the semiparametric efficiency bound $V_{\tau,\eff}$.
The efficiency gain arises from parameterizing the bridge function $h$. We note that the doubly robust estimator, constructed using the efficient score \eqref{eq:eif}, attains the semiparametric efficiency bound $V_{\tau,\eff}$ only if both bridge functions $h$ and $q$ are correctly specified and consistently estimated. We also note that the condition that guarantees $V_{\tau}= V_{\tau,\eff}$ is unlikely to hold, except under highly contrived data distributions. As demonstrated in the simulation studies in \Cref{sec:numerical}, it does not hold even in the most straightforward linear data-generating processes. Therefore, the proposed estimator generally outperforms the doubly robust estimator equipped with consistently estimated $h$ and $q$. 

\begin{remark}\label{remark:bias-var-tradeoff}
      \Cref{thm:asym_normal,coro:super_eff} are established under asymptotics with increasing $K$. Intuitively, this involves a bias-variance trade-off as highlighted by \cite{DonaldImbensNewey2009}. A larger $K$ implies using a richer set of moment conditions, which improves estimation efficiency (reduces variance) but can also increase bias. Assumption \ref{ass:sieve} is crucial here as it restricts $K$ from growing too rapidly, thereby balancing the need for efficiency against the risk of excessive bias. Furthermore, we address the practical selection of $K$ in \Cref{sec:implementation} using a data-driven approach based on MSE minimization.
\end{remark}

To highlight the deficiency of the plug-in procedure, the following \namecref{coro:super_plug_in} compares the proposed estimator $\wh\tau$ with the theoretically optimal plug-in estimator $\wh{\tau}^\plugin(\m_\eff)$. 
\begin{theorem}\label{coro:super_plug_in}
    Suppose that $\sqrt{N}\{\wh{\tau}^\plugin(\m_\eff)-\tau_0\}\convd N(0,V^\plugin_{\tau,\eff})$, then we have $V_{\tau}\le V^\plugin_{\tau,\eff}$, and the equality holds if and only if there is a vector of constants $\bm\alpha$ such that $R(Z,A,X) = \bm\alpha\trans\m_\eff(Z,A,X).$
\end{theorem}
The condition that guarantees $V_{\tau}=V^\plugin_{\tau,\eff}$ is satisfied if the observed distribution concerning the correlation between the score functions of $\tau_0$ and $\bm\gamma_0$, characterized by $R(Z,A,X)$, meets a particular structure. A sufficient condition is that there is no additive interaction between $A$ and $W$ in the model of $h(W,A,X)$. In this case, equality holds for $\bm\alpha=\0$. 
In case of its violation, \Cref{coro:super_plug_in} indicates that the proposed estimator of $\tau_0$ outperforms plug-in estimators, even though an efficient estimator of $\bm\gamma_0$ is employed.

\section{Implementation details}\label{sec:implementation}
The GMM algorithm can be easily implemented using routine software, such as $\mathsf{gmm}$ in R. Here, we briefly discuss the selection of tuning parameters. 
A large class of sieves $\u_K(z,a,x)$ is feasible for implementing the proposed approach. In this paper, we suggest a tensor-product linear sieve basis. To simplify the presentation, we suppose both $Z$ and $X$ are scalar variables. Let $\{\varphi_j(x) : j=1, \ldots, K_1\}$ denote a sieve basis for $\calL_2(\calX)$, the space of square Lebesgue integrable functions on $\calX$. For instance, if $\calX=[0,1]$, one can choose the basis of polynomials $\varphi_j(x)=x^{j-1}$ or B-spline, and if $\calX$ is unbounded, one can select the Hermite polynomial basis $\varphi_j(x)=\exp(-x^2)x^{j-1}$. Similarly, let $\{\phi_k(z):k=1,\ldots,K_2\}$ denote a basis for $\calL_2(\calZ)$. Then, the tensor-product sieve basis for $\calZ\times\{0,1\}\times\calX$ is given by $\{I(a=l)\varphi_j(x)\phi_k(z):j=1,\ldots,K_1;k=1,\ldots,K_2;l=0,1\}$ with the number of terms $K=2K_1K_2$. 
However, the number of tensor-product basis terms grows exponentially with the dimension of the arguments. In practice, additive separable bases can be employed to mitigate the issue of a large $K$. We refer readers to \cite{Newey1997} and \cite{Chen2007Handbook} for further details on sieve selection.\label{res:R1-Q21}

Another tuning parameter is the number of moments $K$. While the large-sample properties of the proposed estimator permit a wide range of values for $K$, practical guidance on selecting smoothing parameters is necessary for applied researchers who generally have only one finite sample at their disposal. Here, we propose an asymptotic mean-square-error-based criterion, as in \citet{DonaldImbensNewey2009}, for the data-driven selection $K$. Specifically, we choose $K\in\{1,\ldots,\bar{K}\}$ to minimize $S_{\rm{GMM}}(K):=\sum_{j=1}^p\{\wh\Pi(K;\e_j)^2/N+\wh\Phi(K;\e_j)\}$ defined in \Cref{sec:app_sel_K}, where $\wh\Pi(K;\e_j)^2/N$ is an estimate of the squared bias term of $\wh\gamma_j$ derived in \cite{NeweySmith2004}, and $\wh\Phi(K;\e_j)$ is an estimate of the asymptotic variance term of $\wh\gamma_j$. We refer interested readers to \cite{DonaldImbensNewey2009} for further insights.

\section{Numerical studies}\label{sec:numerical}

\subsection{Simulation}\label{ssec:simulation}
We construct Monte Carlo simulations to examine the finite-sample performance of the proposed approach.  We generate i.i.d samples from the following data-generating process:
\begin{align*}
&X,U\sim N(0,1),
~~\varepsilon_j\mid X=x\sim N\left(0,\sigma^2_j(x)\right),j=1,2,3,\\
& \text{logit}\{\pr(A=1\mid X,U)\} = (1,X,U)\bm\beta_a,~~
Z=(1,A,X,U)\bm\beta_z+\varepsilon_1, \\ 
& W= (1,X,U)\bm\beta_w+\varepsilon_2,~~
Y= (1,A,W,X,U)\bm\beta_y+\varepsilon_3.
\end{align*}
with the parameters $\bm\beta_a=(-0.1,0.5,0.5)\trans$, $\bm\beta_z=(0.5,1,0.5,1)\trans$, $\bm\beta_w=(1,-1,1)\trans$, and $\bm\beta_y=(1,0.5,0.5,1,1)\trans$.
We consider two scenarios for $\sigma_j(x)$. Scenario I is a simple setting with no heteroskedasticity, $\sigma_j(x)\equiv 1$ for $j=1,2,3$. In Scenario II, heteroskedasticity is present, with $\sigma_1(x)\equiv 1$, $\sigma_2(x)=(0.3+x^2)^{-1/2}$, $\sigma_3(x)=(0.5+0.8x^2)^{-1/2}$. The average treatment effect $\tau_0=0.5$.
We show in \Cref*{sec:app_supp_simu} of the supplementary material that the above data-generating mechanism is compatible with the following models of $h$ and $q$:
\begin{align*}
    & h(w,a,x;\bm\gamma) = \gamma_0 +  \gamma_1w + \gamma_2a + \gamma_3x,\\
    & q(z,a,x;\bm\theta) =  1+\exp\left\{(-1)^{a}(\theta_0+\theta_1z+\theta_2a+\theta_3x)\right\}.
\end{align*}

Our proposed method, denoted GMM-div, is computed using a power series and a data-driven smoothing parameter $K$ with $\bar{K}=12$ in both scenarios. For comparison, we include five additional methods:
    (i) Naive estimation, which regards all three covariates $(W,Z,X)$ as confounders,  and estimates ATE using classical g-formula \citep{GreenlandRobins1986}.
    (ii) RGMM, recursive GMM estimation with a fixed number of moments introduced in \cite{MiaoShiLiTchetgenTchetgen2024}. Following their suggestion, we choose the moment restrictions $\mE[\{Y-h(W,A,X;\bm\gamma)\}(1,Z,A,X)\trans]=0$. In this case, it is equivalent to the \emph{proximal outcome regression} estimator introduced in \cite{CuiPuShiMiaoEtAl2023}.
    (iii) P2SLS, the \emph{proximal two-stage least squares} introduced in \cite{TchetgenTchetgenYingCuiShiEtAl2024}, which can be implemented using the command $\text{ivreg}(Y\!\sim\! A+X+W\mid A+X+Z)$ in R.
    (iv) PIPW and PDR, the \emph{proximal inverse probability weighting} and the \emph{proximal doubly robust} estimator introduced by \cite{CuiPuShiMiaoEtAl2023}, which require modeling the treatment confounding bridge function $q(z,a,x;\bm\theta)$. Specifically, $\wh\tau_{PIPW}=N^{-1}\sum_{i=1}^N (-1)^{1-A_i}q(Z_i,A_i,X_i;\wh{\bm\theta})Y_i$, and $\wh\tau_{DR}$ solves the empirical analogue of \eqref{eq:eif} by replacing $h$ and $q$ with their estimates. Here, $\wh{\bm\theta}$ is obtained by solving the estimating equation $\mE\{(-1)^{1-A}q(Z,A,X;\bm\theta)(1,W,A,X)\trans - (0,0,1,0)\trans\}=0$.
Across all methods, we replicate 500 simulations at sample sizes of 400 and 800.

\Cref{table:simu1} reports the absolute bias, standard error,  root mean squared error, coverage probability, average length of 95\% confidence intervals, and the power of the hypothesis testing problem $H_0:\tau_0=0~~v.s.~~H_1:\tau_0\neq 0$ in both scenarios. The power is calculated as the proportion of rejected cases across 500 replications. It reveals that: (i) The naive estimator is severely biased due to unmeasured confounding, while the other five methods exhibit negligible bias as expected in both scenarios.
(ii) In Scenario I without heteroskedasticity, the five methods perform similarly, and the proposed method mainly selects $K=4$ across replications as shown in \Cref{plt:K_sel}(a).
(iii) In Scenario II, the proposed method demonstrates significantly lower standard error, narrower confidence intervals, and, notably, higher power in detecting non-zero causal effects. It benefits from the increased efficiency of additional selected moments as shown in \Cref{plt:K_sel}(b). \Cref{plt:plt_scen2} compares the empirical distributions of different methods, illustrating that the proposed method performs more concentrated around the true value.
(iv) Due to the tradeoff between type I and type II errors, the proposed estimator provides a relatively anti-conservative confidence interval. Given the small sample size in Scenario II, this yields a slightly smaller CP than other methods. Nevertheless, it approaches the nominal level as the sample size increases.

\begin{table} 
    \centering
    \caption{Simulation results: absolute bias, standard error (SE), root mean squared error (RMSE), coverage probability (CP), average length of the 95\% confidence interval, and power.}
    \label{table:simu1}
    \par
    \resizebox{\linewidth}{!}{
    \begin{tabular}{cccccccccccccc} \toprule 
        & \multicolumn{6}{c}{$n=400$}
        && \multicolumn{6}{c}{$n=800$} \\  \cline{2-7}\cline{9-14}
        Method & Bias & SE & RMSE & Length & CP & Power && Bias & SE & RMSE & Length & CP & Power\\ \midrule
        \multicolumn{12}{l}{\textbf{Scenario I: homoskedasticity}} \\
        Naive & 0.17 & 0.13 & 0.21 & 0.53 & 0.776 & 0.708 && 0.17 & 0.10 & 0.19 & 0.37 & 0.580 & 0.922\\
        RGMM & 0.01 & 0.16 & 0.16 & 0.60 & 0.942 & 0.888 && 0.01 & 0.11 & 0.11 & 0.43 & 0.950 & 0.988\\
        P2SLS & 0.00 & 0.16 & 0.16 & 0.61 & 0.944 & 0.898 && 0.01 & 0.11 & 0.11 & 0.43 & 0.956 & 0.986\\
        PIPW & 0.00 & 0.16 & 0.16 & 0.64 & 0.948 & 0.852 && 0.01 & 0.12 & 0.12 & 0.45 & 0.958 & 0.980\\
        PDR & 0.00 & 0.16 & 0.16 & 0.63 & 0.944 & 0.854 && 0.01 & 0.12 & 0.12 & 0.44 & 0.962 & 0.982\\
        GMM-div & 0.00 & 0.16 & 0.16 & 0.60 & 0.942 & 0.886 && 0.01 & 0.11 & 0.11 & 0.43 & 0.950 & 0.988\\
        \midrule
        \multicolumn{12}{l}{\textbf{Scenario II: heteroskedasticity}} \\
        Naive & 0.23 & 0.18 & 0.29 & 0.71 & 0.764 & 0.312 && 0.22 & 0.13 & 0.26 & 0.50 & 0.598 & 0.602\\
        RGMM & 0.00 & 0.28 & 0.28 & 1.04 & 0.954 & 0.500 && 0.01 & 0.19 & 0.19 & 0.72 & 0.952 & 0.726\\
        P2SLS & 0.01 & 0.28 & 0.28 & 1.03 & 0.952 & 0.490 && 0.01 & 0.19 & 0.19 & 0.72 & 0.948 & 0.730\\
        PIPW & 0.01 & 0.29 & 0.29 & 1.08 & 0.956 & 0.534 && 0.03 & 0.22 & 0.22 & 0.78 & 0.954 & 0.734\\
        PDR & 0.02 & 0.31 & 0.31 & 1.08 & 0.950 & 0.518 && 0.03 & 0.23 & 0.23 & 0.75 & 0.942 & 0.728\\
        GMM-div & 0.01 & 0.13 & 0.13 & 0.48 & 0.936 & 0.966 && 0.00 & 0.09 & 0.09 & 0.34 & 0.956 & 1.000\\
        \bottomrule
    \end{tabular}
    }
    \end{table}

    \begin{figure}
        \centering
        \begin{subfigure}[b]{0.49\textwidth}
            \centering
            \includegraphics[width=\textwidth]{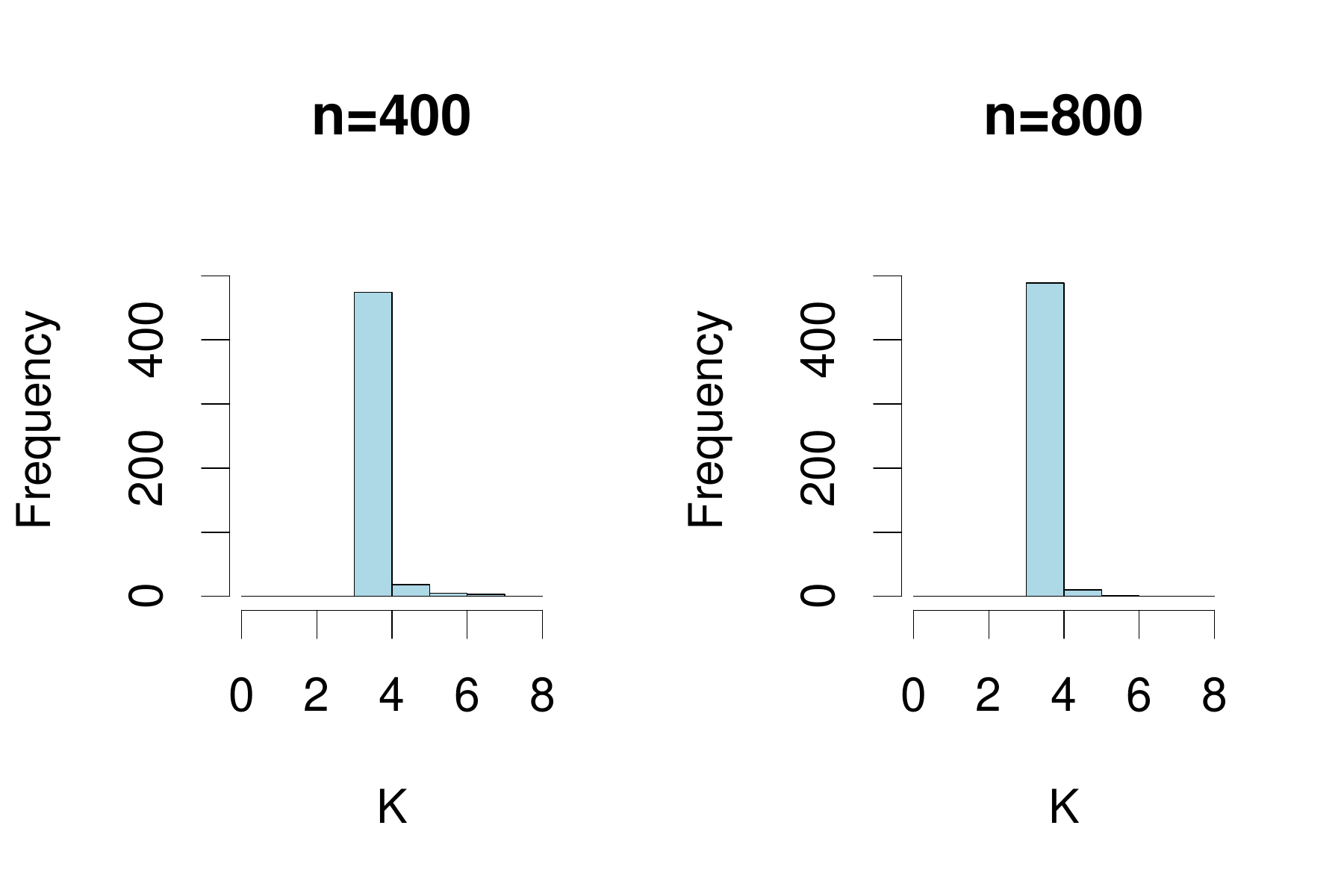}
            \caption{Scenario I: homoskedasticity}
        \end{subfigure}
        \hfill
        \begin{subfigure}[b]{0.49\textwidth}
            \centering
            \includegraphics[width=\textwidth]{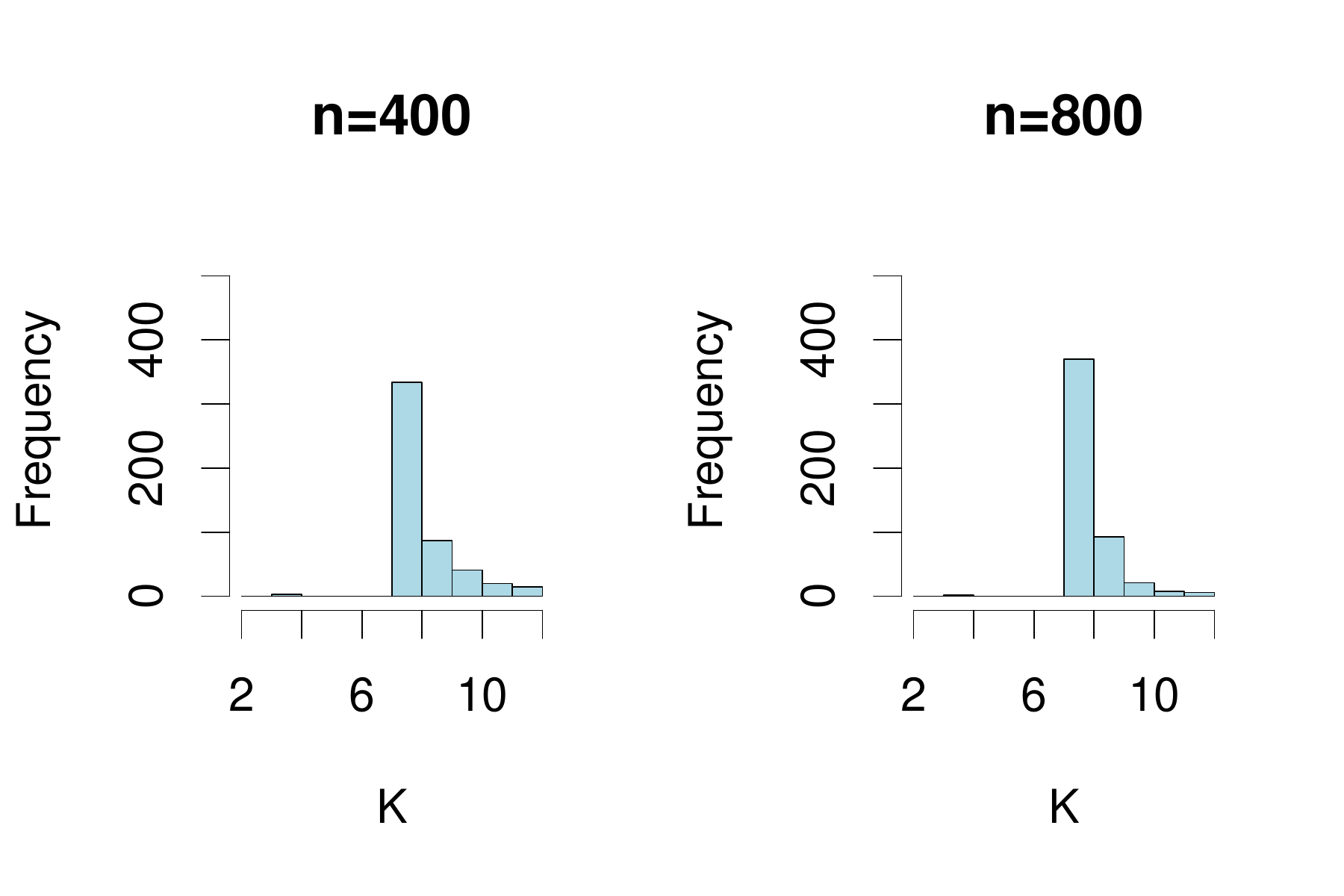}
            \caption{Scenario II: heteroskedasticity}
        \end{subfigure}
        \caption{Histogram of selected $K$.}
    \label{plt:K_sel}
    \end{figure}

    \begin{figure}
        \centering   
        \begin{subfigure}[b]{0.48\textwidth}
            \centering
            \includegraphics[width=\textwidth]{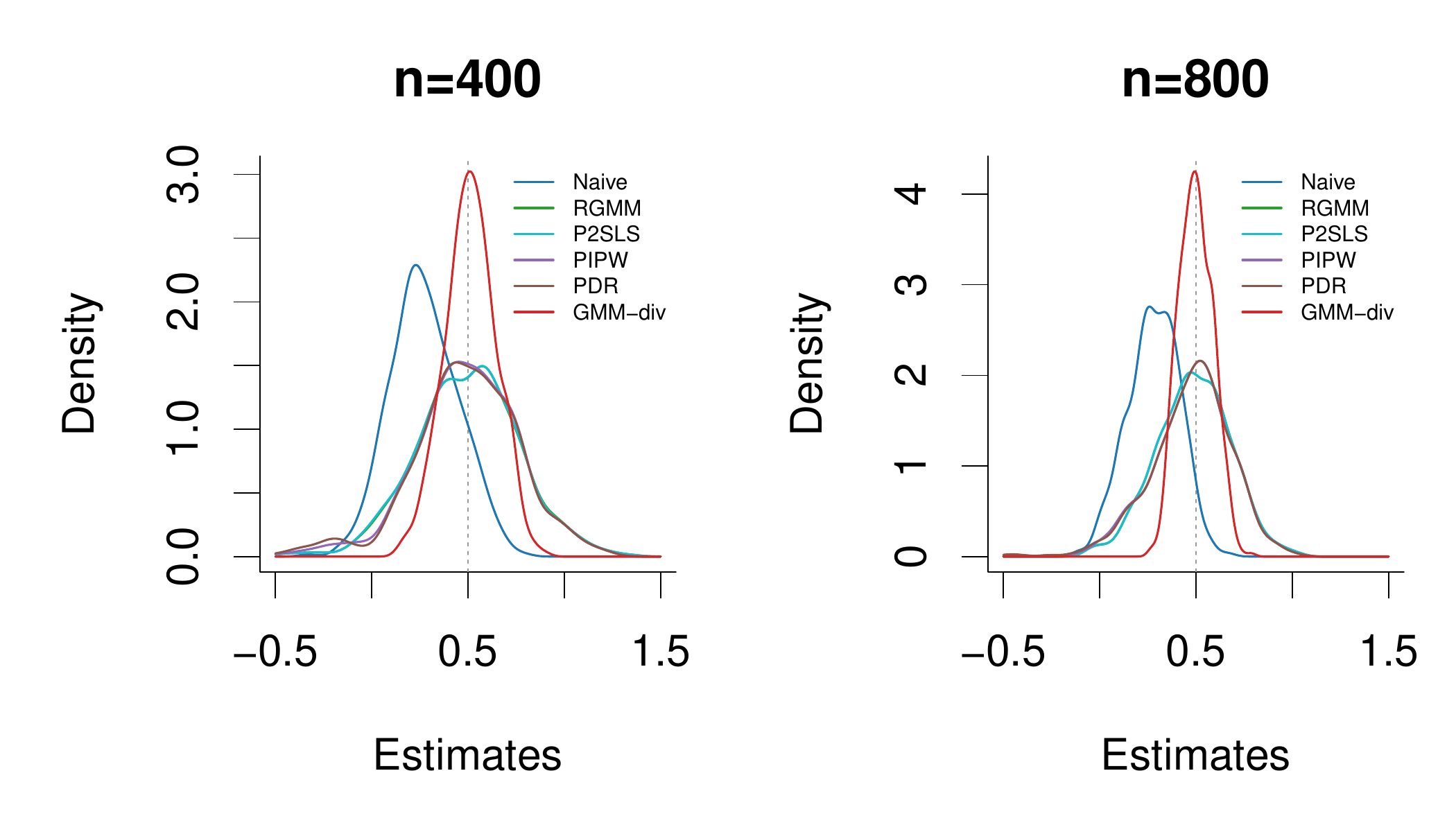}
            \caption{}
        \end{subfigure}
        \hfill
        \begin{subfigure}[b]{0.48\textwidth}
            \centering
            \includegraphics[width=\textwidth]{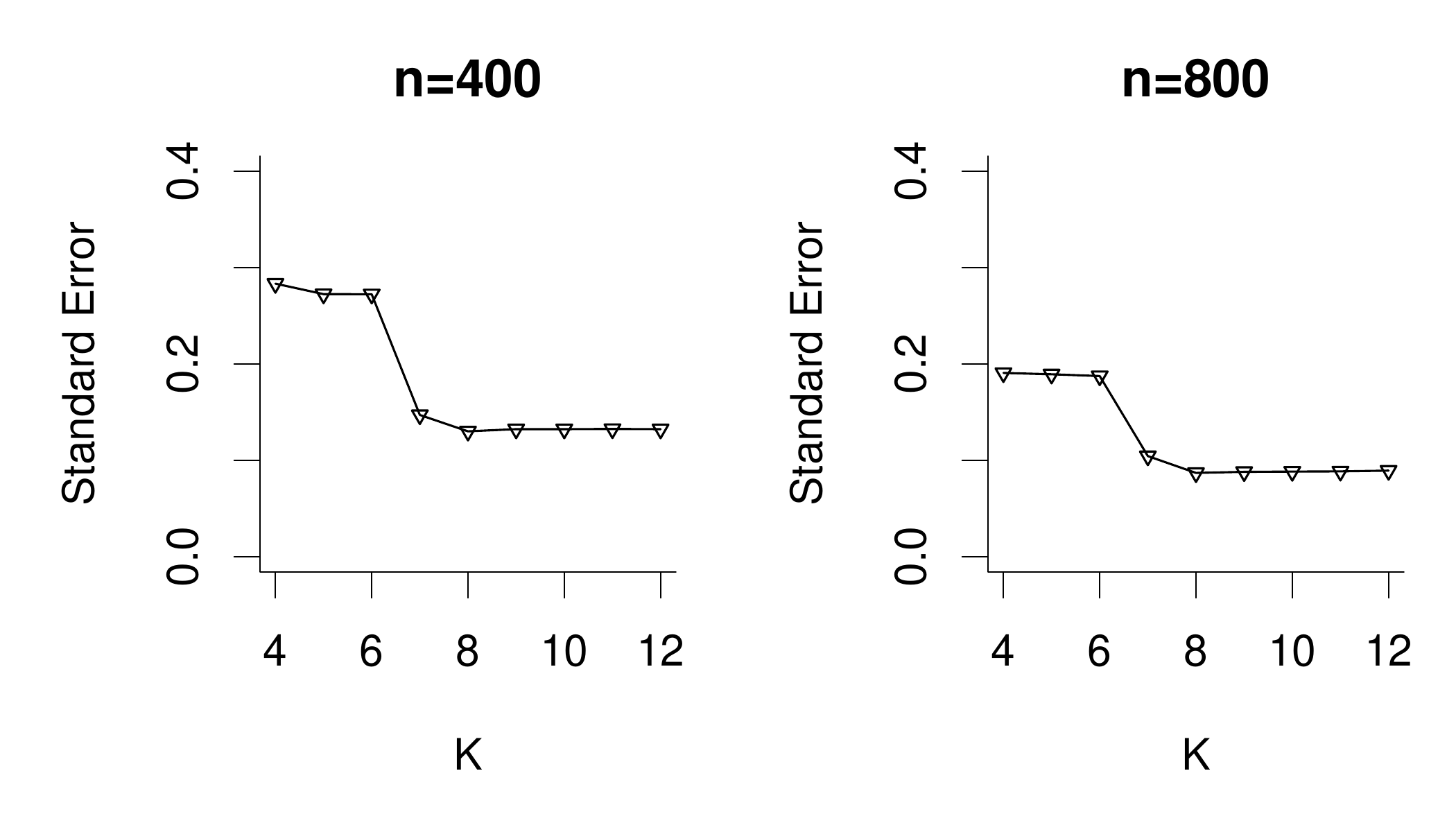}
            \caption{}
        \end{subfigure}
        \caption{(a) Empirical distributions of different methods under Scenario II. (b) Standard error of the proposed estimator versus $K$ under Scenario II.}
    \label{plt:plt_scen2}
    \end{figure}

\subsection{Real data application}
We apply the proposed method to re-analyze the Study to Understand Prognoses and Preferences for Outcomes and Risks of Treatments (SUPPORT), as considered in \cite{CuiPuShiMiaoEtAl2023} and \cite{TchetgenTchetgenYingCuiShiEtAl2024}. The study aims to evaluate the effectiveness of right heart catheterization (RHC) in the initial care of critically ill patients \citep{ConnorsSperoffDawsonThomasEtAl1996}. This dataset has also been widely analyzed in the causal inference literature, assuming no unmeasured confounding \citep{HiranoImbens2001,Tan2006,VermeulenVansteelandt2015}. The treatment $A$ indicates whether a patient received an RHC within 24 hours of admission. Among 5735 patients, 2184 received the treatment, while 3551 did not. The outcome $Y$ is the number of days between admission and death or censoring at 30 days. The data include 71 baseline covariates, comprising 21 continuous variables and 50 dummy variables derived from categorical variables. These covariates include demographics (age, sex, race, education, income, and insurance status), estimated probability of survival, comorbidity, vital signs, physiological status, and functional status. Following \cite{TchetgenTchetgenYingCuiShiEtAl2024}, we select $Z = (\mathsf{pafi1}, \mathsf{paco21})$, $W = (\mathsf{ph1}, \mathsf{hema1})$, and use the remaining covariates as $X$.

To implement the proposed method, we select $(1,A,Z,X)$, along with quadratic and cubic polynomials of the continuous variables in $X$, as candidates for constructing moment equations.  We then apply the proposed data-driven approach to select $K$. \Cref{plt:loss_curve} plots the loss curve against $K$, from which we choose $K=81$ for estimation. Our method yields a point estimate of $-1.610$ with a standard error of 0.272, and the corresponding 95\% confidence interval is (-2.143,-1.077). As a comparison, OLS yields an estimate of -1.249 (SE = 0.275) with a 95\% CI of (-1.789,-0.709), while the proximal 2SLS proposed by \cite{TchetgenTchetgenYingCuiShiEtAl2024} produces an estimate of -1.798 (SE = 0.431) with a 95\% CI of (-2.643,-0.954). Our estimate is closely aligned with the proximal 2SLS estimate, suggesting that OLS, which assumes no unmeasured confounding, may underestimate the harmful effect of RHC on 30-day survival among critically ill patients. Moreover, our estimate has a smaller standard error than the proximal 2SLS estimate, resulting in a narrower confidence interval.

\begin{figure}
    \centering
    \includegraphics[width=0.6\textwidth]{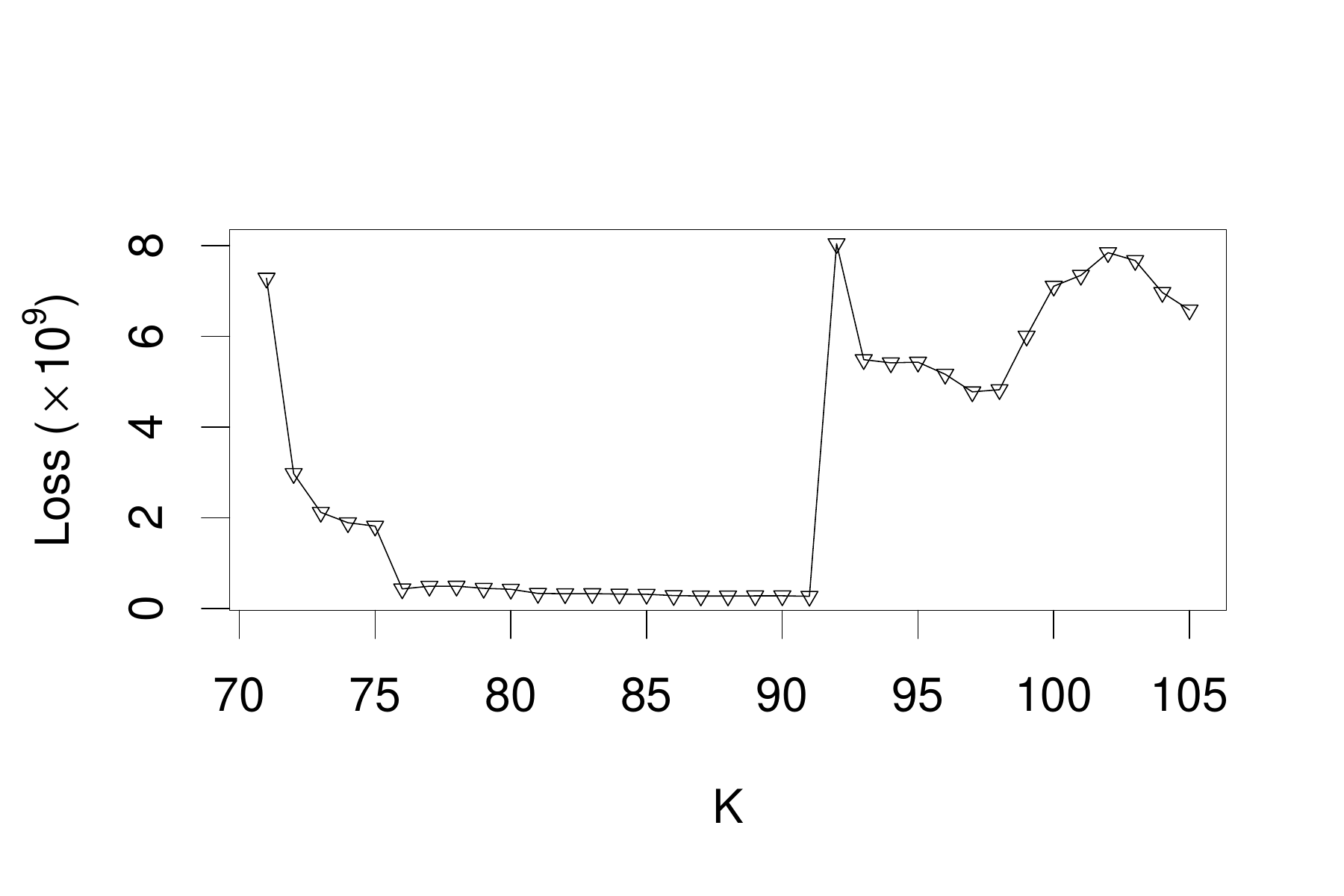}
    \caption{Loss curve for the selection of $K$ in the application of SUPPORT.}
    \label{plt:loss_curve}
\end{figure}

\section{Discussion}\label{sec:discuss}
This paper proposes a simple, data-driven method for obtaining fully efficient estimates of the bridge function and average treatment effect. A key feature of this method is that it requires model assumptions only for the bridge function without imposing any additional assumptions on the data distribution. The proposed estimator typically outperforms the locally efficient DR estimator, demonstrating substantial advantages in detecting significant causal effects. Its feasibility and ease of implementation make it highly useful in practical applications.

Although our method is sufficiently efficient, we acknowledge that it lacks the protective capacity against potential model misspecification that the DR estimator provides.
Both reviewers suggested jointly estimating $h$ and $q$ efficiently, similar to our approach, using the DR estimating equation for the treatment effect. However, as noted by \cite{CuiPuShiMiaoEtAl2023}, this would not improve the efficiency of the treatment effect estimator. The DR estimating equation is Neyman orthogonal and thus locally insensitive to perturbations in $h$ and $q$, so efficiency gains in estimating $h$ and $q$ do not translate into efficiency gains for the ATE estimator. This reflects a trade-off between efficiency and robustness. In practice, we therefore recommend using both methods and comparing their results. For example, when the proposed method and the DR method yield similar estimates, it is reasonable to assume the bridge function is correctly specified and to adopt the proposed method's narrower confidence intervals. \label{resp:discuss}

\section*{Supplementary Materials}
The supplementary material includes intermediate lemmas, additional simulation studies, and all technical proofs.

\appendix
\crefalias{section}{appendix}
\renewcommand{\thetable}{A\arabic{table}}
\setcounter{table}{0}

\section{Regularity conditions}\label{sec:app_regularity_conditions}

\noindent\textbf{Smoothness classes of functions:} Let $\underline{p}$ be the largest integer satisfying $\underline{p}<p$, and $\a=(a_1,\ldots,a_d)$. A function $\phi(\v)$ with domain $\calV\subset\mR^d$ is called a \emph{$p$-smooth function} if it is $\underline{p}$ times continuously differentiable on $\calV$ and 
\begin{align*}
    \max _{\sum_{i=1}^{d} a_i=\underline{p}}
    \left|\frac{\partial^{\underline{p}}\phi(\v)}{\partial v_1^{\alpha_1} \cdots \partial v_d^{\alpha_d}} -
    \frac{\partial^{\underline{p}}\phi(\v^\prime)}{\partial v_1^{\alpha_1} \cdots \partial v_d^{\alpha_d}}
    \right|
     \le C \left\|\v-\v^{\prime}\right\|^{p-\underline{p}},
\end{align*}
for all $\v,\v^\prime\in\mathcal{V}$ and some constant $C>0$.

\noindent\textbf{Regularity conditions:}
The following regularity conditions \ref{cond:compact}--\ref{cond:unif_bound2} are common in the GMM literature \citep[p. 2132]{NeweyMcFadden1994}. Condition \ref{cond:p_smooth} imposes mild smoothness restrictions to ensure a good sieve approximation.

\begin{enumerate}[({A}1)]
    \item $\Gamma$ is a compact subset in $\mR^p$; $\bm\gamma_0$ lies in the interior of $\Gamma$ and is the unique solution to \eqref{eq:iden_gamma_sieve}; $\calT$ is a compact subset in $\mR$, and $\tau_0$ lies in the interior of $\calT$; \label{cond:compact}
    \item $h(W,A,X;\bm\gamma)$ is twice continuously differentiable in $\bm\gamma\in\Gamma$; \label{cond:diff}
    \item $\mE[\sup_{\bm\gamma\in\Gamma}\{Y-h(W,A,X;\bm\gamma)\}^2]<\infty$ and $\mE[\sup_{\bm\gamma\in\Gamma}\|\nabla_{\bm\gamma}h(W,A,X;\bm\gamma)\|^2]<\infty$; \label{cond:unif_bound1}
    \item $\mE[\sup_{\bm\gamma\in\Gamma}\left\{h(W,1,X;\bm\gamma)-h(W,0,X;\bm\gamma)\right\}^2]<\infty$ and $\mE[\sup_{\bm\gamma\in\Gamma}\|\nabla_{\bm\gamma}\{h(W,1,X;\bm\gamma)-h(W,0,X;\bm\gamma)\}\|^2]<\infty$. \label{cond:unif_bound2}
    \item $\mE[\nabla_{\bm\gamma} h(W,A,X;\bm\gamma_0)\mid z,a,x]$, $ \mE\left[\{Y-h(W,A,X)\}^2 \mid z,a,x\right]$ and \\ $\mE[\{Y-h(W,A,X)\}\{h(W,1,X)-h(W,0,X)-\tau_0\}\mid z,a,x]$ are $p$-smooth functions for some $p>0$. \label{cond:p_smooth}
\end{enumerate}

\section{Notation}\label{sec:app_sel_K}

The following notations, adapted slightly from \cite{DonaldImbensNewey2009}, are used in \Cref{sec:implementation} for selecting $K$:
\begin{align*}
    & 
    \wh{\bm\Upsilon}_{K \times K}=\frac{1}{N} \sum_{i=1}^N \{Y_i-h(W_i,A_i,X_i;\check{\bm\gamma})\}^2 \u_K(Z_i,A_i,X_i)\u_K(Z_i,A_i,X_i)\trans, \\
    & \wh\B_{K \times p}=-\frac{1}{N} \sum_{i=1}^N \u_K(Z_i,A_i,X_i) \nabla_{\bm\gamma} h(W_i,A_i,X_i; \check{\bm\gamma})\trans, ~~
    \wh{\bm\Omega}_{p \times p}=(\wh\B_{K \times p})\trans \wh{\bm\Upsilon}_{K \times K}^{-1} \wh\B_{K \times p}, \\
    & \wt{\d}_i=(\wh\B_{K \times p})\trans\bigg\{\frac{1}{N} \sum_{j=1}^N \u_K(Z_j,A_j,X_j)\u_K(Z_j,A_j,X_j)\trans\bigg\}^{-1} \u_K(Z_i,A_i,X_i),\\
    & \wt{\bm\eta}_i=-\nabla_{\bm\gamma} h(W_i,A_i,X_i; \check{\bm\gamma})-\wt{\d}_i,
    ~~ \D_i^*=(\wh\B_{K \times p})\trans \wh{\bm\Upsilon}_{K \times K}^{-1} \u_K(Z_i,A_i,X_i),\\
    &  \wh{\xi}_{i j}= \u_K(Z_i,A_i,X_i)\trans \wh{\bm\Upsilon}_{K \times K}^{-1} \u_K(Z_j,A_j,X_j)/N.
\end{align*}
For a fixed $\t\in\mR^p$,  
\begin{align*}
    \wh{\Pi}(K ; \t)= & \sum_{i=1}^N \wh{\xi}_{i i} \{Y_i-h(W_i,A_i,X_i;\check{\bm\gamma})\} \t\trans \wh{\bm\Omega}_{p \times p}^{-1} \wt{\bm\eta}_i, \\
    \wh{\Phi}(K ; \t)= & \sum_{i=1}^N \wh{\xi}_{i i}\left\{\t\trans \wh{\bm\Omega}_{p \times p}^{-1}\left[\wh{\D}_i^* \{Y_i-h(W_i,A_i,X_i;\check{\bm\gamma})\}^2+\nabla_\gamma h(Z_i,A_i,X_i;\check{\gamma})\right]\right\}^2 
   -\t\trans \wh{\bm\Omega}_{p \times p}^{-1} \t.
\end{align*}
The loss function is
$$S_{\rm{GMM}}(K)=\sum_{j=1}^p\wh\Pi(K;\e_j)^2/N+\wh\Phi(K;\e_j),$$ 
where $\e_j$ is the unit vector with 1 in the $j$-th component and 0 in all others.

In addition, Tables~\ref{table:list_rv}--\ref{table:ass} summarize the random variables, notation, and assumptions used in the paper for ease of reference.

\begin{table}[h!]
  \caption{List of random variables.}
  \label{table:list_rv}
  \centering
  \vspace{-0.2cm}
  \begin{tabular}{ll}
    \toprule
    $A$ & Binary treatment assignment. \\ 
    $Y(a)$ & Potential outcomes under treatment $A=a$. \\
    $Y$ & Observed outcome. \\
    $X$ & Observed confounders. \\
    $U$ & Unobserved confounders. \\
    $(Z, W)$ & Treatment/outcome-inducing confounding proxies. \\
    $\O$ & Observed variables $\O=(A,Y,W,X,Z)$. \\
    \bottomrule
  \end{tabular}
\end{table}

\begin{table}[h!]
  \caption{List of notation.}
  \label{table:list_notation}
  \centering
  \vspace{-0.2cm}
  \begin{tabular}{ll}
    \toprule
    $h(w,a,x)$ & Outcome-confounding bridge function; see Eq. \eqref{eq:iden_bridge_h}. \\
    $q(z,a,x)$ & Treatment-confounding bridge function; see Eq. \eqref{eq:iden_bridge_q}. \\
    $\u_{K}(z,a,x)$ & Vector of basis functions; see Eq. \eqref{eq:iden_gamma_sieve}. \\
    $\m_\eff(z,a,x)$ & Efficient score for estimating $h$; see Eq. \eqref{eq:ee_gamma}. \\
    $\g_K(\O;\gamma,\tau)$ & Joint score function for estimating $h$ and $\tau_0$; see \Cref{sec:estimation}. \\
    $\G_K(\gamma,\tau)$ & Sample average of $\g_K(\O;\gamma,\tau)$; see \Cref{sec:estimation}. \\
    ${\bm\Upsilon}_{(K+1)\times(K+1)}$ & Optimal weight in GMM estimation; see \Cref{sec:estimation}. \\
    $\B_{(K+1)\times(p+1)}$ & Jacobian matrix; see Eq. \eqref{eq:GMM_normal}. \\
    $\{\bm\psi_1(\O),\psi_2(\O)\}$ & Influence functions of $(\wh{\bm\gamma},\wh\tau)$; see \Cref{thm:asym_normal}. \\
    $\psi_\eff(\O)$ & Efficient influence function of $\tau_0$; see Eq. \eqref{eq:eif}. \\
    $\{t(\cdot),R(\cdot),\bm\kappa\}$ & Terms in the influence function; see \Cref{thm:asym_normal}. \\
    $\V_K$ & Asymptotic variance of the GMM estimator; see Eq. \eqref{eq:GMM_normal}. \\
    $(\V_{\bm\gamma}, V_\tau)$ & Asymptotic variance of $(\wh{\bm\gamma}, \wh\tau)$; see \Cref{thm:asym_normal}. \\
    $V_{\tau,\eff}$ & Semiparametric efficiency bound for $\tau_0$; see Eq. \eqref{eq:eif}. \\
    $V^\plugin_{\tau,\eff}$ & Variance of the optimal plug-in estimator; see \Cref{coro:super_plug_in}. \\
    \bottomrule
  \end{tabular}
\end{table}

\begin{table}[h!]
  \caption{Summary of assumptions.}
  \label{table:ass}
  \centering
  \vspace{-0.2cm}
  \begin{tabular}{ll}
    \toprule
    \Cref{ass:causal_and_proxy} &  Conditions for identifying treatment effects using proxies. \\ 
    \Cref{ass:sieve} & Conditions for the use of sieve techniques. \\
    \Cref{ass:proxy_cui} &  Conditions for developing semiparametric theory. \\
    Conditions \ref{cond:compact}-\ref{cond:p_smooth} &  Conditions for the asymptotic analysis of the GMM estimator. \\
    \bottomrule
  \end{tabular}
\end{table}

\clearpage
\bibliographystyle{agsm}
\bibliography{reference}

\clearpage
\appendix
\setcounter{page}{1}
\renewcommand{\thepage}{S\arabic{page}}
{\centering
\section*{SUPPLEMENT}}
\setcounter{equation}{0}\renewcommand{\theequation}{S.\arabic{equation}}
\setcounter{subsection}{0}\renewcommand{\thesubsection}{S.\arabic{subsection}}
\renewcommand{\thetheorem}{S\arabic{theorem}}
\renewcommand{\thelemma}{S\arabic{lemma}}
\renewcommand{\theremark}{S\arabic{remark}}
\renewcommand{\theequation}{S\arabic{equation}}
\renewcommand{\thesection}{S\arabic{section}}
\renewcommand{\thetable}{S\arabic{table}}

\section{Preliminaries}\label{sec:app_pre}

\noindent\textbf{Notations:} Unless otherwise specified, the norm $\|\cdot\|$ refers to the Euclidean/Frobenius norm for a vector/matrix.

\noindent\textbf{Inverse of a $2\times2$ block matrix:}
Let a $(K+1) \times (K+1)$ matrix $\M$ be partitioned into a block form
\beqrs
\M=\left(\begin{array}{ll}
    \A & \b \\
    \b\trans  & d
    \end{array}\right),
\eeqrs
where $\A$ is a $K\times K$ matrix, $\b$ is a $K$ dimensional column vector, $d$ is a constant. Then, Theorem 2.1 of \citet{LuShiou2002} shows that
\beqr\label{eq:mat_inv}
 \M^{-1}=\left(\begin{array}{cc}
    \A^{-1}+\frac{1}{k} \A^{-1} \b \b\trans \A^{-1} & -\frac{1}{k} \A^{-1} \b \\
    -\frac{1}{k} \b\trans \A^{-1} & \frac{1}{k}
    \end{array}\right),
\eeqr
where $k=d-\b\trans \A^{-1} \b$ .

\section{Intermediate lemmas} \label{sec:app_intermediate_result}
We first prove some intermediate results given in the following lemmas, which are useful for the proof of main theorems.

\begin{lemma}\label{lemma:app_ini_consistency}
    Suppose \Cref{ass:sieve} and regularity conditions \ref{cond:compact}--\ref{cond:unif_bound2} in \Cref{sec:app_regularity_conditions} hold. We have 
    $\|(\check{\bm\gamma},\check\tau)-(\bm\gamma_0,\tau_0)\|\convp 0$.
\end{lemma} 

\begin{proof}
    Recall that $(\check{\bm\gamma},\check\tau)=\amin_{(\bm\gamma,\tau)\in\Gamma\times\calT}\wh{Q}_N(\bm\gamma, \tau)$ with
    \begin{align}\label{eq:app_empirical_criterion}
    \wh{Q}_N(\bm\gamma, \tau)=\G_K(\bm\gamma,\tau)\trans \wh{\bm\Theta}_{(K+1)\times(K+1)}^{-1}\G_K(\bm\gamma,\tau).
    \end{align}
    where 
    \begin{align*}
        \wh{\bm\Theta}_{(K+1)\times(K+1)} = \left(
            \begin{array}{cc}
                \frac{1}{N}\sum_{i=1}^{N}\u_K(Z_i,A_i,X_i)\u_K(Z_i,A_i,X_i)\trans & \0\\
                \0\trans & 1
            \end{array}
            \right).
    \end{align*}
    Note that if the basis functions are orthonormalized, then $\wh{\bm\Theta}_{(K+1)\times(K+1)}=I_{K+1}$ as indicated in \Cref{remark:moment_sel}. We maintain this general notation here for the sake of clarity.
    Meanwhile, the true value $(\bm\gamma,\tau)=\amin_{(\bm\gamma,\tau)\in\Gamma\times\calT}Q(\bm\gamma, \tau)$ with
    \begin{align*}
    Q(\bm\gamma, \tau)=\mE\{\g_K(\O;\bm\gamma,\tau)\}\trans {\bm\Theta}_{(K+1)\times(K+1)}^{-1} \mE\{\g_K(\O;\bm\gamma,\tau)\},
    \end{align*}
    where
    \begin{align*}
        {\bm\Theta}_{(K+1)\times(K+1)} = \left(
            \begin{array}{cc}
                \mE\{\u_K(Z,A,X)\u_K(Z,A,X)\trans\} & \0\\
                \0\trans & 1
            \end{array}
            \right).
    \end{align*}
    In what follows, we will show that
    \begin{align}\label{eq:app_consis_0}
        \sup_{(\bm\gamma,\tau)\in\Gamma\times\calT}\left|\wh{Q}_N(\bm\gamma,\tau)-Q(\bm\gamma,\tau)\right| \convp 0.
    \end{align}
    Then by the general theory about consistency of M-estimators \citep[][p. 45]{Vaart1998}, the consistency of $(\check{\bm\gamma},\check\tau)$ follows immediately.
    By the triangle inequality,
\begin{align}
    & \left|\wh{Q}_N(\bm\gamma, \tau)-Q(\bm\gamma, \tau)\right| \nonumber\\
    \le & \left|\left[\G_K(\bm\gamma,\tau)-\mE\{\g_K(\O;\bm\gamma,\tau)\}\right]\trans 
    \wh{\bm\Theta}_{(K+1)\times(K+1)}^{-1}
    \left[\G_K(\bm\gamma,\tau)-\mE\{\g_K(\O;\bm\gamma,\tau)\}\right]\right| \label{eq:app_consis_T1}\\
    & +\left|\mE\{\g_K(\O;\bm\gamma,\tau)\}\trans\left\{\wh{\bm\Theta}_{(K+1)\times(K+1)}^{-1}-{\bm\Theta}_{(K+1)\times(K+1)}^{-1}\right\}\mE\{\g_K(\O;\bm\gamma,\tau)\}\right| \label{eq:app_consis_T2}\\
    & +2\left|\mE\{\g_K(\O;\bm\gamma,\tau)\}\trans {\bm\Theta}_{(K+1)\times(K+1)}^{-1}\left[\G_K(\bm\gamma,\tau)-\mE\{\g_K(\O;\bm\gamma,\tau)\}\right]\right|.\label{eq:app_consis_T3}
\end{align}
Next we show \eqref{eq:app_consis_T1}--\eqref{eq:app_consis_T3} are all $o_p(1)$. Consider the term \eqref{eq:app_consis_T1}. Let
$\lambda_{\max}(\bm{A})$ (resp. $\lambda_{\min}(\bm{A})$) denote the maximum (resp. minimum) eigenvalue of a matrix $\bm{A}$. Then we have 
\begin{align}\label{eq:app_consis_T1.1}
|\eqref{eq:app_consis_T1}|^2
\le&~ \|\G_K(\bm\gamma,\tau)-\mE\{\g_K(\O;\bm\gamma,\tau)\}\|^2\left\|\wh{\bm\Theta}_{(K+1)\times(K+1)}^{-1}\left[\G_K(\bm\gamma,\tau)-\mE\{\g_K(\O;\bm\gamma,\tau)\}\right]\right\|^2\nonumber\\ 
\le&~ \lambda_{\max}\left(\wh{\bm\Theta}_{(K+1)\times(K+1)}^{-2}\right)\|\G_K(\bm\gamma,\tau)-\mE\{\g_K(\O;\bm\gamma,\tau)\}\|^4,
\end{align}
where the last inequality holds due to the inequality of the Frobenius norm $\|\A\B\|^2\le\lambda_{\max}(\A\trans\A)\|\B\|^2$ for matrices $\A,\B$.
Note that 
\begin{align*}
    \lambda_{\min}\left(\wh{\bm\Theta}_{(K+1)\times(K+1)}\right)
    =&~ \min\left[\lambda_{\min}\left\{\frac{1}{N}\sum_{i=1}^{N}\u_K(Z_i,A_i,X_i)\u_K(Z_i,A_i,X_i)\trans\right\},1\right],
\end{align*}
and 
\begin{align*}
    &\lambda_{\min}\left\{\frac{1}{N}\sum_{i=1}^{N}\u_K(Z_i,A_i,X_i)\u_K(Z_i,A_i,X_i)\trans\right\} \\ 
    =&~  \min_{\|\nu\|=1} \nu\trans\left\{\frac{1}{N}\sum_{i=1}^{N}\u_K(Z_i,A_i,X_i)\u_K(Z_i,A_i,X_i)\trans\right\} \nu \\ 
    =&~ \min_{\|\nu\|=1} \nu\trans\mE\left\{\u_K(Z,A,X)\u_K(Z,A,X)\trans\right\} \nu \\ 
    &+ \min_{\|\nu\|=1} \nu\trans\left\{\frac{1}{N}\sum_{i=1}^{N}\u_K(Z_i,A_i,X_i)\u_K(Z_i,A_i,X_i)\trans-\mE\left\{\u_K(Z,A,X)\u_K(Z,A,X)\trans\right\}\right\} \nu \\ 
    \ge&~ \lambda_{\min}\left(\mE\left\{\u_K(Z,A,X)\u_K(Z,A,X)\trans\right\}\right) \\
    &- \left\|\frac{1}{N}\sum_{i=1}^{N}\u_K(Z_i,A_i,X_i)\u_K(Z_i,A_i,X_i)\trans-\mE\left\{\u_K(Z,A,X)\u_K(Z,A,X)\trans\right\}\right\| \\ 
    =&~  \underline{c}-o_p(1),
\end{align*}
for some consistent $\underline{c}>0$,
where the last equality is by \Cref{ass:sieve}(i) and 
\begin{align}\label{eq:app_consis_T1.2.1}
    & \mE\left\{\left\|\frac{1}{N}\sum_{i=1}^{N}\u_K(Z_i,A_i,X_i)\u_K(Z_i,A_i,X_i)\trans-\mE\left\{\u_K(Z,A,X)\u_K(Z,A,X)\trans\right\}\right\|^2\right\} \nonumber\\ 
    =&~ \frac{1}{N} \mE\left\{\Big\|\u_K(Z,A,X)\u_K(Z,A,X)\trans-\mE\left\{\u_K(Z,A,X)\u_K(Z,A,X)\trans\right\}\Big\|^2\right\} \nonumber\\
    \le&~ \frac{1}{N} \mE\left\{\Big\|\u_K(Z,A,X)\u_K(Z,A,X)\trans\Big\|^2\right\}\nonumber\\ 
    =&~ \frac{1}{N} \mE\{\u_K(Z,A,X)\trans\u_K(Z,A,X)\u_K(Z,A,X)\trans\u_K(Z,A,X)\}\nonumber\\ 
    \le&~ \frac{1}{N}\sup_{(z,a,x)\in\calZ\times\{0,1\}\times\calX} \|\u_K(z,a,x)\|^2 \cdot\mE\left\{\|\u_K(Z,A,X)\|^2\right\} \nonumber\\ 
    \le&~\frac{1}{N}\sup_{(z,a,x)\in\calZ\times\{0,1\}\times\calX} \|\u_K(z,a,x)\|^2 \cdot K\lambda_{\max}\left(\mE\left\{\u_K(Z,A,X)\u_K(Z,A,X)\trans\right\}\right) \nonumber\\ 
    =&~\frac{1}{N}\cdot O\left\{\zeta(K)^2\right\}\cdot K \cdot O(1) = O\left\{\zeta(K)^2K/N\right\}=o(1).
\end{align}
Then we have 
\begin{align}\label{eq:app_consis_T1.2}
    \lambda_{\max}\left(\wh{\bm\Theta}_{(K+1)\times(K+1)}^{-2}\right) 
    = \lambda_{\min}\left(\wh{\bm\Theta}_{(K+1)\times(K+1)}\right)^{-2} = O_p(1).
\end{align}
Note that
\begin{align}\label{eq:app_consis_T1.3}
& E\left\{\|\G_K(\bm\gamma,\tau)-\mE\{\g_K(\O;\bm\gamma,\tau)\}\|^2\right\} \nonumber\\ 
=&~  \mE\left\{\left\|\frac{1}{N}\sum_{i=1}^N \g_K(\O_i;\bm\gamma,\tau)-\mE[\g_K(\O;\bm\gamma,\tau)]\right\|^2\right\}\nonumber\\ 
=&~ \frac{1}{N} \mE\left\{\left\| \g_K(\O;\bm\gamma,\tau)-\mE[\g_K(\O;\bm\gamma,\tau)]\right\|^2\right\} \nonumber\\ 
\le&~ \frac{1}{N} \mE\left\{\left\| \g_K(\O;\bm\gamma,\tau)\right\|^2\right\} \nonumber\\
=&~ \frac{1}{N} \mE\left[\{Y-h(W,A,X;\bm\gamma)\}^2\left\|\u_K(Z,A,X)\right\|^2\right]\nonumber\\ 
&~ +\frac{1}{N} \mE\left[\left\{\tau-h(W,1,X;\bm\gamma)+h(W,0,X;\bm\gamma)\right\}^2\right]\nonumber\\ 
=&~ \frac{1}{N}\sup_{(z,a,x)\in\calZ\times\{0,1\}\times\calX} \|\u_K(z,a,x)\|^2 \cdot \mE\left[\{Y-h(W,A,X;\bm\gamma)\}^2\right]+O(1/N)\nonumber\\ 
=&~  \frac{1}{N}\cdot O\left\{\zeta(K)^2\right\} O(1) + O(1/N) = O\left\{\zeta(K)^2/N\right\}=o(1).
\end{align}
By Chebyshev's inequality, we have $\|\G_K(\bm\gamma,\tau)-\mE\{\g_K(\O;\bm\gamma,\tau)\}\|=o_p(1)$. 
Combining \eqref{eq:app_consis_T1.1}, \eqref{eq:app_consis_T1.2} and \eqref{eq:app_consis_T1.3} yields that $\eqref{eq:app_consis_T1}=o_p(1)$. Next, we consider the term \eqref{eq:app_consis_T2}. We have 
\begin{align*} 
    |\eqref{eq:app_consis_T2}|^2 
    \le \left\|\wh{\bm\Theta}_{(K+1)\times(K+1)}^{-1}-{\bm\Theta}_{(K+1)\times(K+1)}^{-1}\right\|^2\|\mE\{\g_K(\bm\gamma,\tau)\}\|^4.
\end{align*}
Note that 
\begin{align}\label{eq_app_W_inverse_rate}
    & \left\|\wh{\bm\Theta}_{(K+1)\times(K+1)}^{-1}-{\bm\Theta}_{(K+1)\times(K+1)}^{-1}\right\|^2 \nonumber\\
    =& \left\|\wh{\bm\Theta}_{(K+1)\times(K+1)}^{-1}\left\{\wh{\bm\Theta}_{(K+1)\times(K+1)}-{\bm\Theta}_{(K+1)\times(K+1)}\right\}{\bm\Theta}_{(K+1)\times(K+1)}^{-1}\right\|^2\nonumber\\ 
    \le& \lambda_{\max}\left(\wh{\bm\Theta}_{(K+1)\times(K+1)}^{-2}\right) 
    \left\|\left\{\wh{\bm\Theta}_{(K+1)\times(K+1)}-{\bm\Theta}_{(K+1)\times(K+1)}\right\}{\bm\Theta}_{(K+1)\times(K+1)}^{-1}\right\|^2 \nonumber\\
    \le& \lambda_{\max}\left(\wh{\bm\Theta}_{(K+1)\times(K+1)}^{-2}\right) 
    \lambda_{\max}\left({\bm\Theta}_{(K+1)\times(K+1)}^{-2}\right) 
    \left\| \wh{\bm\Theta}_{(K+1)\times(K+1)}-{\bm\Theta}_{(K+1)\times(K+1)} \right\|^2 \nonumber\\ 
    =&~ O_p(1)\cdot O(1) \cdot O_p\left\{\zeta(K)^2K/N\right\}=o_p(1).
\end{align}
where the last equality follows \Cref{ass:sieve}(i), \eqref{eq:app_consis_T1.2.1} and \eqref{eq:app_consis_T1.2}. Moreover,
    \begin{align*}
        & \|\mE\{\g_K(\bm\gamma,\tau)\}\|^2 \\
        =&~ \mE[\{Y-h(W,A,X;\bm\gamma)\}\u_K(Z,A,X)]\trans\mE[\{Y-h(W,A,X;\bm\gamma)\}\u_K(Z,A,X)] \\ 
        &~ +  \left|\mE\{\tau-h(W,1,X;\bm\gamma)+h(W,0,X;\bm\gamma)\}\right|^2 \\ 
        \le&~ \lambda_{\max}\{\u_K(Z,A,X)\u_K(Z,A,X)\trans\} \mE[\{Y-h(W,A,X;\bm\gamma)\}\u_K(Z,A,X)]\trans \\ 
        &\qquad \times \mE\{\u_K(Z,A,X)\u_K(Z,A,X)\trans\}^{-1}\mE[\{Y-h(W,A,X;\bm\gamma)\}\u_K(Z,A,X)] +O(1) \\ 
        \le&~\lambda_{\max}\{\u_K(Z,A,X)\u_K(Z,A,X)\trans\} \mE[\{Y-h(W,A,X;\bm\gamma)\}^2]+O(1)\\ 
        =&~  O(1), 
    \end{align*} 
where the second inequality is because 
\begin{align*}
    &\mE[\{Y-h(W,A,X;\bm\gamma)\}\u_K(Z,A,X)]\trans  \mE\{\u_K(Z,A,X)\u_K(Z,A,X)\trans\}^{-1} \\ 
    &\qquad \times \mE[\{Y-h(W,A,X;\bm\gamma)\}\u_K(Z,A,X)] \\ 
    =&~ \mE\left|\mE[\{Y-h(W,A,X;\bm\gamma)\}\u_K(Z,A,X)]\trans  \mE\{\u_K(Z,A,X)\u_K(Z,A,X)\trans\}^{-1}\u_K(Z,A,X)\right|^2\\ 
    =:&~ \mE\left\{a_K(Z,A,X)^2\right\}, 
\end{align*}
and $a_K(Z,A,X)$ is the least squared projection of $Y-h(W,A,X;\bm\gamma)$ onto the space linear spanned by $\{\u_K(Z,A,X)\}$.
It follows that $\eqref{eq:app_consis_T2}=o_p(1)$. We can also show $\eqref{eq:app_consis_T3}=o_p(1)$ following a very similar procedure as in calculating \eqref{eq:app_consis_T1} and \eqref{eq:app_consis_T2}.  Therefore, for each $(\bm\gamma,\tau)\in\Gamma\times\calT$, we have $|\wh{Q}_N(\bm\gamma,\tau)-Q(\bm\gamma,\tau)|=o_p(1)$.

We next derive the uniform convergence. Note that 
\begin{small}
    \begin{align}
    & \frac{1}{2} \left\|\nabla_{\bm\gamma,\tau} \wh{Q}_N(\bm\gamma, \tau)\right\|^2  \nonumber\\ 
    =&~ 
    \left\|\left(\begin{array}{cc}
        -\frac{1}{N}\sum_{i=1}^{N} \u_K(Z_i,A_i,X_i)\nabla_{\bm\gamma} h(W_i,A_i,X_i;\bm\gamma)\trans & \0 \nonumber\\
        -\frac{1}{N}\sum_{i=1}^{N} \nabla_{\bm\gamma}\left\{ h(W_i,1,X_i;\bm\gamma) - h(W_i,0,X_i;\bm\gamma)\right\}\trans & 1
        \end{array}\right)\trans
    \wh{\bm\Theta}_{(K+1)\times(K+1)}^{-1} \G_K(\bm\gamma,\tau)\right\|^2 \nonumber\\ 
    =&~ \Bigg\|\left\{\frac{1}{N}\sum_{i=1}^{N} \nabla_{\bm\gamma} h(W_i,A_i,X_i;\bm\gamma)\u_K(Z_i,A_i,X_i)\trans\right\}\left\{\frac{1}{N}\sum_{i=1}^{N} \u_K(Z_i,A_i,X_i)\u_K(Z_i,A_i,X_i)\trans\right\}^{-1} \nonumber\\ 
    &\qquad\times \left\{\frac{1}{N}\sum_{i=1}^{N} \{Y-h(W_i,A_i,X_i;\bm\gamma)\}\u_K(Z_i,A_i,X_i)\right\}\Bigg\|^2 \label{eq:app_consis_deriv_T1}\\
    &~  +\left\|\left\{\frac{1}{N}\sum_{i=1}^{N} \nabla_{\bm\gamma}\left\{ h(W_i,1,X_i;\bm\gamma) - h(W_i,0,X_i;\bm\gamma)\right\}\right\}\left\{\frac{1}{N}\sum_{i=1}^{N}\{\tau - h(W_i,1,X_i;\bm\gamma)+h(W_i,0,X_i;\bm\gamma)\}\right\}\right\|^2 \nonumber\\ 
    &~ + \left|\frac{1}{N}\sum_{i=1}^{N}\{\tau - h(W_i,1,X_i;\bm\gamma)+h(W_i,0,X_i;\bm\gamma)\}\right|^2.\nonumber
    \end{align}
\end{small}
Consider the term \eqref{eq:app_consis_deriv_T1}. We denote 
    \begin{align}\label{eq:app_consis_deriv_T1.1}
    \wh\Xi = \left\{\frac{1}{N}\sum_{i=1}^{N} \u_K(Z_i,A_i,X_i)\u_K(Z_i,A_i,X_i)\trans\right\}^{-1}\left\{\frac{1}{N}\sum_{i=1}^{N} \u_K(Z_i,A_i,X_i)\nabla_{\bm\gamma} h(W_i,A_i,X_i;\bm\gamma)\trans\right\},
    \end{align}
Then  $\wh\Xi\trans\u_K(z,a,x)$ is the $\ell_2$projection  $\nabla_{\bm\gamma} h(w,a,x;\bm\gamma)$ onto the space linearly spanned by $\{\u_K(z,a,x)\}$. Thus we have
\begin{align*}
\frac{1}{N}\sum_{i=1}^{N} \left\|\wh\Xi\trans\u_K(Z_i,A_i,X_i)\right\|^2 \le \frac{1}{N}\sum_{i=1}^{N} \left\|\nabla_{\bm\gamma} h(W_i,A_i,X_i;\bm\gamma)\right\|^2.
\end{align*}
Then, by Cauchy-Schwarz inequality, we have 
\begin{align*}
    \eqref{eq:app_consis_deriv_T1} =&~  \left\|\frac{1}{N}\sum_{i=1}^{N} \{Y-h(W_i,A_i,X_i;\bm\gamma)\}\wh\Xi\trans\u_K(Z_i,A_i,X_i)\right\|^2 \\ 
    &~ \le \frac{1}{N}\sum_{i=1}^{N} \{Y-h(W_i,A_i,X_i;\bm\gamma)\}^2\cdot \frac{1}{N}\sum_{i=1}^{N} \left\|\wh\Xi\trans\u_K(Z_i,A_i,X_i)\right\|^2 \\ 
    &~ \le \frac{1}{N}\sum_{i=1}^{N} \{Y-h(W_i,A_i,X_i;\bm\gamma)\}^2\cdot \frac{1}{N}\sum_{i=1}^{N} \left\|\nabla_{\bm\gamma} h(W_i,A_i,X_i;\bm\gamma)\right\|^2 .
\end{align*}
By the uniform laws of large numbers, with regularity conditions \ref{cond:compact}--\ref{cond:unif_bound2}, Lemma 2.4 of \cite{NeweyMcFadden1994} implies that $\sup_{\bm\gamma\in\Gamma}\eqref{eq:app_consis_deriv_T1}=O_p(1)$. With some similar arguments, we finally obtain $\sup_{(\bm\gamma,\tau)\in\Gamma\times\calT}\|\nabla_{\bm\gamma,\tau} \wh{Q}_N(\bm\gamma, \tau)\|=O_p(1)$. Similarly, we can obtain $\sup_{(\bm\gamma,\tau)\in\Gamma\times\calT}\|\nabla_{\bm\gamma,\tau} Q(\bm\gamma, \tau)\|=O_p(1)$. Then the uniform convergence \eqref{eq:app_consis_0} follows immediately from Corollary 2.2 of \cite{Newey1991}, indicating consistency of $(\check{\bm\gamma},\check\tau)$.
\end{proof}

\begin{lemma}\label{lemma:app_ini_rate}
    Suppose \Cref{ass:sieve} and regularity conditions \ref{cond:compact}--\ref{cond:unif_bound2} in \Cref{sec:app_regularity_conditions} hold. We have 
    $\|(\check{\bm\gamma},\check\tau)-(\bm\gamma_0,\tau_0)\|=O_p(N^{-1/2})$.
\end{lemma}

\begin{proof}
Recall that $(\check{\bm\gamma},\check\tau)$ minimizes \eqref{eq:app_empirical_criterion}, so it should satisfy the first-order condition
\begin{align*}
\nabla_{\bm\gamma,\tau} \G_K(\check{\bm\gamma},\check\tau) \wh{\bm\Theta}_{(K+1)\times(K+1)}^{-1}\G_K(\check{\bm\gamma},\check\tau)=\0.
\end{align*}
By the mean value theorem, we have 
\begin{align}
    \0=&~ \nabla_{\bm\gamma,\tau} \G_K(\bm\gamma_0,\tau_0) \wh{\bm\Theta}_{(K+1)\times(K+1)}^{-1}\sqrt{N}\G_K(\bm\gamma_0,\tau_0) \nonumber\\
    &~ + \nabla_{\bm\gamma,\tau} \G_K(\wt{\bm\gamma},\wt\tau) \wh{\bm\Theta}_{(K+1)\times(K+1)}^{-1}\nabla_{\bm\gamma,\tau}\G_K(\wt{\bm\gamma},\wt\tau)\trans \binom{\sqrt{N}(\check{\bm\gamma}-\bm\gamma_0)}{\sqrt{N}(\check{\tau}-\tau_0)} \nonumber\\ 
    &~ + \sqrt{N}\G_K(\wt{\bm\gamma},\wt\tau)\trans \wh{\bm\Theta}_{(K+1)\times(K+1)}^{-1} \nonumber\\
    &\qquad \times \left[\frac{\partial^2 \G_K(\wt{\bm\gamma},\wt\tau)}{\partial \gamma_1\partial\bm\gamma\trans}(\check{\bm\gamma}-\bm\gamma_0),
    \ldots,
    \frac{\partial^2 \G_K(\wt{\bm\gamma},\wt\tau)}{\partial \gamma_p\partial\bm\gamma\trans}(\check{\bm\gamma}-\bm\gamma_0),
    \0_{K+1}\right].\label{eq:app_rootn_T0}
\end{align}
where $(\wt{\bm\gamma},\wt\tau)$ lies between $(\check{\bm\gamma},\check\tau)$ and $(\bm\gamma_0,\tau_0)$. We only need to show 
\begin{align}
    & \nabla_{\bm\gamma,\tau} \G_K(\bm\gamma_0,\tau_0) \wh{\bm\Theta}_{(K+1)\times(K+1)}^{-1}\sqrt{N}\G_K(\bm\gamma_0,\tau_0) = O_p(1),\label{eq:app_rootn_T1}\\
    & \nabla_{\bm\gamma,\tau} \G_K(\wt{\bm\gamma},\wt\tau) \wh{\bm\Theta}_{(K+1)\times(K+1)}^{-1}\nabla_{\bm\gamma,\tau}\G_K(\wt{\bm\gamma},\wt\tau)\trans = O_p(1),\label{eq:app_rootn_T2}\\
    & \sqrt{N}\G_K(\wt{\bm\gamma},\wt\tau)\trans \wh{\bm\Theta}_{(K+1)\times(K+1)}^{-1} \nonumber\\
    &\qquad \times \left[\frac{\partial^2 \G_K(\wt{\bm\gamma},\wt\tau)}{\partial \gamma_1\partial\bm\gamma\trans}(\check{\bm\gamma}-\bm\gamma_0),
    \ldots,
    \frac{\partial^2 \G_K(\wt{\bm\gamma},\wt\tau)}{\partial \gamma_p\partial\bm\gamma\trans}(\check{\bm\gamma}-\bm\gamma_0),
    \0_{K+1}\right] = o_p(1).\label{eq:app_rootn_T3}
\end{align}
To show \eqref{eq:app_rootn_T1}, note that by computing the second  moment and using Chebyshev's inequality, we have 
\begin{align} 
& \|\G_K(\bm\gamma_0,\tau_0)\| = O\left(\sqrt{K/N}\right),\label{eq:app_rootn_T1.1}\\
& \|\nabla_{\bm\gamma,\tau} \G_K(\bm\gamma_0,\tau_0) - \B_{(K+1)\times(p+1)}\trans\| = O\left(\sqrt{K/N}\right),\label{eq:app_rootn_T1.2} \\
& \left\|\frac{\partial^2 \G_K(\bm\gamma_0,\tau_0)}{\partial \gamma_j\partial\bm\gamma\trans}
 - \mE\left\{\frac{\partial^2 \g_K(\O;\bm\gamma_0,\tau_0)}{\partial \gamma_j\partial\bm\gamma\trans}\right\}\right\| = O\left(\sqrt{K/N}\right),~j=1,\ldots,p.\label{eq:app_rootn_T1.3.0}
\end{align}
Combining \eqref{eq_app_W_inverse_rate}, \eqref{eq:app_rootn_T1.1}, \eqref{eq:app_rootn_T1.2} and \Cref{ass:sieve}(i), we have 
\begin{align*}
    &\nabla_{\bm\gamma,\tau} \G_K(\bm\gamma_0,\tau_0) \wh{\bm\Theta}_{(K+1)\times(K+1)}^{-1}\sqrt{N}\G_K(\bm\gamma_0,\tau_0) \\ 
    &~ = \B_{(K+1)\times(p+1)}\trans {\bm\Theta}_{(K+1)\times(K+1)}^{-1}\sqrt{N}\G_K(\bm\gamma_0,\tau_0) + o_p(1).
\end{align*}
Computing the variance of $\B_{(K+1)\times(p+1)}\trans {\bm\Theta}_{(K+1)\times(K+1)}^{-1}\sqrt{N}\G_K(\bm\gamma_0,\tau_0)$ yields that 
    \begin{align}
    & \left\|\var\left\{\B_{(K+1)\times(p+1)}\trans {\bm\Theta}_{(K+1)\times(K+1)}^{-1}\sqrt{N}\G_K(\bm\gamma_0,\tau_0)\right\}\right\|  \nonumber\\ 
    =&~ \left\|\B_{(K+1)\times(p+1)}\trans {\bm\Theta}_{(K+1)\times(K+1)}^{-1}\mE\{\g_K(\O,\bm\gamma_0,\tau_0)\g_K(\O,\bm\gamma_0,\tau_0)\trans\}  {\bm\Theta}_{(K+1)\times(K+1)}^{-1}\B_{(K+1)\times(p+1)}\right\| \nonumber\\
    =&~ \mE\left[\left\|\B_{(K+1)\times(p+1)}\trans {\bm\Theta}_{(K+1)\times(K+1)}^{-1}\g_K(\O,\bm\gamma_0,\tau_0)\right\|^2\right]\nonumber\\
    =&~ 
    \mE\left[\left\|\left(
        \begin{array}{cl}
            -\mE\{\u_K(Z,A,X)\nabla_{\bm\gamma} h(W,A,X;\bm\gamma_0)\trans\}  & \0\\
            \nabla_{\bm\gamma}\left\{ h(W,1,X;\bm\gamma_0) - h(W,0,X;\bm\gamma_0)\right\}  & 1
        \end{array}
    \right)\trans
    {\bm\Theta}_{(K+1)\times(K+1)}^{-1} \g_K(\O,\bm\gamma_0,\tau_0)\right\|^2\right] \nonumber\\ 
    =&~ \mE\Bigg[\bigg\|\mE \left\{\nabla_{\bm\gamma} h(W,A,X;\bm\gamma_0)\u_K(Z,A,X)\trans\right\}\left\{\mE \u_K(Z,A,X)\u_K(Z,A,X)\trans\right\}^{-1} \nonumber\\ 
    &\qquad\times \{Y-h(W,A,X;\bm\gamma_0)\}\u_K(Z,A,X)\bigg\|^2\Bigg]  \label{eq:app_rootn_T1.3}\\
    &~ + \left\{\big\|\mE [\nabla_{\bm\gamma}\left\{ h(W,1,X;\bm\gamma_0) - h(W,0,X;\bm\gamma_0)\right\}]\big\|^2 +1\right\} \mE\left[\{\tau_0 - h(W,1,X;\bm\gamma_0)+h(W,0,X;\bm\gamma_0)\}^2\right].\nonumber
    \end{align}
Consider the term \eqref{eq:app_rootn_T1.3}. We denote 
    \begin{align*}
    \Xi =  \mE \left\{\nabla_{\bm\gamma} h(W,A,X;\bm\gamma_0)\u_K(Z,A,X)\trans\right\}\left\{\mE \u_K(Z,A,X)\u_K(Z,A,X)\trans\right\}^{-1},
    \end{align*}
Then  $\Xi\trans\u_K(z,a,x)$ is the $L_2$projection  $\nabla_{\bm\gamma} h(w,a,x;\bm\gamma_0)$ onto the space linearly spanned by $\{\u_K(z,a,x)\}$. Thus we have
\begin{align*}
    \mE \left[\left\|\Xi\trans\u_K(Z,A,X)\right\|^2\right] \le \mE \left[\left\|\nabla_{\bm\gamma_0} h(W,A,X;\bm\gamma_0)\right\|^2 \right].
\end{align*}
Then by Cauchy-Schwarz inequality and the regularity condition \ref{cond:unif_bound1}, we have 
\begin{align*}
    \eqref{eq:app_rootn_T1.3} =&~  \mE\left[\left\|Y-h(W,A,X;\bm\gamma_0)\Xi\trans\u_K(Z,A,X)\right\|^2\right] \\ 
    &~ \le \mE \left[\{Y-h(W,A,X;\bm\gamma_0)\}^2\right]\cdot \mE \left[\left\|\Xi\trans\u_K(Z,A,X)\right\|^2\right] \\ 
    &~ \le \mE \left[\{Y-h(W,A,X;\bm\gamma_0)\}^2\right]\cdot \mE \left[\left\|\nabla_{\bm\gamma_0} h(W,A,X;\bm\gamma_0)\right\|^2 \right]<\infty.
\end{align*}
Then we have 
\begin{align*}
    \var\left\{\B_{(K+1)\times(p+1)}\trans {\bm\Theta}_{(K+1)\times(K+1)}^{-1}\sqrt{N}\G_K(\bm\gamma_0,\tau_0)\right\} = O(1),
\end{align*}
and \eqref{eq:app_rootn_T1} holds. To show \eqref{eq:app_rootn_T2}, note that
    \begin{align*}
        & \nabla_{\bm\gamma,\tau} \G_K({\bm\gamma},\tau) \wh{\bm\Theta}_{(K+1)\times(K+1)}^{-1}\nabla_{\bm\gamma,\tau}\G_K({\bm\gamma},\tau)\trans
            =\left(\begin{array}{cc}
                \M_{p\times p}^{(1)}(\gamma) & \M_{p\times 1}^{(2)}(\gamma)\\ 
                \M_{p\times 1}^{(2)}(\gamma)\trans & 1
                \end{array}\right),
    \end{align*}
where 
\begin{footnotesize}
    \begin{align*}
        \M_{p\times p}^{(1)}(\gamma) =&~             
                \left\{\frac{1}{N}\sum_{i=1}^{N} \nabla_{\bm\gamma} h(W_i,A_i,X_i;\bm\gamma)\u_K(Z_i,A_i,X_i)\trans\right\}  \left\{\frac{1}{N}\sum_{i=1}^{N} \u_K(Z_i,A_i,X_i)\u_K(Z_i,A_i,X_i)\trans\right\}^{-1} \\ 
                &\qquad \times \left\{\frac{1}{N}\sum_{i=1}^{N} \u_K(Z_i,A_i,X_i)\nabla_{\bm\gamma} h(W_i,A_i,X_i;\bm\gamma)\trans\right\} \\ 
                &~ +   \left\{\frac{1}{N}\sum_{i=1}^{N} \nabla_{\bm\gamma}\{ h(W_i,1,X_i;\bm\gamma) - h(W_i,0,X_i;\bm\gamma)\}\right\}
                \left\{\frac{1}{N}\sum_{i=1}^{N} \nabla_{\bm\gamma}\{ h(W_i,1,X_i;\bm\gamma) - h(W_i,0,X_i;\bm\gamma)\}\right\}\trans,
    \end{align*}
\end{footnotesize}
and 
\begin{align*}
    \M_{p\times p}^{(2)}(\gamma)
    =-\frac{1}{N}\sum_{i=1}^{N} \nabla_{\bm\gamma}\{ h(W_i,1,X_i;\bm\gamma) - h(W_i,0,X_i;\bm\gamma)\}.
\end{align*}
Recall the definition of $\wh\Xi$ in \eqref{eq:app_consis_deriv_T1.1}. By the triangle inequality, we have
\begin{align*}
\| \M_{p\times p}^{(1)}(\gamma)\| \le&~ \frac{1}{N} \sum_{i=1}^{N} \left\|\wh\Xi\trans \u_K(Z_i,A_i,X_i)\right\|^2 + \frac{1}{N} \sum_{i=1}^{N} \left\|\nabla_{\bm\gamma}\{ h(W_i,1,X_i;\bm\gamma) - h(W_i,0,X_i;\bm\gamma)\}\right\|^2 \\ 
\le&~ \frac{1}{N} \sum_{i=1}^{N} \left\|\nabla_{\bm\gamma} h(W_i,A_i,X_i;\bm\gamma)\right\|^2 + \frac{1}{N} \sum_{i=1}^{N} \left\|\nabla_{\bm\gamma}\{ h(W_i,1,X_i;\bm\gamma) - h(W_i,0,X_i;\bm\gamma)\}\right\|^2,
\end{align*}
where the last inequality is similar to the calculation of \eqref{eq:app_consis_deriv_T1}. Meanwhile, 
\begin{align*}
    \left\|\M_{p\times p}^{(2)}(\gamma)\right\| \le 
    \frac{1}{N} \sum_{i=1}^{N} \left\|\nabla_{\bm\gamma}\{ h(W_i,1,X_i;\bm\gamma) - h(W_i,0,X_i;\bm\gamma)\}\right\|^2.
\end{align*}
Then \eqref{eq:app_rootn_T2} holds under regularity conditions \ref{cond:unif_bound1}-\ref{cond:unif_bound2}. 
Finally, proving \eqref{eq:app_rootn_T3} is similar to proving $\sup _{(\bm\gamma, \tau) \in \Gamma \times \calT}|\wh{Q}_N(\bm\gamma, \tau)-Q(\bm\gamma, \tau)| \convp 0$ in \Cref{lemma:app_ini_consistency}.
We can first obtain that for each $(\bm\gamma, \tau) \in \Gamma \times \calT$ and $j=1,\ldots,p$, 
\begin{align*}
&\G_K({\bm\gamma},\tau)\trans \wh{\bm\Theta}_{(K+1)\times(K+1)}^{-1}  \partial^2 \G_K({\bm\gamma},\tau) / \partial \gamma_j\partial\bm\gamma\trans\\
&~~ -\mE\{g_K(\O;\bm\gamma,\tau)\}\trans {\bm\Theta}_{(K+1)\times(K+1)}^{-1}  \mE\{\partial^2 \g_K(\O;{\bm\gamma},\tau) / \partial \gamma_j\partial\bm\gamma\trans\}=o_p(1).
\end{align*}
Then, it can be strengthened to achieve uniform convergence over  $(\bm\gamma, \tau) \in \Gamma \times \calT$. Then, in light of the facts that $\left\|\wt{\bm\gamma}-\bm\gamma_0\right\| \le\left\|\check{\bm\gamma}-\bm\gamma_0\right\| \convp 0$ and $\mE\left[g_K\left(\O; \bm\gamma_0, \tau_0\right)\right]=$ 0, we can obtain \eqref{eq:app_rootn_T3}. Combining \eqref{eq:app_rootn_T0}--\eqref{eq:app_rootn_T3} yields that $\|(\check{\bm\gamma},\check\tau)-(\bm\gamma_0,\tau_0)\|=O_p(N^{-1/2})$, then it completes proof.

\end{proof}

\begin{lemma}\label{lemma:app_root_n_Vk}
     Suppose \Cref{ass:sieve} and regularity conditions \ref{cond:compact}--\ref{cond:unif_bound2} in \Cref{sec:app_regularity_conditions} hold. We have
    \begin{align*}
        \V_K^{-1/2}\left( 
            \begin{array}{c}
                \sqrt{N}(\wh{\bm\gamma}-\bm\gamma_0)\\\sqrt{N}(\wh\tau-\tau_0)
            \end{array}
         \right) \convd N\left(0,\I_{(p+1)\times(p+1)}\right),
    \end{align*}
    where $K\rightarrow\infty$  as $N\rightarrow\infty$.
\end{lemma}

\begin{proof}

    We define $(\bar{\bm\gamma},\bar\tau)=\amin_{(\bm\gamma,\tau)\in\Gamma\times\calT} Q_N(\bm\gamma, \tau)$ with
    \begin{align*} 
     Q_N(\bm\gamma, \tau)=\G_K(\bm\gamma,\tau)\trans \bm\Upsilon_{(K+1)\times(K+1)}^{-1}\G_K(\bm\gamma,\tau).
    \end{align*}
    We prove this lemma by first showing that 
    \begin{align}\label{eq:app_pf_lemma_hat_rootn_T1}
        \V_K^{-1/2}\left( 
            \begin{array}{c}
                \sqrt{N}(\bar{\bm\gamma}-\bm\gamma_0)\\\sqrt{N}(\bar\tau-\tau_0)
            \end{array}
         \right) \convd N\left(0,\I_{(p+1)\times(p+1)}\right),
    \end{align}
    and then showing 
    \begin{align}\label{eq:app_pf_lemma_hat_rootn_T2}
        \left(\begin{array}{c} 
            \sqrt{N}(\wh{\bm\gamma}-\bar{\bm\gamma})\\\sqrt{N}(\wh\tau-\bar\tau)
        \end{array}
        \right) = o_p(1).
    \end{align}

    \paragraph{For term \eqref{eq:app_pf_lemma_hat_rootn_T1}. }
    We can show that following a very similar procedure as \Cref{lemma:app_ini_consistency} $\|(\bar{\bm\gamma},\bar\tau)-(\bm\gamma_0,\tau_0)\|=o_p(1)$. 
    By the definition of $(\bar{\bm\gamma},\bar\tau)$, it should satisfy the first-order condition
    \begin{align*}
    \nabla_{\bm\gamma,\tau} \G_K(\bar{\bm\gamma},\bar\tau) \bm\Upsilon_{(K+1)\times(K+1)}^{-1}\G_K(\bar{\bm\gamma},\bar\tau)=\0.
    \end{align*}
    By the mean value theorem, we have 
    \begin{align}
        \0=&~ \nabla_{\bm\gamma,\tau} \G_K(\bm\gamma_0,\tau_0) \bm\Upsilon_{(K+1)\times(K+1)}^{-1}\sqrt{N}\G_K(\bm\gamma_0,\tau_0) \nonumber\\
        &~ + \nabla_{\bm\gamma,\tau} \G_K(\wt{\bm\gamma},\wt\tau) \bm\Upsilon_{(K+1)\times(K+1)}^{-1}\nabla_{\bm\gamma,\tau}\G_K(\wt{\bm\gamma},\wt\tau)\trans 
        \binom{\sqrt{N}(\bar{\bm\gamma}-\bm\gamma_0)}{\sqrt{N}(\bar{\tau}-\tau_0)} \nonumber\\ 
        &~ + \sqrt{N}\G_K(\wt{\bm\gamma},\wt\tau)\trans \bm\Upsilon_{(K+1)\times(K+1)}^{-1} \nonumber\\
        &\qquad \times \left[\frac{\partial^2 \G_K(\wt{\bm\gamma},\wt\tau)}{\partial \gamma_1\partial\bm\gamma\trans}(\bar{\bm\gamma}-\bm\gamma_0),
        \ldots,
        \frac{\partial^2 \G_K(\wt{\bm\gamma},\wt\tau)}{\partial \gamma_p\partial\bm\gamma\trans}(\bar{\bm\gamma}-\bm\gamma_0),
        \0_{K+1}\right].\label{eq:app_pf_lemma_hat_rootn_T1_1}
    \end{align}
    where $(\wt{\bm\gamma},\wt\tau)$ lies between $(\bar{\bm\gamma},\bar\tau)$ and $(\bm\gamma_0,\tau_0)$. Note that \eqref{eq:app_pf_lemma_hat_rootn_T1_1} has the same structure as \eqref{eq:app_rootn_T0} except that the weighting matrix used in \eqref{eq:app_pf_lemma_hat_rootn_T1_1} is $\bm\Upsilon_{(K+1)\times(K+1)}^{-1}$ while the weighting matrix used in \eqref{eq:app_rootn_T0} is $\wh{\bm\Theta}_{(K+1)\times(K+1)}^{-1}$. Using a similar argument of showing $\sup_{(\bm\gamma,\tau)\in\Gamma\times\calT}|\wh{Q}_N(\bm\gamma,\tau)-Q(\bm\gamma,\tau)| \convp 0$ in \Cref{lemma:app_ini_consistency} and \eqref{eq:app_rootn_T3}, we can obtain that
    \begin{align}
        &\nabla_{\bm\gamma,\tau} \G_K(\wt{\bm\gamma},\wt\tau) \bm\Upsilon_{(K+1)\times(K+1)}^{-1}\nabla_{\bm\gamma,\tau}\G_K(\wt{\bm\gamma},\wt\tau)\trans \nonumber\\
        &\qquad = \B_{(K+1)\times(p+1)}\trans {\bm\Upsilon}_{(K+1)\times(K+1)}^{-1} \B_{(K+1)\times(p+1)} +o_p(1), \label{eq:app_pf_lemma_hat_rootn_T1_2}
    \end{align}
    and 
    \begin{align}
        & \sqrt{N}\G_K(\wt{\bm\gamma},\wt\tau)\trans \bm\Upsilon_{(K+1)\times(K+1)}^{-1} \nonumber\\
        &\qquad \times \left[\frac{\partial^2 \G_K(\wt{\bm\gamma},\wt\tau)}{\partial \gamma_1\partial\bm\gamma\trans}(\check{\bm\gamma}-\bm\gamma_0),
        \ldots,
        \frac{\partial^2 \G_K(\wt{\bm\gamma},\wt\tau)}{\partial \gamma_p\partial\bm\gamma\trans}(\check{\bm\gamma}-\bm\gamma_0),
        \0_{K+1}\right] = o_p(1).\label{eq:app_pf_lemma_hat_rootn_T1_3}
    \end{align}
    Combining \eqref{eq:app_rootn_T1.2}, \eqref{eq:app_pf_lemma_hat_rootn_T1_1}, \eqref{eq:app_pf_lemma_hat_rootn_T1_2} and \eqref{eq:app_pf_lemma_hat_rootn_T1_3} yields that
    \begin{align*}
        \left( 
            \begin{array}{c}
                \sqrt{N}(\bar{\bm\gamma}-\bm\gamma_0)\\\sqrt{N}(\bar\tau-\tau_0)
            \end{array}
         \right)
        =&~  \left\{\B_{(K+1)\times(p+1)}\trans {\bm\Upsilon}_{(K+1)\times(K+1)}^{-1} \B_{(K+1)\times(p+1)} \right\}^{-1} \\ 
        &\qquad \times \B_{(K+1)\times(p+1)}\trans \bm\Upsilon_{(K+1)\times(K+1)}^{-1}\sqrt{N}\G_K(\bm\gamma_0,\tau_0) + o_p(1).
    \end{align*}
    Then 
    \begin{align*}
    \var 
        \left( 
            \begin{array}{c}
                \sqrt{N}(\bar{\bm\gamma}-\bm\gamma_0)\\\sqrt{N}(\bar\tau-\tau_0)
            \end{array}
         \right)
     =&~ \left\{\B_{(K+1)\times(p+1)}\trans {\bm\Upsilon}_{(K+1)\times(K+1)}^{-1} \B_{(K+1)\times(p+1)} \right\}^{-1}
    + o_p(1) \\ 
    =&~ \V_K + o_p(1) .
    \end{align*}
    It follows that 
    \begin{align*}
        \V_K^{-1/2}\left( 
            \begin{array}{c}
                \sqrt{N}(\bar{\bm\gamma}-\bm\gamma_0)\\\sqrt{N}(\bar\tau-\tau_0)
            \end{array}
         \right) 
         =&~ \V_K^{-1/2}  \left\{\B_{(K+1)\times(p+1)}\trans \bm\Upsilon_{(K+1)\times(K+1)}^{-1}\sqrt{N}\G_K(\bm\gamma_0,\tau_0)\right\} + o_p(1).
    \end{align*}
    We next follow the multivariate central limit theorem for random linear vector forms established in \cite{Eicker1966} to show $\V_K^{-1/2}\left( 
        \begin{array}{c}
            \sqrt{N}(\bar{\bm\gamma}-\bm\gamma_0)\\\sqrt{N}(\bar\tau-\tau_0)
        \end{array}
     \right)$ is asymptotically normal.  
     Note that 
     \begin{align*}
        \sqrt{N}\G_K(\bm\gamma_0,\tau_0) =&~ 
        \left(
    \begin{array}{c}
        \frac{1}{\sqrt{N}}\sum_{i=1}^{N}\{Y_i-h(W_i,A_i,X_i;\bm\gamma_0)\}\u_K(Z_i,A_i,X_i)\\
        \frac{1}{\sqrt{N}}\sum_{i=1}^{N}\{\tau_0-h(W_i,1,X_i;\bm\gamma_0)+h(W_i,0,X_i;\bm\gamma_0)\} 
    \end{array}
    \right)\\ 
    =:&~ \wt{\A}_{(K+1)\times N(K+1)}\calE,
     \end{align*}
     where 
     \begin{align*}
        &\wt{\A}_{(K+1)\times N(K+1)} = \left(\begin{matrix}
            \frac{1}{\sqrt{N}}\1_N\trans, & \0_N\trans, & \0_N\trans, & \cdots & \0_N\trans, & \0_N\trans \\ 
            \0_N\trans, & \frac{1}{\sqrt{N}}\1_N\trans, & \0_N\trans, & \cdots & \0_N\trans, & \0_N\trans  \\ 
            \vdots &  \vdots &  \vdots & &  \vdots &  \vdots    \\
            \0_N\trans, &  \0_N\trans, & \cdots & \0_N\trans, & \0_N\trans, & \frac{1}{\sqrt{N}}\1_N\trans \\ 
        \end{matrix}\right),\\
     &\calE = \big(\v_{N,1}\trans,\ldots,v_{N,K}\trans,\w_N\trans\big)\trans, \\ 
     &\v_{N,k} = \left( \begin{matrix}
        \{Y_1-h(W_1,A_1,X_1;\bm\gamma_0)\}u_{Kk}(Z_1,A_1,X_1)\\ 
        \vdots\\ 
        \{Y_N-h(W_N,A_N,X_N;\bm\gamma_0)\}u_{Kk}(Z_N,A_N,X_N)
     \end{matrix} \right),\\
     &\w_N = \left( \begin{matrix}
        \tau_0-h(W_1,1,X_1;\bm\gamma_0)+h(W_1,0,X_1;\bm\gamma_0)\\ 
        \vdots\\ 
        \tau_0-h(W_N,1,X_N;\bm\gamma_0)+h(W_N,0,X_N;\bm\gamma_0)
     \end{matrix} \right),
     \end{align*}
     for $k=1,\ldots,K$, and $\1_N$ (resp. $\0_N$) denotes an N-dimensional column vector whose elements are all of 1's (resp. 0's).
     Let $\A_{(p+1)\times N(K+1)}=\B_{(K+1)\times(p+1)}\trans \bm\Upsilon_{(K+1)\times(K+1)}\wt{\A}_{(K+1)\times N(K+1)}$. We have 
     \begin{align*}
         \B_{(K+1)\times(p+1)}\trans \bm\Upsilon_{(K+1)\times(K+1)}^{-1}\sqrt{N}\G_K(\bm\gamma_0,\tau_0)  = \A_{(p+1)\times N(K+1)}\calE.
     \end{align*}
     Following \cite{Eicker1966}, \eqref{eq:app_pf_lemma_hat_rootn_T1} holds if the following three conditions hold:
     \begin{enumerate}
        \item $\max_{i=1,\ldots,N(K+1)}\a_i\trans\{\A_{(p+1)\times N(K+1)}\A_{(p+1)\times N(K+1)}\trans\}^{-1}\a_i\to 0$, where $\a_i$ is the $i$-th column of $\A_{(p+1)\times N(K+1)}$;
        \item $\sup_{k\in\{1,\ldots,K\},j\in\{1,\ldots,N\}} \mE\left\{v_{N,k,j}^2I\left(|v_{N,k,j}|>s\right)\right\}\to 0$ and 
        \\ $\sup_{j\in\{1,\ldots,N\}} \mE\left\{w_{N,j}^2I\left(|w_{N,j}|>s\right)\right\}\to 0$ as $s\to \infty$, where $v_{N,k,j}$ (resp. $w_{N,j}$) is the $j$-th component of $\v_{N,k}$ (resp. $\w_N$);
        \item $\inf_{k\in\{1,\ldots,K\},j\in\{1,\ldots,N\}} \mE(v_{N,k,j}^2)>0$ and $\inf_{j\in\{1,\ldots,N\}} \mE(w_{N,j}^2)>0$.
     \end{enumerate}
     Condition 1 is naturally satisfied by the definition of $\wt{\A}_{(K+1)\times N(K+1)}$. Condition 2 is satisfied due to the dominated convergence theorem and the fact\\ $\sup_{k\in\{1,\ldots,K\}} \mE\{u_{Kk}(Z,A,X)^2\} \le \lambda_{\max}(\mE\{\u_{K}(Z,A,X)\u_{K}(Z,A,X)\trans\})<\infty$ by \Cref{ass:sieve}(i). Condition 3 is satisfied due to the fact \\ $\inf_{k\in\{1,\ldots,K\}} \mE\{u_{Kk}(Z,A,X)^2\} \ge \lambda_{\min}(\mE\{\u_{K}(Z,A,X)\u_{K}(Z,A,X)\trans\})>0$ by \Cref{ass:sieve}(i). Then \eqref{eq:app_pf_lemma_hat_rootn_T1} holds following \cite{Eicker1966}.

     \paragraph{For term \eqref{eq:app_pf_lemma_hat_rootn_T2}. }
     It suffices to prove 
     \begin{align*}
     \left\|\wh{\bm\Upsilon}_{(K+1)\times(K+1)}-\bm\Upsilon_{(K+1)\times(K+1)}\right\|=o_p(1).
     \end{align*}
    By definition, we have 
    \begin{align}
        & \left\|\wh{\bm\Upsilon}_{(K+1)\times(K+1)}-\bm\Upsilon_{(K+1)\times(K+1)}\right\|^2 \nonumber\\
        =&~ \Bigg\| \frac{1}{N}\sum_{i=1}^N \left\{Y_i-h(W_i,A_i,X_i;\check{\bm\gamma})\right\}^2\u_K(Z_i,A_i,X_i)^{\otimes 2} \nonumber\\ 
        &\qquad\quad - \mE\left[ \{Y-h(W,A,X;\bm\gamma_0)\}^2\u_K(Z,A,X)^{\otimes 2}\right]\Bigg\|^2 \label{eq:app_pf_lemma_hat_rootn_T2_1.1}\\ 
        +&~   2\Bigg\| \frac{1}{N}\sum_{i=1}^N \{Y_i-h(W_i,A_i,X_i)\}\{\check\tau - h(W_i,1,X_i) +h(W_i,0,X_i)\} \u_K(Z_i,A_i,X_i) \nonumber\\ 
        &\qquad - \mE\left[\{Y-h(W,A,X)\}\{\tau_0 - h(W,1,X) +h(W,0,X)\} \u_K(Z,A,X)\right]\Bigg\|^2 \label{eq:app_pf_lemma_hat_rootn_T2_1.2}\\ 
        +&~ \Bigg| \frac{1}{N}\sum_{i=1}^N \{\check\tau-h(W_i,1,X_i;\check{\bm\gamma})+h(W_i,0,X_i;\check{\bm\gamma})\}^2 \nonumber\\ 
        &\qquad - \mE\left[\{\tau_0-h(W,1,X;\bm\gamma_0)+h(W,0,X;\bm\gamma_0)\}^2\right] \Bigg| .\label{eq:app_pf_lemma_hat_rootn_T2_1.3}
    \end{align}
    We next show that \eqref{eq:app_pf_lemma_hat_rootn_T2_1.1} is $o_p(1)$, and we can similarly obtain \eqref{eq:app_pf_lemma_hat_rootn_T2_1.2} and  \eqref{eq:app_pf_lemma_hat_rootn_T2_1.3} are $o_p(1)$.
    Using the mean value theorem and triangle inequality, we obtain 
     \begin{align*}
        & \left\| \frac{1}{N}\sum_{i=1}^N \left\{Y_i-h(W_i,A_i,X_i;\check{\bm\gamma})\right\}^2\u_K(Z_i,A_i,X_i)^{\otimes 2} - \mE\left[ \{Y-h(W,A,X;\bm\gamma_0)\}^2\u_K(Z,A,X)^{\otimes 2}\right]\right\| \\
        \le&  \left\| \frac{1}{N}\sum_{i=1}^N \left\{Y_i-h(W_i,A_i,X_i;\bm\gamma_0)\right\}^2\u_K(Z_i,A_i,X_i)^{\otimes 2} - \mE\left[ \{Y-h(W,A,X;\bm\gamma_0)\}^2\u_K(Z,A,X)^{\otimes 2}\right]\right\|   \\ 
        &+ \left\| (\check{\bm\gamma}-\bm\gamma_0)\trans \cdot \frac{2}{N}\sum_{i=1}^N \left\{Y_i-h(W_i,A_i,X_i;\bm\gamma^*)\right\}\nabla_{\bm\gamma}h(W_i,A_i,X_i;\bm\gamma^*) \u_K(Z_i,A_i,X_i)^{\otimes 2}\right\| ,
     \end{align*}
   where $\bm\gamma^*$ lies between $\check{\bm\gamma}$ and $\bm\gamma_0$. The first term is $o_p(1)$ by directly computing its second-order moment like \eqref{eq:app_consis_T1.2.1} and using Chebyshev inequality. The second term is $O_p(\sqrt{K/N})=o_p(1)$, following from \Cref{lemma:app_ini_rate}, \Cref{ass:sieve} and the fact
   $$\left\|\mE\{\u_K(Z,A,X)^{\otimes 2}\}\right\|^2 = \lambda_{\max}\left(\mE\left\{\u_K(Z,A,X)^{\otimes 2}\right\}\right) \|\I_K\|^2=O(1)\cdot O(K)=O(K).$$
   Then we complete the proof.

\end{proof}

\section{Proofs of main results}\label{sec:app_proof_main}

\subsection{Proof of Theorem \ref{thm:asym_normal}}\label{sec:pf_thm_asy_normal}
\begin{proof}
    Following \Cref{lemma:app_root_n_Vk}, we only need to show that 
           $\V_K\rightarrow
            \left( 
                \begin{matrix}
                    \V_\gamma&*\\ 
                    *&V_\tau
                \end{matrix}
            \right)$,
    with $\V_\gamma$ and $V_\tau$ defined in \Cref{thm:asym_normal}. It then implies that $\sqrt{N}(\wh{\bm\gamma}-\bm\gamma_0)\convd N(0,\V_{\bm\gamma})$ and $\sqrt{N}(\wh\tau-\tau_0)\convd N(0,V_{\tau})$.
    Recall that $$\V_K = \{\B_{(K+1)\times(p+1)}\trans {\bm\Upsilon}_{(K+1)\times(K+1)}^{-1} \B_{(K+1)\times(p+1)}\}^{-1}.$$
    Using the inverse matrix formula \eqref{eq:mat_inv}, we have: 
    \beqrs
    {\bm\Upsilon}_{(K+1)\times(K+1)}^{-1} = \left(\begin{array}{cc}
        \A_{K \times K}^{-1}+\frac{1}{c} \A_{K \times K}^{-1} \b_K \b_K\trans  \A_{K \times K}^{-1}, & -\frac{1}{c} \A_{K \times K}^{-1} \b_K \\
        -\frac{1}{c} \mathbf{b}_K\trans \A_{K \times K}^{-1}, & \frac{1}{c}
        \end{array}\right),
    \eeqrs
    where 
        \begin{align}
        & \A_{K \times K}= \mE\left[\{Y-h(W,A,X)\}^2 \u_K(Z,A,X)^{\otimes 2}\right], \label{eq:A}\\ 
        & \b_K= \mE\left[\{Y-h(W,A,X)\}\{\tau_0 - h(W,1,X)+h(W,0,X)\} \u_K(Z,A,X)\right],\label{eq:b_k} \\
        & c= \mE\left[\{\tau_0 - h(W,1,X)+h(W,0,X)\}^2\right] -\b_K\trans \A_{K \times K}^{-1} \b_K.\nonumber
        \end{align}
    Then we have 
    \beqrs
    \B_{(K+1)\times(p+1)}\trans  {\bm\Upsilon}_{(K+1)\times(K+1)}^{-1} \B_{(K+1)\times(p+1)}=\left(\begin{array}{cc}
        \wt{\A}_{p \times p} & \wt{\b}_p \\
        & \\
        \wt{\b}_p\trans  & \frac{1}{c}
        \end{array}\right),
    \eeqrs
    where 
    \begin{align}
    &\wt\A_{p \times p} =   \mE\{\nabla_{\bm\gamma} h(W,A,X;\bm\gamma_0)\u_K(Z,A,X)\trans\} \nonumber\\ 
     &\qquad\qquad\times\A_{K \times K}^{-1} \cdot  \mE\{\u_K(Z,A,X)\nabla_{\bm\gamma} h(W,A,X;\bm\gamma_0)\trans\} + c \wt{\b}_p\wt{\b}_p\trans, 
    \label{eq:A_tilde}\\
    &\wt{\b}_p =  \frac{1}{c} \mE\{\nabla_{\bm\gamma} h(W,A,X;\bm\gamma_0)\u_K(Z,A,X)\trans\}  \A_{K \times K}^{-1} \b_K \nonumber\\ 
    &\qquad\qquad  - \frac{1}{c}  \mE\left\{\frac{(-1)^{1-A}}{f(A\mid W,X)}\nabla_{\bm\gamma} h(W,A,X;\bm\gamma_0)\right\}. 
    \label{eq:b_tilde}
    \end{align} 
    Using the matrix inversion formula \eqref{eq:mat_inv} again, we can obtain that
    \begin{align*}
    \V_K  =&~   \left(\B_{(K+1)\times(p+1)}\trans  {\bm\Upsilon}_{(K+1)\times(K+1)}^{-1}  \B_{(K+1)\times(p+1)}\right)^{-1} \\
    =&~   \left(\begin{array}{cc}
    \wt{\A}_{p \times p}^{-1}+\frac{1}{\wt{c}} \wt{\A}_{p \times p}^{-1} \wt{\b}_p \wt{\b}_p\trans \wt{\A}_{p \times p}^{-1}  & -\frac{1}{\wt{c}} \wt{\A}_{p \times p}^{-1} \wt{\b}_p \\
    -\frac{1}{\wt{c}} \wt{\b}_p\trans \wt{\A}_{p \times p}^{-1} & \frac{1}{\wt{c}}
    \end{array}\right),
    \end{align*}
    where
    $$
    \wt{c}=\frac{1}{c}-\wt{\b}_p\trans \wt{\A}_{p \times p}^{-1} \wt{\b}_p.
    $$
    Next, we show that 
    \begin{align}
       & \wt{\A}_{p \times p}^{-1}+\frac{1}{\wt{c}} \wt{\A}_{p \times p}^{-1} \wt{\b}_p \wt{\b}_p\trans \wt{\A}_{p \times p}^{-1}= \V_{\gamma}+o(1), \label{eq:app_var_gamma}\\ 
       & \frac{1}{\wt{c}} = V_{\tau}+o(1). \label{eq:app_var_tau}
    \end{align}

    We first prove \eqref{eq:app_var_gamma}.
    Note that
    \beqrs
    && \left(\wt{\A}_{p \times p}^{-1}+\frac{1}{\wt{c}} \wt{\A}_{p \times p}^{-1} \wt{\b}_p \wt{\b}_p\trans \wt{\A}_{p \times p}^{-1}\right) \cdot\left(\wt{\A}_{p \times p}-c \cdot \wt{\b}_p \wt{\b}_p\trans\right) \\
    &=&  \I_{p \times p}-c \cdot \wt{\A}_{p \times p}^{-1} \wt{\b}_p \wt{\b}_p\trans+\frac{1}{\wt{c}} \wt{\A}_{p \times p}^{-1} \wt{\b}_p \wt{\b}_p\trans-\frac{c}{\wt{c}} \wt{\A}_{p \times p}^{-1} \wt{\b}_p \wt{\b}_p\trans \wt{\A}_{p \times p}^{-1} \wt{\b}_p \wt{\b}_p\trans \\
    &=&  \I_{p \times p}-\frac{1}{\wt{c}} \cdot\left\{\wt{c} c \cdot \wt{\A}_{p \times p}^{-1} \wt{\b}_p \wt{\b}_p\trans+c \cdot \wt{\A}_{p \times p}^{-1} \wt{\b}_p \wt{\b}_p\trans \cdot \wt{\b}_p\trans \wt{\A}_{p \times p}^{-1} \wt{\b}_p\right\}+\frac{1}{\wt{c}} \wt{\A}_{p \times p}^{-1} \wt{\b}_p \wt{\b}_p\trans \\
    &=&  \I_{p \times p}-\frac{1}{\wt{c}} \cdot\left\{\wt{\A}_{p \times p}^{-1} \wt{\b}_p \wt{\b}_p\trans\right\}+\frac{1}{\wt{c}} \wt{\A}_{p \times p}^{-1} \wt{\b}_p \wt{\b}_p\trans \quad\left(\text { since } c \wt{c}=1-c \cdot \wt{\b}_p\trans \wt{\A}_{p \times p}^{-1} \wt{\b}_p\right) \\
    &=&  \I_{p \times p}.
    \eeqrs
    We have 
    \begin{align}
    \label{eq:app_var_gamma_1}
    &\wt{\A}_{p \times p}^{-1}+\frac{1}{\wt{c}} \wt{\A}_{p \times p}^{-1} \wt{\b}_p \wt{\b}_p\trans \wt{\A}_{p \times p}^{-1}=\left\{\wt{\A}_{p \times p}-c \cdot \wt{\b}_p \wt{\b}_p\trans\right\}^{-1} .
    \end{align}
    Then, by \eqref{eq:A} and \eqref{eq:A_tilde}, we have 
    \begin{align*}
    &\wt{\A}_{p \times p}^{-1}+\frac{1}{\wt{c}} \wt{\A}_{p \times p}^{-1} \wt{\b}_p \wt{\b}_p\trans \wt{\A}_{p \times p}^{-1} \\ 
    =&~   \Big[\mE\{\nabla_{\bm\gamma} h(W,A,X;\bm\gamma_0)\u_K(Z,A,X)\trans\}\cdot \A_{K \times K}^{-1} \cdot  \mE\{\u_K(Z,A,X)\nabla_{\bm\gamma} h(W,A,X;\bm\gamma_0)\trans\}\Big]^{-1}\\
    =:&~  \Big[ \mE\{\p_{K}(Z,A,X)\p_{K}(Z,A,X)\trans\}\Big]^{-1},
    \end{align*}
     where 
    \begin{align*}
         \p_{K}(Z,A,X) =&~   \mE\left[\nabla_{\bm\gamma} h(W,A,X;\bm\gamma_0)\u_K(Z,A,X)\trans\right]\times  \mE\left[\{Y-h(W,A,X)\}^2 \u_K(Z,A,X)^{\otimes 2}\right]^{-1} \\
        & \qquad\times \sqrt{ \mE\left[\{Y-h(W,A,X)\}^2 \mid Z,A,X\right]}\u_K(Z,A,X)\\
        =&~   \mE\left[\frac{ \mE\{\nabla_{\bm\gamma} h(W,A,X;\bm\gamma_0)\mid Z,A,X\}}{\sqrt{ \mE\left[\{Y-h(W,A,X)\}^2 \mid Z,A,X\right]}}\sqrt{ \mE\left[\{Y-h(W,A,X)\}^2 \mid Z,A,X\right]}\u_K(Z,A,X)\trans\right] \\ 
        &\quad \times~  \mE\left[\left\{\sqrt{ \mE\left[\{Y-h(W,A,X)\}^2 \mid Z,A,X\right]} \u_K(Z,A,X)\right\}^{\otimes 2}\right]^{-1}\\ 
        &\quad\times~ \sqrt{ \mE\left[\{Y-h(W,A,X)\}^2 \mid Z,A,X\right]}\u_K(Z,A,X).
    \end{align*}
    Note that $\p_K(Z,A,X)$ is the least squared projection of $\dfrac{ \mE\{\nabla_{\bm\gamma} h(W,A,X;\bm\gamma_0)\mid Z,A,X\}}{\sqrt{ \mE\left[\{Y-h(W,A,X)\}^2 \mid Z,A,X\right]}}$ onto the space linear spanned by $\bigg\{\sqrt{ \mE[\{Y-h(W,A,X)\}^2 \mid Z,A,X]}\u_K(Z,A,X)\bigg\}$. By \Cref{ass:sieve}(ii) and \ref{cond:p_smooth}, we have
    \beqrs\label{eq:app_var_gamma_1.1}
    \mE\left[\left\|\p_K(Z,A,X) - \frac{ \mE\{\nabla_{\bm\gamma} h(W,A,X;\bm\gamma_0)\mid Z,A,X\}}{\sqrt{ \mE[\{Y-h(W,A,X)\}^2 \mid Z,A,X]}} \right\|^2\right] = o(1).
    \eeqrs
    Therefore we have that 
    \begin{align}
    \label{eq:app_var_gamma_2}
     \mE[\p_{K}(Z,A,X)\p_{K}(Z,A,X)\trans]
    \rightarrow&~ \mE\left[\frac{\bigbc{ \mE[\nabla_{\bm\gamma} h(W,A,X;\bm\gamma_0)\mid Z,A,X]}^{\otimes2}}{ \mE[\{Y-h(W,A,X)\}^2 \mid Z,A,X]}\right] .
    \end{align}
    It follows that
    \begin{align}
    \label{eq:app_var_gamma_3}
    \wt{\A}_{p \times p}^{-1}+\frac{1}{\wt{c}} \wt{\A}_{p \times p}^{-1} \wt{\b}_p \wt{\b}_p\trans \wt{\A}_{p \times p}^{-1} \rightarrow \left(\mE\left[\frac{\bigbc{ \mE[\nabla_{\bm\gamma} h(W,A,X;\bm\gamma_0)\mid Z,A,X]}^{\otimes2}}{ \mE[\{Y-h(W,A,X)\}^2 \mid Z,A,X]}\right]\right)^{-1} = \V_{\bm\gamma},
    \end{align}
    where the last equality is because
        \begin{align}\label{eq:app_var_gamma_verify}
            \V_{\bm\gamma}=&~ E\{\bm\psi_1(\O)\bm\psi_1(\O)\trans\}\nonumber\\
            =&~   \left[\mE\left\{\frac{\partial h(W,A,X;\bm\gamma_0)\m_\eff(Z,A,X)}{
                \partial\bm\gamma\trans
            }\right\}\right]^{-1} \nonumber \\
            &\qquad \times E\left[\{Y-h(W,A,X)\}^2\m_\eff(Z,A,X)\m_\eff(Z,A,X)\trans\right]\nonumber\\ 
            &\qquad \times \left[\mE\left\{\frac{\partial h(W,A,X;\bm\gamma_0)\m_\eff(Z,A,X)\trans}{
                \partial\bm\gamma}\right\}\right]^{-1} \nonumber\\ 
            =&~ \left(\mE\left[\frac{ \mE\{\nabla_{\bm\gamma} h(W,A,X;\bm\gamma_0)\mid Z,A,X\}}{ \mE[\{Y-h(W,A,X)\}^2 \mid Z,A,X]}\nabla_{\bm\gamma} h(W,A,X;\bm\gamma_0)\trans\right]\right)^{-1}\nonumber\\ 
            &\qquad \times E\left[\{Y-h(W,A,X)\}^2\frac{\bigbc{ \mE[\nabla_{\bm\gamma} h(W,A,X;\bm\gamma_0)\mid Z,A,X]}^{\otimes2}}{ \{\mE[\{Y-h(W,A,X)\}^2 \mid Z,A,X]\}^2}\right]\nonumber\\ 
            &\qquad \times \left(\mE\left[ \nabla_{\bm\gamma} h(W,A,X;\bm\gamma_0)\frac{ \mE\{\nabla_{\bm\gamma} h(W,A,X;\bm\gamma_0)\mid Z,A,X\}\trans}{ \mE[\{Y-h(W,A,X)\}^2 \mid Z,A,X]}\right]\right)^{-1}\nonumber\\ 
            =&~ \left(\mE\left[\frac{\bigbc{ \mE[\nabla_{\bm\gamma} h(W,A,X;\bm\gamma_0)\mid Z,A,X]}^{\otimes2}}{ \mE[\{Y-h(W,A,X)\}^2 \mid Z,A,X]}\right]\right)^{-1}.
        \end{align}
    Then \eqref{eq:app_var_gamma} is obtained.

    We next prove \eqref{eq:app_var_tau}. By definition of  $\wt{\A}_{p \times p}$ and $\wt{\b}_p$ in \eqref{eq:A_tilde} and \eqref{eq:b_tilde}, we have
        \begin{align} \label{eq:app_var_tau_t0}
        &\frac{1}{\wt{c}}= \left(\frac{1}{c}-\wt{\b}_p\trans \wt{\A}_{p \times p}^{-1} \wt{\b}_p\right)^{-1} 
        =  c-c \wt{\b}_p\trans \left(-\wt{\A}_{p \times p}+c \wt{\b}_p \wt{\b}_p\trans\right)^{-1} \wt{\b}_p  c  \nonumber\\
        =&~  c+ \left( \mE\{\nabla_{\bm\gamma} h(W,A,X;\bm\gamma_0)\u_K(Z,A,X)\trans\}  \A_{K \times K}^{-1} \b_K -  \mE\left\{\frac{(-1)^{1-A}}{f(A\mid W,X)}\nabla_{\bm\gamma} h(W,A,X;\bm\gamma_0)\right\}\right)\trans \nonumber\\
        & \times\Big( \mE\{\nabla_{\bm\gamma} h(W,A,X;\bm\gamma_0)\u_K(Z,A,X)\trans\} \A_{K \times K}^{-1}  \mE\{\u_K(Z,A,X)\nabla_{\bm\gamma} h(W,A,X;\bm\gamma_0)\trans\}\Big)^{-1} \nonumber\\
        & \times\left( \mE\{\nabla_{\bm\gamma} h(W,A,X;\bm\gamma_0)\u_K(Z,A,X)\trans\}  \A_{K \times K}^{-1} \b_K -   \mE\left\{\frac{(-1)^{1-A}}{f(A\mid W,X)}\nabla_{\bm\gamma} h(W,A,X;\bm\gamma_0)\right\}\right).
        \end{align}
    Recall that
    \begin{align*}
    R(Z,A,X)=&~  \frac{ \mE[\{Y-h(W,A,X)\}\{h(W,1,X)-h(W,0,X)-\tau_0\}\mid Z,A,X]}{ \mE[\{Y-h(W,A,X)\}^2 \mid Z,A,X]}. 
    \end{align*}
    In what follows, we show that 
        \begin{align}
        &c = \mE\left[\{\tau_0 - h(W,1,X)+h(W,0,X)\}^2\right] \nonumber\\ 
        &\qquad- \mE\left[ R(Z,A,X)^2 \{Y-h(W,A,X)\}^2\right] +o(1), \label{eq:app_var_tau_t1}\\ 
        &\mE\{\nabla_{\bm\gamma} h(W,A,X;\bm\gamma_0)\u_K(Z,A,X)\trans\}  \A_{K \times K}^{-1} \b_K -  \mE\left\{\frac{(-1)^{1-A}}{f(A\mid W,X)}\nabla_{\bm\gamma} h(W,A,X;\bm\gamma_0)\right\} \nonumber\\
        &~~= -\mE\left[\nabla_{\bm\gamma} h(W,A,X;\bm\gamma_0) \left\{R(Z,A,X)+\frac{(-1)^{1-A}}{f(A\mid W,X)}\right\} \right] + o(1).\label{eq:app_var_tau_t2}
        \end{align}
    {Consider the term \eqref{eq:app_var_tau_t1}.} Recall that
    \beqrs\label{eq:app_var_tau_1}
    c= \mE\left[\{\tau_0 - h(W,1,X)+h(W,0,X)\}^2\right] -\b_K\trans \A_{K \times K}^{-1} \b_K.
    \eeqrs
    By the definition of $\A_{K \times K}$ and $\b_K$ in \eqref{eq:A} and \eqref{eq:b_k}, we have
    \begin{align*}
    c=&~ \mE\left[\{\tau_0 - h(W,1,X)+h(W,0,X)\}^2\right]  \\
    & -   \mE\big[\{Y-h(W,A,X)\}\{\tau_0 - h(W,1,X)+h(W,0,X)\} \u_K(Z,A,X)\trans\big] \\ 
    &\qquad \times~  \mE\left[\{Y-h(W,A,X)\}^2 \u_K(Z,A,X)^{\otimes 2}\right]^{-1}  \\
    &\qquad \times~\mE\big[\{Y-h(W,A,X)\}\{\tau_0 - h(W,1,X)+h(W,0,X)\} \u_K(Z,A,X)\big] \\ 
    =:&~ \mE\left[\{\tau_0 - h(W,1,X)+h(W,0,X)\}^2\right]
    - \mE\{\wt{p}_K(Z,A,X)^2\},
    \end{align*}
    where
        \begin{align*}
            \wt{p}_K(Z,A,X) =&~ \mE\big[\{Y-h(W,A,X)\}\{\tau_0 - h(W,1,X)+h(W,0,X)\} \u_K(Z,A,X)\trans\big] \\ 
            &\quad \times~  \mE\left[\{Y-h(W,A,X)\}^2 \u_K(Z,A,X)^{\otimes 2}\right]^{-1} \\ 
            &\quad \times~  \u_K(Z,A,X)\sqrt{ \mE[\{Y-h(W,A,X)\}^2 \mid Z,A,X]}  \\ 
            =&~ -  \mE\Bigg[R(Z,A,X)\sqrt{ \mE[\{Y-h(W,A,X)\}^2 \mid Z,A,X]}  \\ 
            &\qquad\qquad \times \sqrt{ \mE\left[\{Y-h(W,A,X)\}^2 \mid Z,A,X\right]}\u_K(Z,A,X)\trans\Bigg] \\ 
            &\quad \times~  \mE\left[\left\{\sqrt{ \mE\left[\{Y-h(W,A,X)\}^2 \mid Z,A,X\right]} \u_K(Z,A,X)\right\}^{\otimes 2}\right]^{-1}\\ 
            &\quad\times~ \sqrt{ \mE\left[\{Y-h(W,A,X)\}^2 \mid Z,A,X\right]}\u_K(Z,A,X).
        \end{align*}
    Note that $\wt{p}_K(Z,A,X)$ is the least squared projection of $$R(Z,A,X)\sqrt{\mE[\{Y-h(W,A,X)\}^2 \mid Z,A,X]}$$ onto the space linear spanned by $\Big\{\sqrt{ \mE[\{Y-h(W,A,X)\}^2 \mid Z,A,X]}\u_K(Z,A,X)\Big\}$. By \Cref{ass:sieve}(ii) and \ref{cond:p_smooth}, we have
        \begin{align}
        \label{eq:app_var_tau_1.1}
        \mE\left[\left\|\wt{p}_K(Z,A,X) 
        -R(Z,A,X)\sqrt{\mE[\{Y-h(W,A,X)\}^2 \mid Z,A,X]}
        \right\|^2\right] = o(1).
        \end{align}
    Therefore,
    \begin{align*}
        \mE\{\wt{p}_K(Z,A,X)^2\}
    \rightarrow&~\mE\left[R(Z,A,X)^2\mE[\{Y-h(W,A,X)\}^2\mid Z,A,X] \right]\\ 
    =&~\mE\left[R(Z,A,X)^2 \{Y-h(W,A,X)\}^2 \right].
    \end{align*}
    Then \eqref{eq:app_var_tau_t1} is obtained.
    Next, we consider the term \eqref{eq:app_var_tau_t2}.
    Note that 
    \begin{align*}
        & \mE\{\nabla_{\bm\gamma} h(W,A,X;\bm\gamma_0)\u_K(Z,A,X)\trans\}  \A_{K \times K}^{-1} \b_K \\ 
        =&~    \mE\big[\nabla_{\bm\gamma} h(W,A,X;\bm\gamma_0) \u_K(Z,A,X)\trans\big] \cdot
         \mE\left[\{Y-h(W,A,X)\}^2 \u_K(Z,A,X)^{\otimes 2}\right]^{-1}  \\
        &\quad \times~\mE\big[\{Y-h(W,A,X)\}\{\tau_0 - h(W,1,X)+h(W,0,X)\} \u_K(Z,A,X)\big] \\ 
        =&~  -\left[  \frac{ \nabla_{\bm\gamma} h(W,A,X;\bm\gamma_0)}{\sqrt{ \mE[\{Y-h(W,A,X)\}^2 \mid Z,A,X]}} 
        \sqrt{ \mE[\{Y-h(W,A,X)\}^2 \mid Z,A,X]}
        \u_K(Z,A,X)\trans \right] \\
        &\quad \times~  \mE\left[\left\{\sqrt{ \mE\left[\{Y-h(W,A,X)\}^2 \mid Z,A,X\right]} \u_K(Z,A,X)\right\}^{\otimes 2}\right]^{-1}\\ 
        &\quad \times~  \mE\Bigg[R(Z,A,X)\sqrt{ \mE[\{Y-h(W,A,X)\}^2 \mid Z,A,X]}  \\ 
        &\qquad\qquad \times \sqrt{ \mE\left[\{Y-h(W,A,X)\}^2 \mid Z,A,X\right]}\u_K(Z,A,X) \Bigg] \\ 
        =&~ -\mE\left[  \frac{ \nabla_{\bm\gamma} h(W,A,X;\bm\gamma_0)}{\sqrt{ \mE[\{Y-h(W,A,X)\}^2 \mid Z,A,X]}} \wt{p}_K(Z,A,X) \right]
    \end{align*}
    Therefore, \eqref{eq:app_var_tau_1.1} implies that 
    \begin{align*}
        &\mE\{\nabla_{\bm\gamma} h(W,A,X;\bm\gamma_0)\u_K(Z,A,X)\trans\}  \A_{K \times K}^{-1} \b_K 
        =-\mE\left[\nabla_{\bm\gamma} h(W,A,X;\bm\gamma_0)R(Z,A,X) \right]+o(1).
    \end{align*}
    Then we obtain \eqref{eq:app_var_tau_t2}.
    Combining \eqref{eq:app_var_tau_t0}--\eqref{eq:app_var_tau_t2}, and recalling that
    \begin{align*}
        \bm\kappa = \V_{\bm\gamma}\mE\left[\left\{R(Z,A,X)+\frac{(-1)^{1-A}}{f(A\mid W,X)}\right\}\nabla_{\bm\gamma} h(W,A,X;\bm\gamma_0)\right],
    \end{align*}
    and
        \begin{align*}
            \left[\mE\{\nabla_{\bm\gamma} h(W,A,X;\bm\gamma_0)\u_K(Z,A,X)\trans\} \A_{K \times K}^{-1}   \mE\{\u_K(Z,A,X)\nabla_{\bm\gamma} h(W,A,X;\bm\gamma_0)\trans\}\right]^{-1}
            = \V_\gamma + o(1),
        \end{align*}
    from \eqref{eq:app_var_gamma_3}, we have
    \begin{align}\label{eq:app_var_tau_2}
    \frac{1}{\wt{c}} =&~  \mE\left[\{\tau_0 - h(W,1,X)+h(W,0,X)\}^2\right] \nonumber\\
    &~ - \mE\left[R(Z,A,X)^2 \{Y-h(W,A,X)\}^2\right] +\bm\kappa\trans\V_\gamma^{-1}\bm\kappa
    + o(1).
    \end{align}
    To complete the proof of \eqref{eq:app_var_tau}, we need to show $\frac{1}{\wt{c}}=\mE\{\psi_2(\O)^2\}+o(1)$. Recall that
    \begin{align*}
        & \psi_2(\O) = h(W,1,X)-h(W,0,X)-\tau_0+t(Z,A,X)\{Y-h(W,A,X)\},\\
       \text{and}~~ & t(Z,A,X) = \bm\kappa\trans\m_\eff(Z,A,X) - R(Z,A,X).
    \end{align*}
    We have
    \begin{align}\label{eq:app_var_tau_exp}
        \mE\{\psi_2(\O)^2\} =&~ \mE\left(\left[h(W,1,X)-h(W,0,X)-\tau_0+t(Z,A,X)\{Y-h(W,A,X)\}\right]^2\right) \nonumber\\ 
        =&~ \mE\left[ \{h(W,1,X)-h(W,0,X)-\tau_0\}   ^2\right] \nonumber\\ 
        &~ + \mE\left[ t(Z,A,X)^2\{Y-h(W,A,X)\}^2\right]  \nonumber\\ 
        &~ + 2\mE\left[t(Z,A,X) \{h(W,1,X)-h(W,0,X)-\tau_0\}\{Y-h(W,A,X)\}\right] \nonumber\\
        =&~  \mE\left[ \{h(W,1,X)-h(W,0,X)-\tau_0\}   ^2\right] \nonumber\\ 
        &~ + \bm\kappa\trans \mE\left[ \m_\eff(Z,A,X)\m_\eff(Z,A,X)\trans\{Y-h(W,A,X)\}^2\right] \bm\kappa \nonumber\\
        &~ + \mE\left[ R(Z,A,X)^2 \{Y-h(W,A,X)\}^2\right] \nonumber\\ 
        &~ - 2\bm\kappa\trans \mE\left[ \m_\eff(Z,A,X)R(Z,A,X)\{Y-h(W,A,X)\}^2\right] \nonumber\\
        &~ +  2\bm\kappa\trans \mE\left[ \m_\eff(Z,A,X)\{h(W,1,X)-h(W,0,X)-\tau_0\}\{Y-h(W,A,X)\}\right] \nonumber\\
        &~ -  2\mE\left[ R(Z,A,X)\{h(W,1,X)-h(W,0,X)-\tau_0\}\{Y-h(W,A,X)\}\right] \nonumber\\ 
        =&~ \mE\left[\{h(W,1,X)-h(W,0,X)-\tau_0\}^2\right] + \bm\kappa\trans\V_\gamma^{-1}\bm\kappa  \nonumber\\
        &~ - \mE\left[R(Z,A,X)^2 \{Y-h(W,A,X)\}^2 \right] ,
    \end{align}
    where the last equality follows from \eqref{eq:app_var_gamma_verify} and the facts 
    \begin{align*}
       & \mE\left[ \m_\eff(Z,A,X)\{h(W,1,X)-h(W,0,X)-\tau_0\}\{Y-h(W,A,X)\}\right] \\ 
       =&~ \mE\left[ \m_\eff(Z,A,X)R(Z,A,X)\{Y-h(W,A,X)\}^2\right],
    \end{align*}
    and
    \begin{align*}
        & \mE\left[ R(Z,A,X)\{h(W,1,X)-h(W,0,X)-\tau_0\}\{Y-h(W,A,X)\}\right] \\ 
        =&~ \mE\left[ R(Z,A,X)^2 \{Y-h(W,A,X)\}^2\right].
    \end{align*}
    Then \eqref{eq:app_var_tau} is obtained. This completes the proof.
\end{proof}

\subsection{Proof of Theorem~\ref{coro:super_eff}} 
\label{ssec:pf_coro_super_eff}
\begin{proof}
    We first prove (i).
    Recall that the treatment-confounding bridge function $q(z, a, x)$ satisfies
    \begin{align*}
        \mE\{q(Z, A, X) \mid W, A, X\}=\frac{1}{f(A \mid W, X)}.
\end{align*} 
We have 
\begin{align*}
    \bm\kappa =&~  \V_{\bm\gamma}\mE\left[\left\{R(Z,A,X)+\frac{(-1)^{1-A}}{f(A\mid W,X)}\right\}\nabla_{\bm\gamma} h(W,A,X;\bm\gamma_0)\right]\\
    =&~ \V_{\bm\gamma}\mE\left[\left\{R(Z,A,X)+(-1)^{1-A}q(Z,A,X)\right\}\mE\{\nabla_{\bm\gamma} h(W,A,X;\bm\gamma_0)\mid Z,A,X\}\right].
\end{align*}
Then 
\begin{align*}
 \bm\kappa\trans\V_{\bm\gamma}^{-1}\bm\kappa
 =&~ \mE\left[\left\{R(Z,A,X)+(-1)^{1-A}q(Z,A,X)\right\}\mE\{\nabla_{\bm\gamma} h(W,A,X;\bm\gamma_0)\mid Z,A,X\}\right]\trans\\ 
 &\quad \times  \left(\mE\left[\frac{\bigbc{ \mE[\nabla_{\bm\gamma} h(W,A,X;\bm\gamma_0)\mid Z,A,X]}^{\otimes2}}{ \mE[\{Y-h(W,A,X)\}^2 \mid Z,A,X]}\right]\right)^{-1}\\ 
 &\quad \times \mE\left[\left\{R(Z,A,X)+(-1)^{1-A}q(Z,A,X)\right\}\mE\{\nabla_{\bm\gamma} h(W,A,X;\bm\gamma_0)\mid Z,A,X\}\right]\\ 
 =:&~ \mE\{s(Z,A,X)^2\},
\end{align*}
where 
\begin{align*}
    s(Z,A,X) =&~ 
    \mE\Bigg[\left\{R(Z,A,X)+(-1)^{1-A}q(Z,A,X)\right\} \sqrt{\mE[\{Y-h(W,A,X)\}^2 \mid Z,A,X]}\\ 
    &\qquad\quad \times \frac{\mE\{\nabla_{\bm\gamma} h(W,A,X;\bm\gamma_0)\mid Z,A,X\}}{\sqrt{\mE[\{Y-h(W,A,X)\}^2 \mid Z,A,X]}}\Bigg]\trans\\ 
    &\quad \times  \left(\mE\left[\bigbc{\frac{ \mE[\nabla_{\bm\gamma} h(W,A,X;\bm\gamma_0)\mid Z,A,X]}{\sqrt{ \mE[\{Y-h(W,A,X)\}^2 \mid Z,A,X]}}}^{\otimes2} \right]\right)^{-1} \\ 
     &\quad \times \frac{\mE\{\nabla_{\bm\gamma} h(W,A,X;\bm\gamma_0)\mid Z,A,X\}}{\sqrt{\mE[\{Y-h(W,A,X)\}^2 \mid Z,A,X]}}.
\end{align*}
Note that $s(Z,A,X)$ is the least squared projection of $$\left\{R(Z,A,X)+(-1)^{1-A}q(Z,A,X)\right\} \sqrt{\mE[\{Y-h(W,A,X)\}^2 \mid Z,A,X]}$$ onto the space linear spanned by $\Big\{\dfrac{\mE\{\nabla_{\bm\gamma} h(W,A,X;\bm\gamma_0)\mid Z,A,X\}}{\sqrt{\mE[\{Y-h(W,A,X)\}^2 \mid Z,A,X]}}\Big\}$. By the  Hilbert projection theorem, we have 
\begin{align}\label{eq:app_var_tau_ineq}
    \mE\{s(Z,A,X)^2\}\le 
    \mE \left[\left\{R(Z,A,X)+(-1)^{1-A}q(Z,A,X)\right\}^2 \{Y-h(W,A,X)\}^2 \right].
\end{align}
As a result, by \eqref{eq:app_var_tau_exp}, we have 
\begin{align*}
    V_\tau =&~ \mE\left[\{h(W,1,X)-h(W,0,X)-\tau_0\}^2\right] + \bm\kappa\trans\V_\gamma^{-1}\bm\kappa  \nonumber\\
    &~ - \mE\left[R(Z,A,X)^2 \{Y-h(W,A,X)\}^2\right]\\ 
    \le&~ \mE\left[\{h(W,1,X)-h(W,0,X)-\tau_0\}^2\right]  \nonumber\\
    &~ + \mE \left[\left\{R(Z,A,X)+(-1)^{1-A}q(Z,A,X)\right\}^2 \{Y-h(W,A,X)\}^2 \right]\\
    &~ - \mE\left[R(Z,A,X)^2 \{Y-h(W,A,X)\}^2\right]\\ 
    =&~  \mE\left[\{h(W,1,X)-h(W,0,X)-\tau_0\}^2\right]  \nonumber\\
    &~ + \mE\left[q(Z,A,X)^2 \{Y-h(W,A,X)\}^2\right]\\ 
    &~ +  2 \mE\left[(-1)^{1-A}q(Z,A,X)R(Z,A,X) \{Y-h(W,A,X)\}^2\right]\\ 
    =&~ \mE\left[\{h(W,1,X)-h(W,0,X)-\tau_0\}^2\right]  \nonumber\\
    &~ + \mE\left[q(Z,A,X)^2 \{Y-h(W,A,X)\}^2\right]\\ 
    &~ +  2 \mE\left[(-1)^{1-A}q(Z,A,X)\{Y-h(W,A,X)\}\{h(W,1,X)-h(W,0,X)-\tau_0\}\right]\\
    =&~ \mE\left\{\psi_{\eff}(\O)^2\right\}= V_{\tau,\eff},
\end{align*}
where the third equality follows from the definition of $R(Z,A,X)$. 
Additionally, the equality $V_{\tau}= V_{\tau,\eff}$ holds if and only if the equality in \eqref{eq:app_var_tau_ineq} holds, which happens if and only if 
$$\left\{R(Z,A,X)+(-1)^{1-A}q(Z,A,X)\right\} \sqrt{\mE[\{Y-h(W,A,X)\}^2 \mid Z,A,X]}$$ lies within the space linear spanned by $\left\{\dfrac{\mE\{\nabla_{\bm\gamma} h(W,A,X;\bm\gamma_0)\mid Z,A,X\}}{\sqrt{\mE[\{Y-h(W,A,X)\}^2 \mid Z,A,X]}}\right\}$, equivalently, there is a vector of constants $\bm\alpha$ such that 
\begin{align*}
    &\left\{R(Z,A,X)+(-1)^{1-A}q(Z,A,X)\right\} \sqrt{\mE[\{Y-h(W,A,X)\}^2 \mid Z,A,X]} \\ 
    &\qquad =  \frac{\bm\alpha\trans\mE\{\nabla_{\bm\gamma} h(W,A,X;\bm\gamma_0)\mid Z,A,X\}}{\sqrt{\mE[\{Y-h(W,A,X)\}^2 \mid Z,A,X]}} ,
\end{align*}
namely,
\begin{align*}
    R(Z,A,X)+(-1)^{1-A}q(Z,A,X) =&~   \frac{\bm\alpha\trans\mE\{\nabla_{\bm\gamma} h(W,A,X;\bm\gamma_0)\mid Z,A,X\}}{ \mE[\{Y-h(W,A,X)\}^2 \mid Z,A,X]} \\
     =&~  \bm\alpha\trans\m_\eff(Z,A,X).
\end{align*}
This completes the proof of (i).

Next, we give a proof of (ii). Recall that under $\calM_{sub}$, the model is 
\begin{align*} 
    \mE\left\{Y-h(W,A,X;\bm\gamma_0)\mid Z,A,X\right\}=0,~~
    \tau_0 = \mE\left\{h(W,1,X;\bm\gamma_0)-h(W,0,X;\bm\gamma_0)\right\},
\end{align*}
which is a special case of the sequential model studied in \cite{AiChen2012}. So we use result of Remark 2.2 of \cite{AiChen2012} (also Theorem 1 of \cite{BrownNewey1998}) to derive the semiparametric efficiency bound of $\tau_0$ under $\calM_{sub}$. Let 
\begin{align*}
    \varepsilon(\O;\bm\gamma,\tau) =&~  h(W,1,X;\bm\gamma)-h(W,0,X;\bm\gamma)-\tau-R(Z,A,X)\{Y-h(W,A,X;\gamma)\}.
\end{align*}
Then Remark 2.2 of \cite{AiChen2012} showed that the the semiparametric efficiency bound of $\tau_0$ under $\calM_{sub}$ is 
\begin{align*}
& E\left\{\varepsilon(\O;\bm\gamma_0,\tau_0)^2\right\} + \frac{\partial \{\varepsilon(\O;\bm\gamma_0,\tau_0)\}}{\partial\bm\gamma\trans}\V_{\bm\gamma}\frac{\partial \{\varepsilon(\O;\bm\gamma_0,\tau_0)\}}{\partial\bm\gamma}\\
=&~ \mE\left[ \{h(W,1,X;\bm\gamma_0)-h(W,0,X;\bm\gamma_0)-\tau_0\}   ^2\right] \nonumber\\ 
&~ + \mE\left[ R(Z,A,X)^2\{Y-h(W,A,X;\bm\gamma_0)\}^2\right]  \nonumber\\ 
&~ - 2\mE\left[R(Z,A,X) \{h(W,1,X;\bm\gamma_0)-h(W,0,X;\bm\gamma_0)-\tau_0\}\{Y-h(W,A,X;\bm\gamma_0)\}\right] \nonumber\\
&~ + \bm\kappa\trans\V_{\bm\gamma}^{-1}\bm\kappa \nonumber\\ 
=&~ \mE\left[\{h(W,1,X;\bm\gamma_0)-h(W,0,X;\bm\gamma_0)-\tau_0\}^2\right] 
 - \mE\left[R(Z,A,X)^2 \{Y-h(W,A,X;\bm\gamma_0)\}^2 \right] \nonumber\\
&~ + \bm\kappa\trans\V_\gamma^{-1}\bm\kappa \nonumber\\ 
=&~ \mE\{\psi_2(\O)^2\} \equiv V_{\tau},
\end{align*}
where first equality is from the definition of $\bm\kappa$ in \Cref{thm:asym_normal}, and the third equality is from \eqref{eq:app_var_tau_exp}. This completes the proof of the corollary.
\end{proof}

\subsection{Proof of Theorem~\ref{coro:super_plug_in}}
\begin{proof}
    Recall that 
    \begin{align*}
        \wh{\tau}^\plugin(\m_\eff) = \frac{1}{N}\sum_{i=1}^{N} \left\{h(W_i,1,X_i;\wh{\bm\gamma}(\m_\eff))-h(W_i,0,X_i;\wh{\bm\gamma}(\m_\eff))\right\}.
    \end{align*}
    Following uniform weak law of large number \citep{NeweyMcFadden1994} under the regularity conditions, we have 
    \begin{align*}
        & \sqrt{N}\{\wh{\tau}^\plugin(\m_\eff)-\tau_0\}\\
        =&~ \frac{1}{\sqrt{N}}\sum_{i=1}^{N} \left\{h(W_i,1,X_i;\bm\gamma_0)-h(W_i,0,X_i;\bm\gamma_0)-\tau_0\right\}\\ 
        &~ + \mE \left\{\frac{(-1)^{1-A}}{f(A\mid W,X)} \nabla_{\bm\gamma} h(W,A,X;\bm\gamma_0)\right\}\trans \sqrt{N}\left\{\wh{\bm\gamma}(\m_\eff)-\bm\gamma_0\right\}  + o_p(1) \\ 
        =&~ \frac{1}{\sqrt{N}}\sum_{i=1}^{N} \Big[h(W_i,1,X_i;\bm\gamma_0)-h(W_i,0,X_i;\bm\gamma_0)-\tau_0 - R(Z_i,A_i,X_i)\{Y_i-h(W_i,A_i,X_i)\}\Big]\\ 
        &~ + \frac{1}{\sqrt{N}}\sum_{i=1}^{N} \Bigg[\mE \left\{\frac{(-1)^{1-A}}{f(A\mid W,X)} \nabla_{\bm\gamma} h(W,A,X;\bm\gamma_0)\right\}\trans \V_{\bm\gamma}\m_\eff(Z_i,A_i,X_i)\\ 
        &\qquad\qquad\qquad\qquad\qquad\qquad+  R(Z_i,A_i,X_i)\Bigg]\{Y_i-h(W_i,A_i,X_i)\} + o_p(1)\\ 
        =:&~ T_1+T_2+ o_p(1). 
    \end{align*}
    Note that 
    \begin{align*}
        &\mE\big[\{h(W,1,X;\bm\gamma_0)-h(W,0,X;\bm\gamma_0)-\tau_0\}\{Y-h(W,A,X)\}\big]\\ 
        &\quad= \mE\big[ R(Z,A,X)\{Y-h(W,A,X)\}^2\big].
    \end{align*}
    We have 
    \begin{align*}
        V^\plugin_{\tau,\eff} = \mE\{(T_1+T_2)^2\} = \mE(T_1^2)+\mE(T_2^2).
    \end{align*}
    Note that 
    \begin{align*}
        \mE(T_1^2)=&~  \mE\left[ \{h(W,1,X)-h(W,0,X)-\tau_0\}   ^2\right] \nonumber\\ 
        &~ + \mE\left[ R(Z,A,X)^2 \{Y-h(W,A,X)\}^2\right] \nonumber\\ 
        &~ -  2\mE\left[ R(Z,A,X)\{h(W,1,X)-h(W,0,X)-\tau_0\}\{Y-h(W,A,X)\}\right] \nonumber\\ 
        =&~ \mE\left[ \{h(W,1,X)-h(W,0,X)-\tau_0\}   ^2\right]  - \mE\left[ R(Z,A,X)^2 \{Y-h(W,A,X)\}^2\right], 
    \end{align*}
    and 
    \begin{align}
        \mE(T_2^2)=&~
        \mE \left\{\frac{(-1)^{1-A}}{f(A\mid W,X)} \nabla_{\bm\gamma} h(W,A,X;\bm\gamma_0)\right\}\trans \V_{\bm\gamma}
        \mE\left[ \m_\eff(Z,A,X)^{\otimes2} \{Y-h(W,A,X)\}^2\right]\nonumber\\ 
        &\qquad\qquad  \times \V_{\bm\gamma} \mE \left\{\frac{(-1)^{1-A}}{f(A\mid W,X)} \nabla_{\bm\gamma} h(W,A,X;\bm\gamma_0)\right\} \label{eq:app_pf_coro_plug_in_T2.1}\\ 
        &~ + \mE\left[ R(Z,A,X)^2 \{Y-h(W,A,X)\}^2\right] \label{eq:app_pf_coro_plug_in_T2.2}\\ 
        &~ + 2\mE \left\{\frac{(-1)^{1-A}}{f(A\mid W,X)} \nabla_{\bm\gamma} h(W,A,X;\bm\gamma_0)\right\}\trans \V_{\bm\gamma}\nonumber\\
        &\qquad\qquad \times \mE\left[ \m_\eff(Z,A,X)R(Z,A,X)  \{Y-h(W,A,X)\}^2\right].\label{eq:app_pf_coro_plug_in_T2.3}
    \end{align}
    By \eqref{eq:app_var_gamma_verify}, we have 
    \begin{align*}
        \eqref{eq:app_pf_coro_plug_in_T2.1}
        =
        \mE \left\{\frac{(-1)^{1-A}}{f(A\mid W,X)} \nabla_{\bm\gamma} h(W,A,X;\bm\gamma_0)\right\}\trans \V_{\bm\gamma} \mE \left\{\frac{(-1)^{1-A}}{f(A\mid W,X)} \nabla_{\bm\gamma} h(W,A,X;\bm\gamma_0)\right\}.
    \end{align*}
    Similar to \eqref{eq:app_var_tau_ineq}, we have 
    \begin{align}\label{eq:app_pf_coro_plug_in_2}
        \eqref{eq:app_pf_coro_plug_in_T2.2}
        \ge \mE \left\{R(Z,A,X) \nabla_{\bm\gamma} h(W,A,X;\bm\gamma_0)\right\}\trans \V_{\bm\gamma} \mE \left\{R(Z,A,X) \nabla_{\bm\gamma} h(W,A,X;\bm\gamma_0)\right\}.
    \end{align}
     By the definition of $\m_\eff$, we have
    \begin{align*}
        \eqref{eq:app_pf_coro_plug_in_T2.3}=
        2\mE \left\{\frac{(-1)^{1-A}}{f(A\mid W,X)} \nabla_{\bm\gamma} h(W,A,X;\bm\gamma_0)\right\}\trans \V_{\bm\gamma} \mE \left\{R(Z,A,X) \nabla_{\bm\gamma} h(W,A,X;\bm\gamma_0)\right\}.
    \end{align*}
    Combining \eqref{eq:app_pf_coro_plug_in_T2.1}--\eqref{eq:app_pf_coro_plug_in_T2.3} yields that 
    \begin{align*}
        \mE(T_2^2)&\ge~ \mE\left[\left\{R(Z,A,X)+\frac{(-1)^{1-A}}{f(A\mid W,X)}\right\}\nabla_{\bm\gamma} h(W,A,X;\bm\gamma_0)\right]\trans \\ 
        &\qquad\qquad\times \V_{\bm\gamma} \mE\left[\left\{R(Z,A,X)+\frac{(-1)^{1-A}}{f(A\mid W,X)}\right\}\nabla_{\bm\gamma} h(W,A,X;\bm\gamma_0)\right]\\ 
        =&~ \bm\kappa\trans\V_\gamma^{-1}\bm\kappa.
    \end{align*}
    Then by \eqref{eq:app_var_tau_exp}, we have 
    \begin{align*}
        V^\plugin_{\tau,\eff}
        \ge&~  \mE\left[\{h(W,1,X)-h(W,0,X)-\tau_0\}^2\right] + \bm\kappa\trans\V_\gamma^{-1}\bm\kappa  \nonumber\\
        &~ - \mE\left[R(Z,A,X)^2 \{Y-h(W,A,X)\}^2 \right]\\ 
        =&~ V_\tau.
    \end{align*}
    The equality holds if and only if the equality in \eqref{eq:app_pf_coro_plug_in_2} holds, which happens if and only if 
     $$R(Z,A,X) \sqrt{\mE[\{Y-h(W,A,X)\}^2 \mid Z,A,X]}$$ lies within the space linear spanned by $\left\{\dfrac{\mE\{\nabla_{\bm\gamma} h(W,A,X;\bm\gamma_0)\mid Z,A,X\}}{\sqrt{\mE[\{Y-h(W,A,X)\}^2 \mid Z,A,X]}}\right\}$, equivalently, there is a vector of constants $\bm\alpha$ such that 
    \begin{align*}
        & R(Z,A,X)  \sqrt{\mE[\{Y-h(W,A,X)\}^2 \mid Z,A,X]} \\ 
        &\qquad =  \frac{\bm\alpha\trans\mE\{\nabla_{\bm\gamma} h(W,A,X;\bm\gamma_0)\mid Z,A,X\}}{\sqrt{\mE[\{Y-h(W,A,X)\}^2 \mid Z,A,X]}} ,
    \end{align*}
    namely,
    \begin{align*}
        R(Z,A,X)  =&~   \frac{\bm\alpha\trans\mE\{\nabla_{\bm\gamma} h(W,A,X;\bm\gamma_0)\mid Z,A,X\}}{ \mE[\{Y-h(W,A,X)\}^2 \mid Z,A,X]} \\
         =&~  \bm\alpha\trans\m_\eff(Z,A,X).
    \end{align*}
    This completes the proof.
\end{proof}

\section{Supplements to simulation studies}\label{sec:app_supp_simu}

\subsection{Parameterization of bridge functions}
In this section, we show that the data generating mechanism in our simulation study is compatible with the following models of the confounding bridge functions $h$ and $q$:
\begin{align*}
    & h(w,a,x;\bm\gamma) = \gamma_0 +  \gamma_1w + \gamma_2a + \gamma_3x,\\
    & q(z,a,x;\bm\theta) =  1+\exp\left\{(-1)^{a}(\theta_0+\theta_1z+\theta_2a+\theta_3x)\right\}.
\end{align*}
Note that
\begin{align*}
    & \mE(W\mid Z,A,X)= \beta_{w0}+\beta_{wx}X+\beta_{wu}\mE(U\mid Z,A,X), \\
    & \mE(Y\mid Z,A,X)= \beta_{y0}+\beta_{yw} \mE(W\mid Z,A,X)+\beta_{ya}A+\beta_{yx}X+\beta_{yu}\mE(U\mid Z,A,X).
\end{align*}
It follows that 
\begin{align*}
    \mE(Y\mid Z,A,X) = \gamma_0^*+\gamma_1^* \mE(W\mid Z,A,X)+\gamma_2^* A+\gamma_3^* X,
\end{align*}
where 
\begin{align*}
    \gamma_0^* = \beta_{y0}-\beta_{w0}\beta_{yu}/\beta_{wu},~~ 
    \gamma_1^* = \beta_{yw}+\beta_{yu}/\beta_{wu},~~
    \gamma_2^* = \beta_{ya},~~
    \gamma_3^* = \beta_{yx}-\beta_{wx}\beta_{yu}/\beta_{wu}.
\end{align*}
Therefore, $h(w,a,x;\bm\gamma^*)=\gamma_0^*+\gamma_1^*w+\gamma_2^* a+\gamma_3^* x$ satisfies the integral equation \eqref{eq:iden_bridge_h}.

Next we show that there is a function
\begin{align*}
q(z,a,x;\bm\theta^*) = 1+\exp\left\{(-1)^{a}(\theta_0^*+\theta_1^*z+\theta_2^*a+\theta_3^*x)\right\},
\end{align*}
which satisfies the integral equation \eqref{eq:iden_bridge_q}. Note that we only need to show 
\begin{align}\label{eq:app_supp_simu_1}
\int_\calZ q(z,A,X;\bm\theta^*) f(z\mid U,A,X) dz = \frac{1}{f(A\mid U,X)},
\end{align}
holds. Then we have 
\begin{align*}
    \frac{1}{f(A\mid W,X)} =&~  \int_{\mathcal{U}} \frac{1}{f(A\mid u,W,X)} f(u\mid W,A,X) du \\
    =&~ \int_{\mathcal{U}} \frac{1}{f(A\mid u,X)} f(u\mid W,A,X) du \\
    =&~ \int_{\mathcal{U}} \left\{\int_\calZ q(z,A,X;\bm\theta^*) f(z\mid u,A,X) dz\right\} f(u\mid W,A,X) du \\
    =&~ \int_\calZ q(z,A,X;\bm\theta^*) \left\{\int_{\mathcal{U}} f(z\mid u,A,X) f(u\mid W,A,X) du \right\} dz \\
    =&~ \int_\calZ q(z,A,X;\bm\theta^*) \left\{\int_{\mathcal{U}} f(z\mid u,W,A,X) f(u\mid W,A,X) du \right\} dz \\
    =&~ \int_\calZ q(z,A,X;\bm\theta^*) f(z\mid W,A,X) dz \\
    =&~ \mE\{q(Z,A,X;\bm\theta^*)\mid Z,A,X\},
\end{align*}
where the second and fifth equalities are by \Cref{ass:causal_and_proxy}(iii), and the third equality is by \eqref{eq:app_supp_simu_1}. We next show \eqref{eq:app_supp_simu_1}. Recall that by the DGP,
\begin{align}\label{eq:app_supp_simu_2}
    \frac{1}{f(A\mid U,X)}= 1+\exp\left\{(-1)^{A}(\beta_{a0}+\beta_{ax}X +\beta_{au}U)\right\}.
\end{align}
We have 
\begin{align}\label{eq:app_supp_simu_3}
    & \int_\calZ q(z,A,X;\bm\theta^*) f(z\mid U,A,X) dz  \nonumber\\ 
    =&~  1+\exp\left\{(-1)^{A}(\theta_0^*+\theta_2^*A+\theta_3^*X)\right\} \int_\calZ \exp\left\{(-1)^{A}\theta_1^*z\right\} f(z\mid U,A,X) dz \nonumber\\ 
    =&~ 1+\exp\left\{(-1)^{A}(\theta_0^*+\theta_2^*A+\theta_3^*X)\right\} \nonumber\\ 
    &~ \times \int_\calZ \exp\left\{(-1)^{A}\theta_1^*z\right\} 
    \frac{1}{\sqrt{2\pi}}\exp\left\{-\frac{1}{2}(z-\beta_{z0}-\beta_{zu}U-\beta_{za}A-\beta_{zx}X)^2\right\} dz \nonumber\\ 
    =&~ 1+\exp\left\{(-1)^{A}(\theta_0^*+\theta_2^*A+\theta_3^*X) + (-1)^{A}\theta_1^*(\beta_{z0}+\beta_{zu}U+\beta_{za}A+\beta_{zx}X) + \frac{1}{2}\theta_1^{*2}\right\}.
\end{align}
Comparing \eqref{eq:app_supp_simu_2} and \eqref{eq:app_supp_simu_3}, we obtain that \eqref{eq:app_supp_simu_1} holds if 
\begin{align*}
    -(\theta_0^*+\theta_2^*)-\theta_1^*(\beta_{z0}+\beta_{za})+\frac{1}{2}\theta_1^{*2}=-\beta_{a0} ,\\ 
    \theta_0^*+\theta_1^*\beta_{z0}+\frac{1}{2}\theta_1^{*2}=\beta_{a0} ,\\ 
    \theta_3^*+\theta_1^*\beta_{zx} = \beta_{ax},\\ 
    \theta_1^*\beta_{zu} = \beta_{au}.
\end{align*}
Equivalently,
\begin{align*}
    \theta_0^* =&~ \beta_{a0} - 0.5\beta_{au}^2/\beta_{zu}^2 -\beta_{z0}\beta_{au}/\beta_{zu},\\
    \theta_1^* =&~  \beta_{au}/\beta_{zu} ,\\ 
    \theta_2^* =&~ \beta_{au}^2/\beta_{zu}^2-\beta_{za}\beta_{au}/\beta_{zu},\\
    \theta_3^* =&~ \beta_{ax} - \beta_{zx}\beta_{au}/\beta_{zu}.
\end{align*}
It implies that $q(z,a,x;\bm\theta^*)$ satisfies the integral equation \eqref{eq:iden_bridge_q}.

\subsection{Simulation under model misspecification}\label{ssec:additional_simu}

In this section, we construct simulations to evaluate the performance of the proposed method and the proximal DR estimator under a misspecified bridge function $h$. Following \cite{KangSchafer2007}, we simulate the case of misspecified $h$ by applying a transformation to the observed variables. Specifically, we consider three levels of model misspecification:
\begin{align*}
&\text{minor misspecification:} ~~ W_1^*=W + 0.1 W^2;\\ 
&\text{moderate misspecification:} ~~ W_2^*=W + 0.5 W^2;\\ 
&\text{and significant misspecification:} ~~ W_3^*=|W|^{1/2}+1.
\end{align*}
These transformed variables are used instead of $W$ to estimate $h$. We generate 400 and 800 samples under scenario II in \Cref{ssec:simulation} to implement the proposed GMM estimator and proximal DR estimator, respectively. \Cref{table:simu_mis} summarizes the bias, standard deviation (SD), average length of 95\% confidence intervals, coverage probability, and power across various scenarios.
As expected, the proposed method aims for high estimation efficiency but is exposed to the risk of model misspecification. When misspecification is mild, the proposed methods still perform acceptably, but bias increases and coverage worsens significantly as misspecification becomes more severe. In contrast, the proximal DR estimator consistently guarantees a slight bias. However, its variance and coverage performance also deteriorate as one of the bridge functions becomes more severely misspecified.

\begin{table}[h!] 
    \centering
    \caption{Summary of absolute bias, standard error (SE), root mean squared error (RMSE), coverage probability (CP), average length of the 95\% confidence interval, and power under different levels of model misspecification.}
    \label{table:simu_mis}
    \par
    \resizebox{\linewidth}{!}{
    \begin{tabular}{cccccccccccccc} \toprule 
        & \multicolumn{6}{c}{$n=400$}
        && \multicolumn{6}{c}{$n=800$} \\  \cline{2-7}\cline{9-14}
        Method & Bias & SE & RMSE & Length & CP & Power && Bias & SE & RMSE & Length & CP & Power\\ \midrule
        \multicolumn{12}{l}{\textbf{(a) Correct specification}} \\
        GMM & 0.01 & 0.13 & 0.13 & 0.48 & 0.934 & 0.970 &  & 0.01 & 0.09 & 0.09 & 0.34 & 0.956 & 1.000\\
        PDR & 0.03 & 0.30 & 0.30 & 1.07 & 0.946 & 0.508 &  & 0.02 & 0.20 & 0.20 & 0.74 & 0.962 & 0.744\\
        \multicolumn{12}{l}{\textbf{(b) Minor misspecification}} \\
        GMM & 0.04 & 0.18 & 0.19 & 0.61 & 0.876 & 0.886 &  & 0.03 & 0.12 & 0.13 & 0.45 & 0.938 & 0.970\\
        PDR & 0.03 & 0.33 & 0.33 & 1.17 & 0.956 & 0.472 &  & 0.02 & 0.21 & 0.21 & 0.82 & 0.954 & 0.638\\
        \multicolumn{12}{l}{\textbf{(c) Moderate misspecification}} \\
        GMM & 0.28 & 0.28 & 0.40 & 0.74 & 0.598 & 0.928 &  & 0.25 & 0.18 & 0.31 & 0.55 & 0.516 & 0.986\\
        PDR & 0.03 & 0.56 & 0.56 & 1.89 & 0.960 & 0.268 &  & 0.05 & 0.35 & 0.35 & 1.35 & 0.964 & 0.346\\
        \multicolumn{12}{l}{\textbf{(d) Significant misspecification}} \\
        GMM & 0.76 & 0.39 & 0.86 & 0.78 & 0.166 & 0.980 &  & 0.69 & 0.28 & 0.74 & 0.56 & 0.068 & 1.000\\
        PDR & 0.06 & 1.27 & 1.27 & 1.25 & 0.676 & 0.272 &  & 0.01 & 0.71 & 0.71 & 1.77 & 0.724 & 0.184\\
        \bottomrule
    \end{tabular}
    }
  \end{table}

\subsection{Simulation with B-splines}\label{ssec:supp_simu_spline}

To illustrate the performance of the proposed method using different basis functions, we conducted additional simulations under Scenario II using B-splines. \Cref{table:simu_b_spline1} compares the proposed method implemented with power series versus B-splines. The results indicate that the difference in finite-sample performance is negligible.

\begin{table}[h!] 
    \centering
    \caption{Simulation results: absolute bias, standard error (SE), root mean squared error (RMSE), coverage probability (CP), average length of the 95\% confidence interval, and power.}
    \label{table:simu_b_spline1}
    \par
    \resizebox{\linewidth}{!}{
    \begin{tabular}{cccccccccccccc} \toprule 
        & \multicolumn{6}{c}{$n=400$}
        && \multicolumn{6}{c}{$n=800$} \\  \cline{2-7}\cline{9-14}
        Sieve & Bias & SE & RMSE & Length & CP & Power && Bias & SE & RMSE & Length & CP & Power\\ \midrule
        Power series & 0.01 & 0.13 & 0.13 & 0.48 & 0.936 & 0.966 && 0.00 & 0.09 & 0.09 & 0.34 & 0.956 & 1.000\\
        B-splines    & 0.01 & 0.14 & 0.14 & 0.47 & 0.906 & 0.960 && 0.00 & 0.09 & 0.09 & 0.34 & 0.940 & 0.994\\
        \bottomrule
    \end{tabular}
    }
\end{table}

\newpage

\end{document}